\documentclass[conference]{IEEEtran}
\IEEEoverridecommandlockouts

\usepackage{cite}
\usepackage{amsmath,amssymb,amsfonts}
\usepackage{graphicx}
\usepackage{textcomp}
\usepackage{xcolor}
\def\BibTeX{{\rm B\kern-.05em{\sc i\kern-.025em b}\kern-.08em
		T\kern-.1667em\lower.7ex\hbox{E}\kern-.125emX}}

\usepackage{orcidlink}
\usepackage{caption}
\usepackage{enumitem}

\usepackage{xifthen}
\usepackage{xparse}

\usepackage{amsthm}
\usepackage{thmtools}
\usepackage{thm-restate}

\usepackage{mathtools}
\usepackage{proof} 
\usepackage{upgreek} %
\usepackage{mathrsfs}
\usepackage{blkarray}

\usepackage{tikz}
\usetikzlibrary{shapes.symbols}
\usetikzlibrary{arrows}
\usetikzlibrary{petri,positioning}
\usetikzlibrary{arrows.meta,arrows}
\usepackage[plain]{fancyref}
\usepackage{float}
\usepackage{stackengine}
\usepackage{tikz}
\usepackage{nicematrix}
\usepackage{xspace}
\usepackage{tcolorbox}

\usepackage{fancyhdr}
\usepackage{hyperref} %
\hypersetup{
    colorlinks = true,
    citecolor = black,
    anchorcolor = blue,
    linkcolor = black,
    urlcolor = black,
    linkbordercolor = {white}
}

\usepackage{marginfix} %
\usepackage{microtype}
\usepackage{pdflscape} %
\usepackage{lastpage}
\usepackage[capitalise]{cleveref}

\Crefname{equation}{Eq.}{Eqs.}
\Crefname{figure}{Fig.}{Figs.}
\Crefname{tabular}{Tab.}{Tabs.}
\Crefname{remark}{Rem.}{Rems.}
\Crefname{section}{Sec.}{Secs.}
\Crefname{subsection}{Sec.}{Secs.}
\Crefname{theorem}{Thm.}{Thms.}
\Crefname{example}{Ex.}{Exs.}
\Crefname{lemma}{Lem.}{Lems.}
\Crefname{corollary}{Cor.}{Cors.}
\Crefname{definition}{Def.}{Defs.}
\Crefname{appendix}{App.}{Apps.}
\Crefname{algorithm}{Alg.}{Algs.}

\crefname{equation}{Eq.}{Eqs.}
\crefname{figure}{Fig.}{Figs.}
\crefname{tabular}{Tab.}{Tabs.}
\crefname{remark}{Rem.}{Rems.}
\crefname{section}{Sec.}{Secs.}
\crefname{subsection}{Sec.}{Secs.}
\crefname{theorem}{Thm.}{Thms.}
\crefname{example}{Ex.}{Exs.}
\crefname{lemma}{Lem.}{Lems.}
\crefname{corollary}{Cor.}{Cors.}
\crefname{definition}{Def.}{Defs.}
\crefname{appendix}{App.}{Apps.}
\crefname{algorithm}{Alg.}{Algs.}

\usepackage{longtable}
\usepackage{tabularx}
\usepackage{multirow}
\usepackage{booktabs}

\usepackage{multicol} %

\usepackage{listings}

\usepackage[ruled]{algorithm}
\usepackage[noend]{algpseudocode}
\algnewcommand{\TRUE}{\textbf{True}}
\algnewcommand{\FALSE}{\textbf{False}}

\usepackage[acronym]{glossaries}

\usepackage{csquotes}

\usetikzlibrary{%
arrows,%
calc,%
shapes.misc,%
shapes.arrows,%
chains,%
matrix,%
positioning,%
scopes,%
decorations.pathmorphing,%
decorations.pathreplacing,%
shadows,%
er,%
trees,%
shapes,%
automata,%
patterns,%
knots
}

\pgfdeclarelayer{edgelayer}
\pgfdeclarelayer{nodelayer}
\pgfsetlayers{edgelayer,nodelayer,main}

\definecolor{lipicsYellow}{rgb}{0.98,0.77,0.06}
\definecolor{lipicsGray}{rgb}{0.61,0.61,0.61}

\tikzset{initial text={}}
\tikzstyle{state-circle}=[circle,draw=Black,inner sep=2pt,minimum size=0.35cm,outer sep=2pt]
\tikzstyle{state}=[ellipse,draw=Black,inner sep=2pt,outer sep=2pt]
\tikzstyle{state-box}=[rectangle,draw=Black,inner sep=2pt,minimum size=0.35cm,outer sep=2pt]
\tikzstyle{probability-state}=[circle,fill=Black,draw=Black,inner sep=0pt,minimum size=1pt]
\tikzstyle{stateText}=[rectangle,fill=White,draw=Black,inner sep=5pt,minimum height=0pt]

\tikzstyle{arrow-urgent}=[-latex,draw=Black]
\tikzstyle{arrow-non-urgent}=[-latex,draw=Black,dashed]
\tikzstyle{arrow-simple}=[-latex,draw=Black]
\tikzstyle{arrow-reset}=[-latex,dashed]
\tikzstyle{arrow-trigger}=[-latex,decoration={snake,segment length=4,amplitude=.6, post=lineto,post length=2pt},decorate]
\tikzstyle{arrow-segment-after-probabilities}=[-latex,draw=Black]
\tikzstyle{markovian}=[->,draw=Black,postaction={decorate},decoration={markings,mark=at position .5 with {\arrow{>}}}]
\tikzstyle{arrow-segment-before-probabilities-immediate}=[-,draw=Black]
\tikzstyle{arrow-segment-before-probabilities-timed}=[-,draw=Black]
\tikzstyle{non-urgent}=[densely dashed]

\tikzstyle{safe}=[fill={lipicsYellow},postaction={pattern=horizontal
  lines,pattern color=lipicsYellow!30}]
\tikzstyle{reach}=[fill={lipicsGray},postaction={pattern=vertical lines,pattern color=lipicsGray!30}]
\tikzstyle{neutral}=[fill={Gray!30},postaction={pattern=crosshatch dots,pattern color=Gray!80},font=\bf]

\newcommand{\highlight}[1]{\textsf{\textbf{#1}}}

\newcommand*{\fancyrefapplabelprefix}{app}

\fancyrefaddcaptions{english}{%
  \providecommand*{\frefappname}{appendix}%
  \providecommand*{\Frefappname}{Appendix}%
}

\frefformat{plain}{\fancyrefapplabelprefix}{\frefappname\fancyrefdefaultspacing#1}
\Frefformat{plain}{\fancyrefapplabelprefix}{\Frefappname\fancyrefdefaultspacing#1}

\frefformat{vario}{\fancyrefapplabelprefix}{%
  \frefappname\fancyrefdefaultspacing#1#3%
}
\Frefformat{vario}{\fancyrefapplabelprefix}{%
  \Frefappname\fancyrefdefaultspacing#1#3%
}

\newcommand*{\fancyrefexlabelprefix}{ex}

\fancyrefaddcaptions{english}{%
  \providecommand*{\frefexname}{example}%
  \providecommand*{\Frefexname}{Example}%
}

\frefformat{plain}{\fancyrefexlabelprefix}{\frefexname\fancyrefdefaultspacing#1}
\Frefformat{plain}{\fancyrefexlabelprefix}{\Frefexname\fancyrefdefaultspacing#1}

\frefformat{vario}{\fancyrefexlabelprefix}{%
  \frefexname\fancyrefdefaultspacing#1#3%
}
\Frefformat{vario}{\fancyrefexlabelprefix}{%
  \Frefexname\fancyrefdefaultspacing#1#3%
}

\newcommand*{\fancyrefdeflabelprefix}{def}

\fancyrefaddcaptions{english}{%
  \providecommand*{\frefdefname}{definition}%
  \providecommand*{\Frefdefname}{Definition}%
}

\frefformat{plain}{\fancyrefdeflabelprefix}{\frefdefname\fancyrefdefaultspacing#1}
\Frefformat{plain}{\fancyrefdeflabelprefix}{\Frefdefname\fancyrefdefaultspacing#1}

\frefformat{vario}{\fancyrefdeflabelprefix}{%
  \frefdefname\fancyrefdefaultspacing#1#3%
}
\Frefformat{vario}{\fancyrefdeflabelprefix}{%
  \Frefdefname\fancyrefdefaultspacing#1#3%
}

\newcommand*{\fancyreftheolabelprefix}{theo}

\fancyrefaddcaptions{english}{%
  \providecommand*{\freftheoname}{theorem}%
  \providecommand*{\Freftheoname}{Theorem}%
}

\frefformat{plain}{\fancyreftheolabelprefix}{\freftheoname\fancyrefdefaultspacing#1}
\Frefformat{plain}{\fancyreftheolabelprefix}{\Freftheoname\fancyrefdefaultspacing#1}

\frefformat{vario}{\fancyreftheolabelprefix}{%
  \freftheoname\fancyrefdefaultspacing#1#3%
}
\Frefformat{vario}{\fancyreftheolabelprefix}{%
  \Freftheoname\fancyrefdefaultspacing#1#3%
}

\newcommand*{\fancyrefalglabelprefix}{alg}

\fancyrefaddcaptions{english}{%
  \providecommand*{\frefalgname}{algorithm}%
  \providecommand*{\Frefalgname}{Algorithm}%
}

\frefformat{plain}{\fancyrefalglabelprefix}{\frefalgname\fancyrefdefaultspacing#1}
\Frefformat{plain}{\fancyrefalglabelprefix}{\Frefalgname\fancyrefdefaultspacing#1}

\frefformat{vario}{\fancyrefalglabelprefix}{%
  \frefalgname\fancyrefdefaultspacing#1#3%
}
\Frefformat{vario}{\fancyrefalglabelprefix}{%
  \Frefalgname\fancyrefdefaultspacing#1#3%
}

\newcommand*{\fancyreflemlabelprefix}{lem}

\fancyrefaddcaptions{english}{%
  \providecommand*{\freflemname}{lemma}%
  \providecommand*{\Freflemname}{Lemma}%
}

\frefformat{plain}{\fancyreflemlabelprefix}{\freflemname\fancyrefdefaultspacing#1}
\Frefformat{plain}{\fancyreflemlabelprefix}{\Freflemname\fancyrefdefaultspacing#1}

\frefformat{vario}{\fancyreflemlabelprefix}{%
  \freflemname\fancyrefdefaultspacing#1#3%
}
\Frefformat{vario}{\fancyreflemlabelprefix}{%
  \Freflemname\fancyrefdefaultspacing#1#3%
}

\newcommand*{\fancyrefcorlabelprefix}{cor}

\fancyrefaddcaptions{english}{%
  \providecommand*{\frefcorname}{corollary}%
  \providecommand*{\Frefcorname}{Corollary}%
}

\frefformat{plain}{\fancyrefcorlabelprefix}{\frefcorname\fancyrefdefaultspacing#1}
\Frefformat{plain}{\fancyrefcorlabelprefix}{\Frefcorname\fancyrefdefaultspacing#1}

\frefformat{vario}{\fancyrefcorlabelprefix}{%
  \frefcorname\fancyrefdefaultspacing#1#3%
}
\Frefformat{vario}{\fancyrefcorlabelprefix}{%
  \Frefcorname\fancyrefdefaultspacing#1#3%
}

\renewcommand{\phi}{\upvarphi}
\renewcommand{\psi}{\uppsi}

\renewcommand{\rho}{\upvarrho}
\renewcommand{\upsilon}{\upupsilon}
\renewcommand{\chi}{\upchi}
\renewcommand{\max}{\mathop{\mathsf{max}}}
\renewcommand{\min}{\mathop{\mathsf{min}}}
\renewcommand{\sup}{\mathop{\mathsf{sup}}}
\renewcommand{\inf}{\mathop{\mathsf{inf}}}

\definecolor{myblue}{RGB}{1, 26, 110}
\newcommand\newlink[2]{{\protect\hyperlink{#1}{\color{myblue} #2}}}
\makeatletter
\newcommand\newtarget[2]{\Hy@raisedlink{\hypertarget{#1}{}}#2}
\makeatother

\newcommand\BEC{\newlink{def:bec}{\text{BEC}}\xspace}
\newcommand\BECs{\newlink{def:bec}{\text{BECs}}\xspace}
\newcommand{\CSG}{\newlink{def:csg}{\text{CSG}}\xspace}
\newcommand{\CSGs}{\newlink{def:csg}{\text{CSGs}}\xspace}
\newcommand{\MDP}{\newlink{def:mdp}{\text{MDP}}\xspace}

\newcommand{\EC}{\newlink{def:ec}{\text{EC}}\xspace}
\newcommand{\ECs}{\newlink{def:ec}{\text{ECs}}\xspace}
\newcommand{\MEC}{\newlink{def:mec}{\text{MEC}}\xspace}
\newcommand{\MECs}{\newlink{def:mec}{\text{MECs}}\xspace}

\newcommand{\TSGs}{\newlink{def:tsg}{\text{TSGs}}\xspace}

\newcommand{\BVI}{\newlink{alg:bvi}{\text{BVI}}\xspace}
\newcommand{\DEFLATEALG}{\newlink{alg:deflate_mecs}{\procname{DEFLATE}}\xspace}

\newcommand{\VI}{VI\@\xspace}

\newenvironment{lemma*}[1]{\medskip\noindent\textbf{Recap of Lemma\ifthenelse{\isempty{#1}}{}{~#1}.}}{}
\newenvironment{theorem*}[1]{\medskip\noindent\textbf{Recap of Theorem\ifthenelse{\isempty{#1}}{}{~#1}.}}{}

\newcommand{\counter}{\newlink{def:counter}{\mathsf{Trap}}}
\newcommand{\exitingSG}{\newlink{def:exitingSubGame}{\mathsf{Z}^{\mathsf{exit}}_{\valuation}}}

\DeclareDocumentCommand{\variable}{O{} D<>{}}{\mathsf{v}_{#1}^{#2}}
\DeclareDocumentCommand{\variables}{O{} D<>{}}{\mathcal{V}_{#1}^{#2}}

\newcommand{\set}[1]{\{#1\}}
\newcommand{\powerset}[1]{2^{#1}}

\newcommand{\NN}{\mathbb{N}}

\newcommand{\QQ}{\mathbb{Q}}
\newcommand{\RR}{\mathbb{R}}

\DeclareDocumentCommand{\anySet}{O{} D<>{}}{\mathcal{X}_{#1}^{#2}}

\DeclareDocumentCommand{\bijection}{O{} D<>{} D(){}}{\mathsf{bi}_{#1}^{#2}\ifthenelse{\isempty{#3}}{}{(#3)}}

\DeclareMathOperator*{\argmax}{arg\,max}
\DeclareMathOperator*{\argmin}{arg\,min}

\DeclareDocumentCommand{\update}{O{} D<>{} D(){}}{\mathcal{U}_{#1}^{#2}\ifthenelse{\isempty{#3}}{}{(#3)}}
\newcommand{\proj}[2]{\mathsf{proj}_{#1}(#2)}

\newcommand{\deflBell}{(\deflop \circ \preop)}
\newcommand{\deflop}{\newlink{DefDeflop}{\mathfrak{D}}}

\newcommand{\textsll}[1]{\textsf{#1}}
\newcommand{\np}{\textsll{NP}}
\newcommand{\conp}{\textsll{co-NP}}
\newcommand{\pspace}{\textsll{PSPACE}}

\newcommand{\p}{\textsll{P}}

\DeclareDocumentCommand{\alg}{O{} D<>{} D(){}}{\mathcal{A}_{#1}^{#2}\ifthenelse{\isempty{#3}}{}{(#3)}}

\DeclareDocumentCommand{\event}{O{} D<>{} D(){}}{\mathcal{E}_{#1}^{#2}\ifthenelse{\isempty{#3}}{}{(#3)}}
\DeclareDocumentCommand{\distribution}{O{} D<>{}
  D(){}}{\mu_{#1}^{#2}\ifthenelse{\isempty{#3}}{}{(#3)}}
\DeclareDocumentCommand{\distributions}{O{} D<>{} D(){}}{\mathsf{Dist}_{#1}^{#2}\ifthenelse{\isempty{#3}}{}{(#3)}}
\DeclareDocumentCommand{\probability}{O{} D<>{} D(){} D(){}}{\mathbb{P}_{#1}^{#2}\ifthenelse{\isempty{#3}}{}{\big(#3\big)}\ifthenelse{\isempty{#4}}{}{\big(#4\big)}}
\DeclareDocumentCommand{\preop}{O{} D<>{} D(){} D(){}}{\newlink{def:pre}{\mathfrak{B}}_{#1}^{#2}\ifthenelse{\isempty{#3}}{}{(#3)}\ifthenelse{\isempty{#4}}{}{(#4)}}
\DeclareDocumentCommand{\expectation}{O{} D<>{}  D(){}}{\mathbb{E}_{#1}^{#2}\ifthenelse{\isempty{#3}}{}{[#3]}}
\DeclareDocumentCommand{\support}{O{} D<>{} D(){}}{\mathsf{Supp}_{#1}^{#2}\ifthenelse{\isempty{#3}}{}{(#3)}}

\DeclareDocumentCommand{\cylinder}{O{} D<>{} D(){}}{\mathsf{Cyl}_{#1}^{#2}\ifthenelse{\isempty{#3}}{}{\big(#3\big)}}
\DeclareDocumentCommand{\algebra}{O{} D<>{} D(){}}{\mathcal{F}_{#1}^{#2}\ifthenelse{\isempty{#3}}{}{\big(#3\big)}}
\DeclareDocumentCommand{\probspace}{O{} D<>{}}{\big(\Omega_{#1}^{#2},\mathcal{F}_{#1}^{#2},\probability\big)}
\DeclareDocumentCommand{\density}{O{} D<>{} D(){}}{\mathsf{f}_{#1}^{#2}\ifthenelse{\isempty{#3}}{}{(#3)}}

\DeclareDocumentCommand{\randomVariable}{O{} D<>{} D(){} D(){}}{\mathsf{C}_{#1}^{#2}\ifthenelse{\isempty{#3}}{}{(#3)}\ifthenelse{\isempty{#4}}{}{(#4)}}

\DeclareDocumentCommand{\valuation}{t* O{} D<>{} D(){}}{\IfBooleanTF{#1}{\eta}{\upsilon}_{#2}^{#3}\ifthenelse{\isempty{#4}}{}{(#4)}}
\DeclareDocumentCommand{\valuations}{O{} D<>{} D(){}}{\mathsf{Val}_{#1}^{#2}\ifthenelse{\isempty{#3}}{}{(#3)}}

\DeclareDocumentCommand{\matrix}{t* O{} D<>{} D(){}}{\IfBooleanTF{#1}{\mathsf{Z}}{\mathsf{Z}}_{#2}^{#3}\ifthenelse{\isempty{#4}}{}{(#4)}}

\DeclareDocumentCommand{\player}{O{} D<>{} D(){} t* O{}}{\mathsf{p}_{#1}^{#2}\ifthenelse{\isempty{#3}}{}{(#3)}\ifthenelse{\isempty{#5}}{}{[#5]}}
\DeclareDocumentCommand{\execution}{O{} D<>{} D(){} t* O{}}{\pi_{#1}^{#2}\ifthenelse{\isempty{#3}}{}{(#3)}\ifthenelse{\isempty{#5}}{}{[#5]}}
\DeclareDocumentCommand{\executions}{O{} D<>{} D(){\game}}{\mathsf{Play}_{#1}^{#2}\ifthenelse{\isempty{#3}}{}{(#3)}}

\DeclareDocumentCommand{\exampleExecution}{O{} D<>{} D(){} D||{n} t*}{
 \state[0]\IfBooleanTF{#5}{\state[1]\state[2]\cdots}{\cdots\state[#4]}}

\DeclareDocumentCommand{\last}{O{} D<>{} D(){}}{\mathsf{last}_{#1}^{#2}\ifthenelse{\isempty{#3}}{}{(#3)}}

\DeclareDocumentCommand{\path}{O{} D<>{} D(){}}{\pi_{#1}^{#2}\ifthenelse{\isempty{#3}}{}{(#3)}}
\DeclareDocumentCommand{\paths}{O{} D<>{} D(){\game}}{\mathsf{Paths}_{#1}^{#2}\ifthenelse{\isempty{#3}}{}{(#3)}}
\DeclareDocumentCommand{\examplePath}{O{} D<>{} D(){} D||{n}}{\path[#1]<#2>(#3) = \state[0]\dots\state[#4]}

\DeclareDocumentCommand{\position}{O{} D<>{} D(){}}{\mathcal{X}_{#1}^{#2}\ifthenelse{\isempty{#3}}{}{(#3)}}

\DeclareDocumentCommand{\state}{O{} D<>{} D(){}}{s_{#1}^{#2}\ifthenelse{\isempty{#3}}{}{(#3)}}
\DeclareDocumentCommand{\states}{O{} D<>{} D(){}}{\mathsf{S}_{#1}^{#2}\ifthenelse{\isempty{#3}}{}{(#3)}}
\DeclareDocumentCommand{\initialState}{O{} D<>{} D(){}}{\mathsf{s}_{0\ifthenelse{\isempty{#1}}{}{,#1}}^{#2}\ifthenelse{\isempty{#3}}{}{(#3)}}

\DeclareDocumentCommand{\action}{ O{} D(){} D<>{} }{%
  \ifstrequal{#1}{\reach}{\mathsf{a}}{%
    \ifstrequal{#1}{\safe}{\mathsf{b}}{\mathsf{p}}%
  }_{#2}^{#3}%
}
\DeclareDocumentCommand{\actions}{ O{} D<>{} }{%
  \ifstrequal{#1}{\reach}{\mathsf{A}}{%
    \ifstrequal{#1}{\safe}{\mathsf{B}}{\mathcal{A}}%
  }%
  \IfValueT{#2}{^{\mathsf{#2}}}%
}
\DeclareDocumentCommand{\badAction}{ O{} D(){} D<>{} }{%
  \ifstrequal{#1}{\reach}{\hat{\mathsf{a}}}{%
    \ifstrequal{#1}{\safe}{\hat{\mathsf{b}}}{\hat{\mathsf{p}}}%
  }_{#2}^{#3}%
}

\DeclareDocumentCommand{\badActionsPair}{O{} D<>{} D(){}}{(\hat{\action[\reach]}, \hat{\action[\safe]}){#1}^{#2}\ifthenelse{\isempty{#3}}{}{(#3)}}

\DeclareDocumentCommand{\agent}{O{} D<>{} D(){}}{\mathsf{p}_{#1}^{#2}\ifthenelse{\isempty{#3}}{}{(#3)}}
\DeclareDocumentCommand{\agents}{O{} D<>{} D(){}}{\set{\attacker,\defender}_{#1}^{#2}\ifthenelse{\isempty{#3}}{}{(#3)}}

\DeclareDocumentCommand{\transitions}{O{} D<>{} D(){} D(){} t*}{\mathsf{\delta}_{#1}^{#2}\ifthenelse{\isempty{#3}}{}{\big(#3\big)}\ifthenelse{\isempty{#4}}{}{\big(#4\big)}}

\DeclareDocumentCommand{\play}{O{} D<>{} D(){}}{\rho_{#1}^{#2}\ifthenelse{\isempty{#3}}{}{(#3)}}
\DeclareDocumentCommand{\plays}{O{} D<>{} D(){}}{\mathsf{Plays}_{#1}^{#2}\ifthenelse{\isempty{#3}}{}{(#3)}}

\DeclareDocumentCommand{\strategy}{O{} D<>{} D(){} D(){}
  t*}{\IfBooleanTF{#5}{\sigma}{\IfNoValueTF{#1}{\rho}{\ifstrequal{#1}{\safe}{\sigma}{\ifstrequal{#1}{\maximize}{\pi}{\ifstrequal{#1}{\minimize}{\pi}{\rho}}}}}_{\IfNoValueTF{#1}{}{\ifstrequal{#1}{\safe}{}{\ifstrequal{#1}{\reach}{}{#1}}}}^{#2}\ifthenelse{\isempty{#3}}{}{(#3)}\ifthenelse{\isempty{#4}}{}{(#4)}}
\newcommand{\rhoU}{\strategy[\upperBound]}
\newcommand{\rhoV}{\strategy[\val]}
\newcommand{\sigmaU}{\sigma_{\upperBound}}
\newcommand{\sigmaV}{\sigma_{\val}}
  
\DeclareDocumentCommand{\strategies}{O{} D<>{} D(){} D(){}
  t*}{\IfBooleanTF{#5}{\mathcal{S}}{\IfNoValueTF{#1}{\mathcal{R}}{\ifstrequal{#1}{\safe}{\mathcal{S}}{\ifstrequal{#1}{\maximize}{\Pi}{\ifstrequal{#1}{\minimize}{\Pi}{\mathcal{R}}}}}}_{\IfNoValueTF{#1}{}{\ifstrequal{#1}{\safe}{}{\ifstrequal{#1}{\reach}{}{#1}}}}^{#2}\ifthenelse{\isempty{#3}}{}{(#3)}\ifthenelse{\isempty{#4}}{}{(#4)}}

\DeclareDocumentCommand{\selector}{O{} D<>{} D(){}}{\chi_{#1}^{#2}\ifthenelse{\isempty{#3}}{}{(#3)}}
\DeclareDocumentCommand{\selectors}{O{} D<>{} D(){}}{\mathcal{X}_{#1}^{#2}\ifthenelse{\isempty{#3}}{}{(#3)}}

\DeclareDocumentCommand{\game}{O{} D<>{} D(){}}{\mathsf{G}_{#1}^{#2}\ifthenelse{\isempty{#3}}{}{(#3)}}
\DeclareDocumentCommand{\games}{O{} D<>{} D(){}}{\mathcal{G}_{#1}^{#2}\ifthenelse{\isempty{#3}}{}{(#3)}}
\DeclareDocumentCommand{\SG}{O{} D<>{} D(){}}{\mathsf{SG}_{#1}^{#2}\ifthenelse{\isempty{#3}}{}{(#3)}}

\DeclareDocumentCommand{\transitionProbability}{O{} D<>{}  D(){}}{\mathsf{P}_{#1}^{#2}\ifthenelse{\isempty{#3}}{}{(#3)}}

\DeclareDocumentCommand{\safe}{O{} D<>{} D(){}}{\mathscr{S}_{#1}^{#2}\ifthenelse{\isempty{#3}}{}{[#3]}}
\DeclareDocumentCommand{\reach}{O{} D<>{}  D(){}}{\mathscr{R}_{#1}^{#2}\ifthenelse{\isempty{#3}}{}{[#3]}}
\DeclareDocumentCommand{\fail}{O{} D<>{} D(){}}{\mathsf{F}_{#1}^{#2}\ifthenelse{\isempty{#3}}{}{[#3]}}
\DeclareDocumentCommand{\success}{O{} D<>{}  D(){}}{\mathsf{T}_{#1}^{#2}\ifthenelse{\isempty{#3}}{}{[#3]}}
\DeclareDocumentCommand{\val}{O{} D<>{} D(){} D(){}
  t*}{\mathsf{V}_{#1}^{#2}\ifthenelse{\isempty{#3}}{}{(#3)}\ifthenelse{\isempty{#4}}{}{(#4)}}
\DeclareDocumentCommand{\attractor}{O{} D<>{} D(){} D(){}
  t*}{\mathsf{Attr}_{#1}^{#2}\ifthenelse{\isempty{#3}}{}{(#3)}\ifthenelse{\isempty{#4}}{}{(#4)}}
\DeclareDocumentCommand{\rank}{O{} D<>{} D(){} D(){}
  t*}{\mathsf{rank}_{#1}^{#2}\ifthenelse{\isempty{#3}}{}{(#3)}\ifthenelse{\isempty{#4}}{}{(#4)}}
\DeclareDocumentCommand{\maximize}{O{} D<>{} D(){}}{\top_{#1}^{#2}\ifthenelse{\isempty{#3}}{}{[#3]}}
\DeclareDocumentCommand{\minimize}{O{} D<>{} D(){}}{\bot_{#1}^{#2}\ifthenelse{\isempty{#3}}{}{[#3]}}
\DeclareDocumentCommand{\winning}{O{\safe} D<>{} D(){}}{\mathsf{W}_{#1}^{#2}\ifthenelse{\isempty{#3}}{}{[#3]}}
\DeclareDocumentCommand{\differ}{O{} D<>{} D(){}}{\Delta_{#1}^{#2}\ifthenelse{\isempty{#3}}{}{(#3)}}

\DeclareDocumentCommand{\actionAssignment}{O{} D<>{} D(){}}{\Gamma_{#1}^{#2}\ifthenelse{\isempty{#3}}{}{(#3)}}
\DeclareDocumentCommand{\exampleGame}{O{} D<>{} D(){} t*}{\game[#1]<#2>(#3):=(\states[#1]<#2>\IfBooleanTF{#4}{(#3)}{},\actions[#1]\IfBooleanTF{#4}{(#3)}{},\actionAssignment[\reach\ifthenelse{\isempty{#1}}{}{,#1}]<#2>(#3),\actionAssignment[\safe\ifthenelse{\isempty{#1}}{}{,#1}]<#2>(#3),\transitions[#1]<#2>\IfBooleanTF{#4}{(#3)}{},\initialState,\success)}
\DeclareDocumentCommand{\destination}{O{} D<>{} D(){}}{\mathsf{Post}_{#1}^{#2}\ifthenelse{\isempty{#3}}{}{(#3)}}
\DeclareDocumentCommand{\endComponent}{O{} D<>{} D(){}}{\mathsf{C}_{#1}^{#2}\ifthenelse{\isempty{#3}}{}{(#3)}}
\DeclareDocumentCommand{\ecStates}{O{} D<>{} D(){}}{\mathcal{X}_{#1}^{#2}\ifthenelse{\isempty{#3}}{}{(#3)}}
\DeclareDocumentCommand{\ecActions}{O{} D<>{} D(){}}{\mathcal{B}_{#1}^{#2}\ifthenelse{\isempty{#3}}{}{(#3)}}
\DeclareDocumentCommand{\maxEC}{O{} D<>{} D(){}}{\mathcal{M}_{#1}^{#2}\ifthenelse{\isempty{#3}}{}{(#3)}}
\DeclareDocumentCommand{\bestExit}{O{} D<>{} D(){}}{\newlink{def:bestExit}{\mathsf{bestExit}}_{\ifstrequal{#1}{\reach}{}{#1}}^{#2}\ifthenelse{\isempty{#3}}{}{(#3)}}
\DeclareDocumentCommand{\exitVal}{O{} D<>{} D(){}}{\newlink{exitValDef}{\mathsf{exitVal}}_{\ifstrequal{#1}{\reach}{}{#1}}^{#2}\ifthenelse{\isempty{#3}}{}{(#3)}}
\DeclareDocumentCommand{\bestExitVal}{O{} D<>{} D(){}}{\newlink{def:bestExitVal}{\mathsf{bestExitVal}}_{\ifstrequal{#1}{\reach}{}{#1}}^{#2}\ifthenelse{\isempty{#3}}{}{(#3)}}
\DeclareDocumentCommand{\bestExits}{O{} D<>{} D(){}}{\newlink{def:bestExits}{\mathsf{bestExits}}_{\ifstrequal{#1}{\reach}{}{#1}}^{#2}\ifthenelse{\isempty{#3}}{}{(#3)}}
\DeclareDocumentCommand{\upperBound}{O{} D<>{} D(){}}{\mathsf{U}_{#1}^{#2}\ifthenelse{\isempty{#3}}{}{(#3)}}
\DeclareDocumentCommand{\lowerBound}{O{} D<>{} D(){}}{\mathsf{L}_{#1}^{#2}\ifthenelse{\isempty{#3}}{}{(#3)}}
\DeclareDocumentCommand{\exit}{O{} D<>{}  D(){}}{\mathsf{exit}\ifthenelse{\isempty{#2}}{}{(#2)}\ifthenelse{\isempty{#3}}{}{(#3)}}
\DeclareDocumentCommand{\stay}{O{} D<>{} D(){}}{\mathsf{stay}\ifthenelse{\isempty{#2}}{}{(#2)}\ifthenelse{\isempty{#3}}{}{(#3)}}
\DeclareDocumentCommand{\best}{O{} D<>{} D(){}}{\mathsf{best}_{#1}^{#2}\ifthenelse{\isempty{#3}}{}{(#3)}}
\DeclareDocumentCommand{\reasonable}{O{} D<>{} D(){}}{\mathsf{reasonable}_{#1}^{#2}\ifthenelse{\isempty{#3}}{}{(#3)}}
\DeclareDocumentCommand{\upperSet}{O{} D<>{} D(){}}{\mathfrak{U}_{#1}^{#2}\ifthenelse{\isempty{#3}}{}{(#3)}}
\DeclareDocumentCommand{\lowerSet}{O{} D<>{} D(){}}{\mathfrak{L}_{#1}^{#2}\ifthenelse{\isempty{#3}}{}{(#3)}}
\DeclareDocumentCommand{\onestep}{m D<>{} D(){}}{\mathsf{#1}^{#2}\ifthenelse{\isempty{#3}}{}{(#3)}}

\DeclareDocumentCommand{\leavingStrategies}{ O{} D(){} D<>{} }{%
  \ifstrequal{#1}{\reach}{\mathcal{R}_{L}}{%
    \ifstrequal{#1}{\safe}{\mathcal{S}_{L}}{\mathcal{R}_L}%
  }\ifthenelse{\isempty{#2}}{}{(#2)}%
}
\newcommand{\overbar}[1]{\mkern 1.5mu\overline{\mkern-1.5mu#1\mkern-1.5mu}\mkern 1.5mu}
\newcommand{\leaveStratsR}{\newlink{def:leavingStratR}{\mathcal{R}_{L}}}
\newcommand{\leaveStratsS}{\newlink{def:leavingStratS}{\mathcal{S}_{L}}}
\newcommand{\stayStratsR}{\newlink{def:stayingStratR}{\mathcal{R}_{\overbar{L}}}}
\newcommand{\stayStratsS}{\newlink{def:stayingStratS}{\mathcal{S}_{\overbar{L}}}}
\newcommand{\allStayStratsR}{\newlink{def:allStayingStratR}{\mathcal{R}_{S}}}
\newcommand{\allStayStratsS}{\newlink{def:allStayingStratS}{\mathcal{S}_{S}}}

\DeclareDocumentCommand{\leavingStrategy}{ O{} D(){} D<>{} }{%
  \ifstrequal{#1}{\reach}{\rho_L}{%
    \ifstrequal{#1}{\safe}{\sigma_L}{\rho_L}%
  }\ifthenelse{\isempty{#2}}{}{(#2)}%
}

\DeclareDocumentCommand{\stayingStrategies}{ O{} D(){} D<>{} }{%
  \ifstrequal{#1}{\reach}{\mathcal{R}_{S}}{%
    \ifstrequal{#1}{\safe}{\mathcal{S}_{S}}{\mathcal{R}_S}%
  }\ifthenelse{\isempty{#2}}{}{(#2)}
}

\DeclareDocumentCommand{\willStayStrategies}{ O{} D(){} D<>{} }{%
	\ifstrequal{#1}{\reach}{\mathcal{R}_{S}}{%
		\ifstrequal{#1}{\safe}{\mathcal{S}_{WS}}{\mathcal{R}_S}%
	}\ifthenelse{\isempty{#2}}{}{(#2)}
}

\DeclareDocumentCommand{\stayingStrategy}{ O{} D(){} D<>{} }{%
  \ifstrequal{#1}{\reach}{\rho_S}{%
    \ifstrequal{#1}{\safe}{\sigma_S}{\rho_S}%
  }\ifthenelse{\isempty{#2}}{}{(#2)}%
}

\newcommand{\valR}{\newlink{def:valR}{\mathsf{V}_{\reach}}}
\newcommand{\valS}{\newlink{def:valS}{\mathsf{V}_{\safe}}}

\DeclareDocumentCommand{\evaluate}{O{} D<>{} D(){\valuation[\old]} D(){\valuation[\basicEvents]}
  D(){}}{\mathsf{app}_{#1}^{#2}\ifthenelse{\isempty{#3}}{}{(#3,#4)}\ifthenelse{\isempty{#5}}{}{(#5)}}

\newcommand{\old}{\mathsf{o}}

\DeclareDocumentCommand{\side}{O{} D<>{} m m D(){}}{\mathsf{side}_{#1}^{#2}\ifthenelse{\isempty{#3}}{}{(#3,#4)}\ifthenelse{\isempty{#5}}{}{(#5)}}

\DeclareDocumentCommand{\exitActs}{O{} D<>{}  D(){}}{\mathsf{E}\ifthenelse{\isempty{#2}}{}{(#2)}\ifthenelse{\isempty{#3}}{}{(#3)}}

\DeclareDocumentCommand{\mecs}{O{} D<>{}  D(){}}{\mathsf{MEC}\ifthenelse{\isempty{#2}}{}{(#2)}\ifthenelse{\isempty{#3}}{}{(#3)}}

\newcommand{\dominates}{\newlink{def:dominant}{\,\prec\,}}
\newcommand{\dominateseq}{\newlink{def:dominant}{\,\preceq\,}}

\newcommand{\actionR}{\action[\reach]}
\newcommand{\actionS}{\action[\safe]}

\newcommand{\leaves}{\newlink{def:leaves}{\mathsf{leaves}}}
\newcommand{\stays}{\newlink{def:stays}{\mathsf{staysIn}}}

\newcommand{\deflStrats}{\newlink{def:deflStrats}{\mathsf{Defl}}}

\DeclareDocumentCommand{\matrixGame}{O{} D<>{} D(){}}{\mathsf{N}_{#1}^{#2}\ifthenelse{\isempty{#3}}{}{(#3)}}
\DeclareDocumentCommand{\matrixGameMatrix}{O{} D<>{} D(){}}{\mathsf{Z}_{#1}^{#2}\ifthenelse{\isempty{#3}}{}{(#3)}}
\DeclareDocumentCommand{\matrixGameMatrixEntry}{O{} D<>{} D(){}}{\mathsf{z}_{#1}^{#2}\ifthenelse{\isempty{#3}}{}{(#3)}}
\DeclareDocumentCommand{\matrixGamePlayers}{O{} D<>{} D(){}}{\mathit{N}_{#1}^{#2}\ifthenelse{\isempty{#3}}{}{(#3)}}

\DeclareDocumentCommand{\NFGaction}{ O{1} D(){} }{%
  \ifnum#1=1
    a_{#2}%
  \else
    b_{#2}%
  \fi
}

\DeclareDocumentCommand{\NFGstrategy}{ O{} D(){} D<>{} }{%
  \ifstrequal{#1}{\reach}{\chi}{%
    \ifstrequal{#1}{\safe}{\ypsilon}{\mathsf{p}}%
  }_{#2}^{#3}%
}

\DeclareDocumentCommand{\matrixGameActions}{O{} D<>{} D(){}}{\mathit{A}_{#1}^{#2}\ifthenelse{\isempty{#3}}{}{(#3)}}
\DeclareDocumentCommand{\matrixGameUtility}{O{} D<>{} D(){}}{\mathit{u}_{#1}^{#2}\ifthenelse{\isempty{#3}}{}{(#3)}}
\DeclareDocumentCommand{\matrixGameUtilities}{O{} D<>{} D(){}}{\mathit{u}_{#1}^{#2}\ifthenelse{\isempty{#3}}{}{(#3)}}

\DeclareDocumentCommand{\exampleMatrixGame}{O{} D<>{} D(){} t*}{\matrixGame[#1]<#2>(#3)=(\matrixGamePlayers[#1]<#2>\IfBooleanTF{#4}{(#3)}{},\matrixGameActions[#1]<#2>\IfBooleanTF{#4}{(#3)}{},\matrixGameUtilities[#1]<#2>\IfBooleanTF{#4}{(#3)}{})}

\newcommand{\badActs}{\newlink{def:badActs}{\mathsf{Hazard}}}

\DeclareDocumentCommand{\mdp}{O{} D<>{} D(){}}{\mathsf{P}_{#1}^{#2}\ifthenelse{\isempty{#3}}{}{(#3)}}
\DeclareDocumentCommand{\pmdp}{O{} D<>{} D(){}}{\mathfrak{P}_{#1}^{#2}\ifthenelse{\isempty{#3}}{}{(#3)}}
\DeclareDocumentCommand{\MDPtransitions}{O{} D<>{} D(){} D(){} t*}{\Delta_{#1}^{#2}}
\DeclareDocumentCommand{\MDPactions}{O{} D<>{} D(){}}{\mathsf{A}_{#1}^{#2}\ifthenelse{\isempty{#3}}{}{(#3)}}
\DeclareDocumentCommand{\MDPactionsAssignment}{O{} D<>{} D(){}}{\mathsf{Act}_{#1}^{#2}\ifthenelse{\isempty{#3}}{}{(#3)}}

\DeclareDocumentCommand{\parameter}{O{} D<>{} D(){}}{x_{#1}^{#2}\ifthenelse{\isempty{#3}}{}{(#3)}}
\DeclareDocumentCommand{\parameters}{O{} D<>{} D(){}}{X_{#1}^{#2}\ifthenelse{\isempty{#3}}{}{(#3)}}

\DeclareDocumentCommand{\MDPrewards}{O{} D<>{} D(){}}{\mathsf{rew}_{#1}^{#2}\ifthenelse{\isempty{#3}}{}{(#3)}}

\DeclareDocumentCommand{\inducedMDP}{O{} D<>{} D(){}}{\game\left[\strategy\right]_{#1}^{#2}\ifthenelse{\isempty{#3}}{}{(#3)}}

\DeclareDocumentCommand{\exampleMDP}{O{} D<>{} D(){} t*}{\mdp[#1]<#2>(#3):=(\states[#1],\MDPactions[#1],\MDPactionsAssignment[#1],\MDPtransitions[#1]<#2>\IfBooleanTF{#4}{(#3)}{}, \MDPrewards[#1]<#2>\IfBooleanTF{#4}{(#3)}{})}
\DeclareDocumentCommand{\examplepMDP}{O{} D<>{} D(){} t*}{\pmdp[#1]<#2>(#3)=(\states[#1]<#2>\IfBooleanTF{#4}{(#3)}{},\actions[#1]<#2>\IfBooleanTF{#4}{(#3)}{}, \parameters[#1]<#2>\IfBooleanTF{#4}{(#3)}{},\MDPtransitions[#1]<#2>\IfBooleanTF{#4}{(#3)}{})}

\newcolumntype{L}[1]{>{\hsize=#1\hsize\raggedright\arraybackslash}X}%
\newcolumntype{R}[1]{>{\hsize=#1\hsize\raggedleft\arraybackslash}X}%
\newcolumntype{C}[1]{>{\hsize=#1\hsize\centering\arraybackslash}X}%

\algrenewcommand{\algorithmiccomment}[1]{\hfill$\triangleright$ {\footnotesize #1}}
\algblockdefx[Structure]{Structure}{EndStructure}[1][Unknown]{\textsf{\textbf{Structure
      #1 \{}}}{\textsf{\textbf{\}}}}
\algblockdefx[Case]{Case}{EndCase}[1][Unknown]{\textsf{\textbf{case}}
      #1 \textsf{\textbf{of}}}{}

\algblockdefx[CaseOf]{CaseOf}{EndCaseOf}[1][Otherwise]{#1:}{}

\algblockdefx[DoUntil]{DoUntil}{EndDoUntil}{\highlight{repeat}}{\highlight{until\ }}
\newcommand{\procname}[1]{\textnormal{\textsf{#1}}}

\newcommand{\subparagraph}[1]{\smallskip\noindent\emph{#1}.}
\renewcommand{\argmax}{\mathop{\mathsf{arg \, max}}}
\renewcommand{\argmin}{\mathop{\mathsf{arg \, min}}}

\newcommand{\rest}{\mathit{rest}}
\newcommand{\hideorrun}{\texttt{Hide-Run-or-Slip}\xspace}

\newacronym{csg}{CSG}{Concurrent Stochastic Game}
\newacronym{lp}{LP}{Linear Program}
\newacronym{ec}{EC}{End Component}
\newacronym{sec}{SEC}{Simple End Component}
\newacronym{msec}{MSEC}{Maximal SEC}
\newacronym{qcqp}{QCQP}{Quadraticaly Constrained Quadratic Program}
\newacronym{sdp}{SDP}{Semidefinite Programming}
\newacronym{mdp}{MDP}{Markov Decision Process}
\newacronym{mC}{Mc}{Markov Chain}
\newacronym{tsg}{TSG}{Turn-based Stochastic Game}
\newacronym{nfg}{NFG}{Normal Form Game}
\newacronym{vi}{VI}{Value Iteration}
\newacronym{bvi}{BVI}{Bounded Value Iteration}
\newacronym{mec}{MEC}{Maximal End Component}
\newacronym{bec}{BEC}{Bad End Component}
\newacronym{qp}{QP}{Quadratic Program}

\tikzstyle{new style 0}=[fill=none, draw=none, shape=circle]
\tikzstyle{none}=[fill=none, draw=none, shape=circle]
\tikzstyle{left-align}=[fill=none, shape=circle, align=left]
\tikzstyle{rectangle}=[fill={rgb,255: red,232; green,232; blue,232}, draw=black, shape=rectangle, minimum width=4cm, minimum height=1cm]
\tikzstyle{rectangle-large}=[fill={rgb,255: red,232; green,232; blue,232}, draw=black, shape=rectangle, minimum width=4.8cm, minimum height=1cm]
\tikzstyle{center-align}=[fill=none, draw=none, shape=circle, align=center]
\tikzstyle{rectangle-white}=[fill=white, draw=black, shape=rectangle, align=center, minimum width=3.5cm, minimum height=1cm, rounded corners=1ex]

\tikzstyle{arrow}=[->]
\tikzstyle{line}=[-, fill=white]
\tikzstyle{green-box-line}=[-, fill={rgb,255: red,197; green,245; blue,186}]
\tikzstyle{thin line}=[-, line width=1pt]
\tikzstyle{yellow-box-line}=[-, fill={rgb,255: red,247; green,255; blue,170}]
\tikzstyle{arrow 2}=[->]
\tikzstyle{new edge style 0}=[->]
\tikzstyle{blue box}=[-, fill={rgb,255: red,208; green,241; blue,245}]
\tikzstyle{white-box}=[-, fill=white]
\tikzstyle{gray-box}=[-, fill={rgb,255: red,244; green,244; blue,244}]
\tikzstyle{arrow-uthick}=[->, ultra thick]
\tikzstyle{arrow-vthick}=[->, very thick]

\theoremstyle{plain}
\newtheorem{lemma}{Lemma}
\theoremstyle{definition}
\newtheorem{definition}[lemma]{Definition}
\theoremstyle{plain}
\newtheorem{theorem}[lemma]{Theorem}
\theoremstyle{remark}
\newtheorem{remark}[lemma]{Remark}

\AtEndPreamble{%
	\theoremstyle{remark}
	\newtheorem{example}[lemma]{Example}
	\AtBeginEnvironment{example}{\pushQED{}}%
	\AtEndEnvironment{example}{\hfill$\triangle$\par\popQED}%
	
	\declaretheoremstyle[
	headpunct={:}
	]{argumentstyle}
	
	\declaretheorem[name=Claim,style=argumentstyle]{argument}
	
	\AtBeginEnvironment{argument}{\pushQED{}}%
	\AtEndEnvironment{argument}{\hfill$\blacktriangle$\par\popQED}%
}

\usepackage{etoolbox}
\newtoggle{arxiv}
\newcommand{\ifarxivelse}[2]{\iftoggle{arxiv}{#1}{#2}}
\toggletrue{arxiv} %

\begin{document}
	
\newcommand{\anonym}{1} %

\title{Stopping Criteria for Value Iteration on Concurrent Stochastic Reachability and Safety Games
	\ifthenelse{\anonym=1}{
	\thanks{This research was funded in part by the German Research Foundation (DFG) project 427755713 GOPro, the MUNI Award in Science and Humanities (MUNI/I/1757/2021) of the Grant Agency of Masaryk University, the European Union’s Horizon 2020 research and innovation programme under the Marie Sklodowska-Curie grant agreement No 101034413, and the ERC Starting Grant DEUCE (101077178).}
}
}

\ifthenelse{\anonym=0}{
\author{\bigskip
}
}
{
\author{
	\IEEEauthorblockN{
		Marta Grobelna\IEEEauthorrefmark{1}\orcidlink{0009-0003-3314-3358},
		Jan K\v{r}et\'insk\'y\IEEEauthorrefmark{2}\IEEEauthorrefmark{1}\orcidlink{0000-0002-8122-2881},
		Maximilian Weininger\IEEEauthorrefmark{3}\IEEEauthorrefmark{4}\IEEEauthorrefmark{1}\orcidlink{0000-0002-0163-2152}
	}
	\vspace{0.2cm}
	
	\IEEEauthorblockA{
		\begin{tabular}{cccc}
			\shortstack{\IEEEauthorrefmark{1}Technical University of Munich\\Munich, Germany} & 
			\shortstack{\IEEEauthorrefmark{2}Masaryk University\\Brno, Czech Republic} & 
			\shortstack{\IEEEauthorrefmark{3}Ruhr-University Bochum\\ Bochum, Germany} &
			\shortstack{\vspace{0.05cm}\IEEEauthorrefmark{4}Institute of Science and\\ Technology Austria }	
			 \\
			 marta.grobelna@tum.de & jan.kretinsky@fi.muni.cz &maximilian.weininger@rub.de & Klosterneuburg, Austria
		\end{tabular}
	}
}
}

\pagestyle{plain}

\maketitle

\begin{abstract}
We consider two-player zero-sum concurrent stochastic games (CSGs) played on graphs with reachability and safety objectives. These include degenerate classes such as Markov decision processes or turn-based stochastic games, which can be solved by linear or quadratic programming; however, in practice, value iteration (VI) outperforms the other approaches and is the most implemented method. Similarly, for CSGs, this practical performance makes VI an attractive alternative to the standard theoretical solution via the existential theory of reals.

VI starts with an under-approximation of the sought values for each state and iteratively updates them, traditionally terminating once two consecutive approximations are $\epsilon$-close. However, this stopping criterion lacks guarantees on the precision of the approximation, which is the goal of this work. We provide \emph{bounded} (a.k.a. interval) VI for CSGs: it complements standard VI with a converging sequence of over-approximations and terminates once the over- and under-approximations are $\epsilon$-close.   
\end{abstract}


\section{Introduction}
\label{sect:intro}

\newcommand{\para}[1]{\paragraph*{#1}}

\para{Concurrent stochastic games (CSGs, e.g., \cite{alfarothesis})} We consider two-player zero-sum games played on a graph. Every vertex represents a \emph{state}. Edges are directed, originating from one state and leading to one or several other states. An edge is associated with a \emph{probability distribution} over the successor states.
A \emph{play} proceeds through the graph as follows: starting from an initial state, both players simultaneously and independently choose an \emph{action}, determining the edge to follow. Then, the next state is sampled according to the probability distribution, and the process is repeated in the successor.

We focus on infinite-horizon reachability and safety objectives \cite{alfarothesis}. The goal of reachability is to maximize the probability of reaching a given goal state. In contrast, the safety objective aims to maximize the probability of staying within a given set of states. The two objectives are dual, as instead of maximizing the probability of reaching a set of target states, one can minimize the probability of staying within the set of non-target states. Thus, we refer to both types collectively as CSGs.
Popular subclasses of CSGs include turn-based stochastic games (TSGs), where the players make decisions in turns, or Markov decision processes (MDPs), which involve only one player.

\para{The value problem} In CSGs, memoryless (a.k.a.~stationary) strategies suffice for both players, meaning they yield the same supremum probability as history-dependent strategies. 
However, unlike in TSGs, the strategies for CSGs require \emph{randomization}, meaning players choose distributions over actions rather than single actions. 
Additionally, while the safety objective player, Player $\safe$, can attain optimal strategies \cite{Parthasarathy1973}, the reachability objective  player, Player $\reach$, only possesses \emph{$\varepsilon$-optimal} strategies for a given $\varepsilon>0$~\cite{Everett1957}.  
As a result, the problem of deciding whether the supremum probability (a.k.a.\ \emph{the value}) is at least~$p$ for $p \in [0,1]$, is thus more subtle than for the mentioned subclasses.
While for MDPs the value problem is in \p, and for TSGs it is known to be in $\np\cap\conp$, for CSGs it can be elegantly encoded into the \emph{existential theory of reals (ETR)}, which is only known to be decidable in \pspace{} (although not known to be complete for it) \cite{etessamiRecursiveConcurrentStochastic2008}. 
Unfortunately, algorithms for ETR are practically even worse than the more general, doubly exponential methods for the first-order theory of reals \cite{DBLP:conf/mkm/PassmoreJ09}.
\emph{``Finding a practical algorithm remains a very interesting open problem''} \cite{hansenComplexitySolvingReachability2014}. 

\para{Practical approximation} We focus on algorithms \emph{approximating} the value with a predefined precision $\varepsilon>0$.
Both for MDPs and TSGs, dynamic programming techniques such as \emph{value iteration (\VI)} or strategy iteration (SI) are practically more efficient than mathematical programming (linear or quadratic, respectively) \cite{hartmannsPractitionerGuideMDP2023, kretinskyComparisonAlgorithmsSimple2022}. 
Thus, VI algorithms are prevalently used and implemented in popular tools such as \textsc{prism-games} \cite{kwiatkowskaPRISMgamesStochasticGame2020d}, motivating the focus on VI here.

\para{Problem and our contribution} In VI, the lowest possible value is initially assigned to each state and then iteratively improved, computing an under-approximation of the value, converging to it in the limit. 
The algorithm (in practical implementations) terminates once two consecutive approximations are $\varepsilon$-close.  
However, the result can then be arbitrarily imprecise \cite{Haddad2018}.
In this work, we introduce \emph{bounded value iteration for CSGs}, following its previous success for MDPs \cite{Brazdil2014,ashokValueIterationLongRun2017a} or TSGs \cite{eisentrautValueIterationSimple2022}.
Its main idea is to enhance standard VI by introducing an over-approximation of the values computed in parallel with the under-approximation. Once the upper and lower bounds are $\varepsilon$-close, VI terminates, ensuring that the true value is at most $\varepsilon$ away from the obtained approximation. 
Since the na\"ive formulation of an upper bound does not converge to the value in general, 
previous approaches, notably including \cite{Chatterjee2009}, have attempted to fix this but have failed.
In this paper, we finally provide a valid solution. %

\para{Technical challenge} The fundamental technical difficulties arise from the following. The non-convergence of upper-bound approximations is primarily due to cyclic components, so-called \emph{end components (ECs)}, see \cite{alfarothesis, eisentrautValueIterationSimple2022, kretinskyStoppingCriteriaValue2023a}. 
Notably, non-convergence, for this reason, is an issue present already in MDPs and TSGs; see \cite{eisentrautValueIterationSimple2022, kretinskyStoppingCriteriaValue2023a}. 
Solutions have been developed over the past decade for these two subclasses. Indeed, for MDPs with reachability objective, these end components can effectively be removed from the graph without altering the value \cite{brazdilVerificationMarkovDecision2014, haddadIntervalIterationAlgorithm2018}. Other objectives were considered in \cite{baierEnsuringReliabilityYour2017, ashokValueIterationLongRun2017a}. TSGs with reachability objectives already require more careful analysis, decomposing the end components into sub-parts, so-called \emph{simple ECs} \cite{eisentrautValueIterationSimple2022}. A comprehensive framework for various quantitative objectives was proposed in \cite{kretinskyStoppingCriteriaValue2023a}.
Unfortunately, the idea of simple ECs is not easily extendable from TSGs to CSGs due to the absence of optimal strategies.

\para{Summary of our contribution}
We provide a stopping criterion for VI on CSGs, solving an open problem with erroneous solution attempts in the literature (see the related work in Subsec.~\ref{rw} below).
To this end, we unravel the recursive hierarchical structure of end components in CSGs (see Rem.~\ref{rem:struct}) and adapt the bounded VI algorithm.

\subsection{Related Work}

\para{Available approaches}
The \pspace-algorithm introduced in \cite{etessamiRecursiveConcurrentStochastic2008} for deciding whether the value of a given game is at least $p$, for $p \in [0,1]$, allows for a trivial stopping criterion by iteratively executing this algorithm for a suitable sequence of $(p_i)_{i \in \NN}$ (intuitively, we choose $p_i$ such that alternatingly, the value of the game is above and below, while the distance between two consecutive $p_i$'s monotonically decreases). However, this criterion is impractical since it uses the existential theory of reals \cite{etessamiRecursiveConcurrentStochastic2008}. 
The best known complexity upper bound comes from \cite{frederiksenApproximatingValueConcurrent2013} and states that the problem of approximating the value of a CSG is in TFNP[NP], i.e.\ total function from NP with an oracle for NP. However, the proposed algorithm is not practical, as it relies on guessing a floating point representation of the value and optimal strategies for both players.
A recursive bisection algorithm was introduced in \cite{hansenExactAlgorithmsSolving2011}, which is also impractical as its time complexity is best-case doubly exponential in the number of states. 
For the algorithms commonly employed for in the non-concurrent case, VI and SI, \cite{hansenComplexitySolvingReachability2014} provide doubly exponential lower and upper bounds on the number of iterations that VI requires in the worst-case for computing an $\varepsilon$-approximation. 
Their counter-example uses a CSG where all states have value 1. Thus, the worst-case complexity of our approach is the same since an additional over-approximation does not speed up convergence in this example. Nonetheless, these results are worst-case bounds, i.e.\ they hold a priori for all games; earlier termination is possible, but necessarily requires a stopping criterion, which has so far been elusive.
Finally, in~\cite{oliu-bartonNewAlgorithmsSolving2019}, an algorithm is provided that, unlike all other known algorithms, only has a \emph{single}-exponential dependency on the number of states. 
A practical comparison of our value iteration and~\cite{oliu-bartonNewAlgorithmsSolving2019} is an interesting future step, as better worst-case complexity of an algorithm need not translate to better practical performance on typical instances; for example, in MDPs, worst-case exponential VI and SI typically outperform the polynomial approach of linear programming~\cite{hartmannsPractitionerGuideMDP2023}.

\para{Previous attempts at stopping criterion}\label{rw}
A stopping criterion for SI and VI on CSGs was first presented in \cite{Chatterjee2009}, but later found to contain an irreparable mistake \cite{Chatterjee2013}. Specifically, the algorithm returned over-approximations smaller than the actual values in certain situations, as detailed in \cite{Chatterjee2013}. We analyze the counter-example from \cite{Chatterjee2013} in \ifarxivelse{\cref{counter-example:chaterjee}}{\cite[App.~XX]{arxivversion}} %
Later, \cite{eisentrautStoppingCriteriaValue2019} proposed a stopping criterion for VI, which also contains a fundamental flaw: it fails to converge for CSGs with ECs. We analyze the counter-example to this approach in \ifarxivelse{\cref{counter-example:eisentraut}}{\cite[App.~XX]{arxivversion}}. %
\emph{Our work thus delivers the first stopping criterion in this context.}

\para{Further directions of related work}  Variants of CSGs have appeared very early, under the names of Everett, Gilette, or Shapley game. See \cite{hansenExactAlgorithmsSolving2011} for an explanation of all game types, their relations, and algorithms to solve them. 
These games also consider discounted payoff or limit-average payoff, generalizing the reachability and safety CSGs we consider here.
A generalization of CSGs to $\omega$-regular objectives has been
considered in \cite{Alfaro2000, chatterjeeQualitativeConcurrentParity2011}. 
An insightful characterization of optimal strategies in concurrent games with various objectives can be found in \cite{bordaisSubgameOptimalStrategies2023}.

\section{Preliminaries}
\label{sec:preliminaries}

\subsection{Concurrent Stochastic Games}

\para{Probability Distributions} For a countable set $X$, a function $\mu \colon X \to [0,1]$ is called a \emph{distribution} over~$X$ if $\sum_{x \in X} \mu(x) = 1$.  The \emph{support} of $\mu$ is $\support(\mu) := \set{x \mid \mu(x)>0}$. The set of all distributions over $X$ is denoted by $\distributions(X)$. 

\para{Concurrent Stochastic Games} 
A \emph{concurrent stochastic game}~\newtarget{def:csg}{(CSG)} \cite{Alfaro2007} is a tuple $\newtarget{def:G}{\exampleGame}$,
where $\states$ is a finite set of \emph{states}, 
$\actions:= \actions[\reach] \times \actions[\safe]$ is a finite set of \emph{actions} 
with $\actions[\reach]: = \{ \action[\reach](1), \ldots, \action[\reach](l)\}$ 
and $\actions[\safe] := \{ \action[\safe](1), \ldots, \action[\safe](m)\}$ the sets of actions available for player $\reach$ and $\safe$, respectively,
$\actionAssignment[\reach] \colon \states \to (\powerset{\actions[\reach]} \setminus \emptyset)$ and $\actionAssignment[\safe] \colon \states \to (\powerset{\actions[\safe]} \setminus \emptyset)$ are two \emph{enabled actions} assignments and $\transitions \colon \states \times \actions[\reach] \times \actions[\safe] \to \distributions(\states)$ is a \emph{transition function}, where $\transitions(\state,\action[\reach],\action[\safe])(\state<\prime>)$ gives the \emph{probability of a transition} from state~$\state$ to
  state~$\state<\prime>$ when player~$\reach$ chooses action~$\action[\reach]\in
  \actionAssignment[\reach](\state)$ and player~$\safe$ action~\mbox{$\action[\safe]\in \actionAssignment[\safe](\state)$},
  $\initialState\in\states$ is an \emph{initial state}, and $\success\subseteq\states$ is a set of \emph{target states}. 
A \CSG is \emph{turn-based} if for every state~$\state$ only one player has a meaningful choice, i.e.\ either $\actionAssignment[\reach](\state)$ or $\actionAssignment[\safe](\state)$ is a singleton; we call such game a turn-based stochastic game (TSG)\newtarget{def:tsg}.

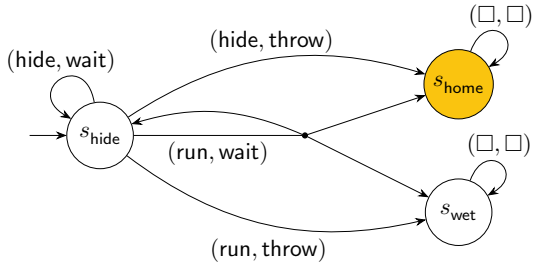
\begin{figure}
	\centering
	\resizebox{0.4\textwidth}{!}{\begin{tikzpicture}[
state/.style={draw,circle, minimum size = 1.3cm, align=center},
stochSt/.style={draw,circle, fill=black, minimum size = 0.1cm, inner sep=0pt},
transition/.style={->, black, -{Stealth[length=2mm]}, align=center},
invalid/.style={draw,rectangle,fill=black!30,minimum width=2cm,minimum height=1cm},
every label/.append style = {font=\small}
]

	\node (Start) at (-1.5, 0) {};
    \node[state] (S0) at (0,0) {\scalebox{1.3}{$\state[\mathsf{hide}]$}};
    \node[stochSt, right of=S0, xshift=3cm] (stoch) {};
    \node[state, fill=lipicsYellow ,right of=stoch, xshift=2cm, yshift=1cm] (S1)  {\scalebox{1.3}{$\state[\mathsf{home}]$}};
    \node[state,right of=stoch, xshift=2cm, yshift=-1.5cm] (S2)  {\scalebox{1.3}{$\state[\mathsf{wet}]$}};

    \draw[transition] (S0) edge[bend left=25] node[above]{\scalebox{1.3}{$(\mathsf{hide},\mathsf{throw})$}} (S1);
    \draw[transition] (S0) edge[bend right=25] node[below,yshift=-0.15cm] {\scalebox{1.3}{$(\mathsf{run},\mathsf{throw})$}} (S2);
    \draw[transition] (stoch) edge[bend right=20] (S0);
    \draw[-] (S0) to node[below,yshift=0cm]{\scalebox{1.3}{$(\mathsf{run},\mathsf{wait})$}} (stoch);
    \draw[transition] (stoch) to (S1); 
    \draw[transition] (stoch) to (S2);
    \draw[transition] (S0) edge[in=150,out=100,looseness=5] node[above]{\scalebox{1.3}{$(\mathsf{hide},\mathsf{wait})$}} (S0);
    \draw[transition] (S1) edge[in=30,out=70,looseness=5] node[above]{\scalebox{1.3}{$(\square, \square)$}} (S1);
    \draw[transition] (S2) edge[in=30,out=70,looseness=5] node[above,yshift=0cm]{\scalebox{1.3}{$(\square, \square)$}} (S2);
    \draw[transition] (Start) to (S0);
\end{tikzpicture}}
	\caption{Example CSG called \hideorrun.}
	\label{fig:single_state_example}
\end{figure}

\begin{example}[\CSGs]\label{ex:runningExampleIntroduction}
		Consider the \CSG \hideorrun depicted in \Cref{fig:single_state_example} (an adaption of the \texttt{Hide-or-Run} game in~\cite{kumarExistenceValueRandomized1981a,Alfaro2007,Everett1957}).
		Circles represent states and black dots depict a probabilistic transition with uniform distribution. Each edge is labeled with a pair of actions, the left for player $\reach$ and the right for player $\safe$; $\square$ is a placeholder for an arbitrary action. We have $\states \coloneqq \{\state[\mathsf{hide}], \state[\mathsf{home}], \state[\mathsf{wet}]\}$, with $\actionAssignment[\reach](\state[\mathsf{hide}]) \coloneqq \{\mathsf{hide},\mathsf{run}\}$ and $\actionAssignment[\safe](\state[0]) \coloneqq \{\mathsf{wait},\mathsf{throw}\}$, $\state[0]$ is the initial state denoted by the arrow with no predecessor state, and $\success \coloneqq\{\state[\mathsf{home}]\}$.
		The game has the following intuitive interpretation: Player $\reach$ wants to get home without getting wet. Player $\safe$ has a single snowball and can make player $\reach$ wet by throwing it at player $\reach$. If player $\reach$ runs and player $\safe$ throws the ball, player $\reach$ gets wet. If player $\reach$ runs but player $\safe$ waits, with a probability of $\frac{1}{3}$ the player reaches home or the player slips instead and with a probability of $\frac{1}{3}$ does not move at all or with a probability of $\frac{1}{3}$ falls on the ground and gets wet.
\end{example}

\para{Plays} A \emph{play}~$\execution$ of a CSG $\newlink{def:G}{\game}$ is an infinite sequence of states $\exampleExecution*$, such that for all $i \in \NN$ there are actions $\action[\reach] \in \actionAssignment[\reach](\state[i])$ and $\action[\safe] \in \actionAssignment[\safe](\state[i])$ with $\transitions(\state[i],\action[\reach],\action[\safe])(\state[i+1]) > 0$. $\executions$ is the set of all plays and $\executions[\state]$ the set of all plays $\exampleExecution*$ with $\state[0]=\state$. 

\para{Strategies}
A \emph{strategy}~for player~$\reach$ (or $\safe$) is a function $\strategy[\reach] \colon \states \to \distributions(\actions[\reach])$ (or $\strategy[\safe] \colon \states \to \distributions(\actions[\safe])$) that assigns a distribution over actions available to player~$\reach$ (or $\safe$) to each state, i.e., for all $\state \in \states$,  $\support(\strategy(\state))
\subseteq \actionAssignment[\reach](\state)$ (or $\support(\strategy*(\state))
\subseteq \actionAssignment[\safe](\state)$).\footnote{Since memoryless strategies are sufficient for the objectives considered in this paper, we do not introduce general history-dependent strategies to avoid clutter. We refer to \cite{Alfaro2007} for more details.} We call a Player $\reach$ (or Player $\safe$) strategy $\strategy$ \emph{pure} if all distributions it returns are Dirac distributions, i.e., at each $\state \in \states$ we have a unique action $\actionR \in \actionAssignment[\reach](\state)$ (or  $\actionS \in \actionAssignment[\safe](\state)$) such that $\strategy(\state)(\actionR) = 1$ (or $\strategy*(\state)(\actionS) = 1$). Otherwise, the strategy is \emph{mixed}. For player $\reach$ (or $\safe$) we denote the set of strategies by $\strategies$ (or $\strategies*$) and a single strategy by $\strategy$ (or $\strategy*$).

\para{Markov Decision Processes} Given a \CSG $\game$, if we fix a strategy $\strategy \in \strategies$ of player $\reach$, the game becomes a $\safe$-\emph{Markov Decision Process} (\newtarget{def:mdp}{MDP},~\cite{putermanMarkovDecisionProcesses2009}) $\game[\strategy]$ with the transition function
\begin{align*}
	\transitions[\strategy](\state, \actionS)(\state<\prime>) \coloneqq \sum_{\actionR \in \actionAssignment[\reach](\state)} \transitions(\state, \actionR, \actionS)(\state<\prime>) \cdot \strategy(\state)(\actionR),
\end{align*}
for all $\state \in \states$ and $\actionR \in \actionAssignment[\reach](\state)$. The \MDP induced by a fixed strategy $\strategy* \in \strategies*$ is defined analogously. 

\para{Markov Chains} Similarly, by fixing a pair of strategies $(\strategy,\strategy*) \in \strategies \times \strategies*$, we obtain a \emph{Markov chain} $\game[\strategy, \strategy*]$ with the same state space $\states$, the initial state $\initialState$, and the transition probabilities $P$ given by 
\[
\transitions[\strategy, \strategy*] (\state)(\state<\prime>)\coloneqq\sum_{(\actionR,\actionS) \in \actions} 
\transitions(\state,\action[\reach],\action[\safe])(\state<\prime>) \cdot
\strategy[\reach](\state)(\action[\reach]) \cdot \strategy[\safe](\state)(\action[\safe]).
\]
Thus, a pair of strategies ($\strategy[\reach],\strategy[\safe]$) induces a unique probability measure $\probability[\initialState]<\strategy[\reach],\strategy[\safe]>$ over plays in the Markov chain as usual, see~\cite[Chap. 10.1]{baierPrinciplesModelChecking2008a}, where the set of paths starting in $\initialState$ has measure~1.

\para{Objectives} We partition $\states$ into $\success$, denoting the set of states player~$\reach$ wants to reach, and $\fail:=\states\setminus\success$, denoting the set of states player~$\safe$ wants to confine the game in. We denote the \emph{reachability} objective by $\Diamond \success \coloneqq \set{ \exampleExecution* \mid \exists i \in \NN: \state[i] \in \success}$ and the \emph{safety} objective by $\Box \fail \coloneqq \set{\exampleExecution* \mid \forall i \in \NN: \state[i] \in \fail}$. The \emph{value} of the objective $\Diamond \success$, i.e. $\newtarget{def:valR}{\valR(\state)}$, and the objective $\Box\fail$, i.e. $\newtarget{def:valS}{\valS(\state)}$, at state~$\state$ are given by
\begin{align*}
	 &\valR(\state) \coloneqq \adjustlimits\sup_{\strategy[\reach] \in \strategies[\reach]} \inf_{\strategy[\safe]\in\strategies[\safe]} \probability[\state]<\strategy[\reach],\strategy[\safe]>(\Diamond \success) \\ 
	 &\valS(\state) \coloneqq \adjustlimits\sup_{\strategy[\safe] \in \strategies[\safe]} \inf_{\strategy[\reach]\in\strategies[\reach]} \probability[\state]<\strategy[\reach],\strategy[\safe]>(\Box \fail).
\end{align*}%

By the determinacy of \CSGs and the duality of these objectives~\cite{Everett1957,martinDeterminacyBlackwellGames1998}, it holds that $\valR(\state)+\valS(\state)=1$.
Consequently, the task of approximating $\valR(\state)$ with a given precision is equivalent to approximating $\valS(\state)$.
Further, the objective of minimizing the reachability for $\success$ is equivalent to the objective of maximizing safety for $\fail$ for the \emph{same} player.
Consequently, in the following we only focus on maximizing reachability as both minimization and the safety objectives can be reduced to it.

\begin{example}[Optimal Strategies Need Not Exist]\label{ex:2-no-opt-strat}
	In \CSGs, an optimal strategy for player $\reach$ might not exist \cite{kumarExistenceValueRandomized1981a}, meaning that at some states, the value is attainable only in the limit. 
	Consider our running example from \cref{fig:single_state_example}, \hideorrun. 
	Assume for the moment that there is no chance of slipping, i.e.\ upon playing $\mathsf{run}$ and $\mathsf{wait}$, the target state is reached.
	To win, Player $\reach$ has to run eventually. 
	However, Player $\safe$ can utilize a strategy that throws with positive probability at all points in time.
	Thus, Player $\reach$ cannot win almost surely.
	
	However, Player $\reach$ has the possibility of \emph{limit-sure winning} in~\cite{Alfaro2007}: By running with vanishingly low probability $\varepsilon$ in every round, the probability of winning is $1-\varepsilon$.
	This is because Player $\safe$ has the highest probability $\varepsilon$ of hitting Player $\reach$ by throwing in the first round; throwing in a later round $n$ only has hitting probability $\varepsilon^n$.
	For any $\varepsilon>0$, this strategy of Player $\reach$ achieves $1-\varepsilon$.	
	The value, being the supremum over all strategies, is 1.
	
	This notion of obtaining a value only in the limit is not restricted to sure winning:
	By adding the chance of slipping, the value of the game becomes 0.5. 
	However, by the same argument as above, Player $\reach$ cannot win with probability 0.5, but only with a probability $0.5-\varepsilon$ for all $\varepsilon>0$.
\end{example}

For $\state \in \states$, $\action[\reach] \in \actionAssignment[\reach](\state)$ and $\action[\safe] \in \actionAssignment[\safe](\state)$, the \emph{set of potential successors} of~$\state$ is denoted by $\destination(\state,\action[\reach],\action[\safe]) \coloneqq \support(\transitions(\state, \action[\reach], \action[\safe]))$. We lift the notation to strategies~$\strategy[\reach] \in \strategies$ and $\strategy[\safe] \in \strategies*$ by \[ \destination(\state,\strategy[\reach], \strategy[\safe]) = \bigcup_{\action[\reach] \in \support(\strategy[\reach](\state))} \bigcup_{\action[\safe] \in \support(\strategy[\safe](\state))} \destination(\state,\action[\reach],\action[\safe]).\]

We denote by $\winning \coloneqq \set{ \state \in \states \mid \valR(\state) = 0}$ the \emph{sure winning} region of player~$\safe$. It can be computed in at most $|\states|$ steps by iteration $\winning<0> \coloneqq (\states \setminus \success)$ and $\winning<k+1> \coloneqq \set{s \in \states \setminus
  \success \mid \exists \action[\safe] \in
  \actionAssignment[\safe](\state): \forall \action[\reach] \in
  \actionAssignment[\reach](\state):
  \destination(\state,\action[\reach],\action[\safe]) \subseteq
  \winning<k>}$ for all $k \in \NN$~\cite{Alfaro2000}. Consequently, we can assume without loss of generality that $\success$ and $\winning$ are both singletons and absorbing.
  
\begin{example}[The Sets $\success$, $\fail$, and $\winning$]\label{ex:running-example-2}
In \Cref{fig:single_state_example}, player $\reach$ wants to reach $\success = \{\state[1]\}$, while player $\safe$ aims to stay in  $\fail = \{\state[0], \state[2]\}$. Since $\state[2]$ is absorbing, $\winning = \{\state[2]\}$.
\end{example}

\para{Matrix Games} 
At each state of a \CSG the players $\reach$ and $\safe$ play a two-player zero-sum \emph{matrix game}~\cite{maschlerGameTheory2020,santosAutomaticVerificationStrategy2020}. In general, a matrix game is a tuple $\mathsf{Z} \coloneqq (N, A, u)$ \cite{santosAutomaticVerificationStrategy2020}  where, $N \coloneqq  \{ 1, \ldots, n \}$ is a finite set of players, $A \coloneqq  \{ \alpha_1, \ldots, \alpha_m\}$ is a finite set of actions available to each player, and $\matrixGameUtility \colon A \to \mathbb{Q}$  is a utility function. In \CSGs, the matrix game played at a specific state $\state$ can be represented by a matrix $\matrixGameMatrix[\valR](\state) \in \mathbb{Q}^{l\times m}$, where $\actionAssignment[\reach](\state) = \{ \NFGaction[1](1), \ldots, \NFGaction[1](l)  \}$ and $\actionAssignment[\safe](\state) = \{ \NFGaction[2](1), \ldots, \NFGaction[2](m)\}$. The entries of the matrix correspond to the utility, i.e., the value attainable upon choosing a pair of actions $(\actionR(i), \actionS(j)) \in A$. Thus, the $i$-th row and the $j$-th column is given by $\matrixGameMatrix[\valR](\state)(i,j) \coloneqq \sum_{\state<\prime> \in \states} \transitions(\state,\NFGaction[1](i),\NFGaction[2](j))(\state<\prime>) \cdot \valR(\state<\prime>)$. 

\begin{example}[Matrix Game]
	Consider the CSG in \cref{fig:single_state_example}. The matrix game played at state $\state[\mathsf{hide}]$ is given by the following matrix.
		\begin{align}\label{eq:matrixGameMatrixVr}
		\matrixGameMatrix[\valR](\state[\mathsf{hide}]) = \begin{blockarray}{ccc}
			\scriptstyle \mathsf{throw} & \scriptstyle \mathsf{wait}\\
			\begin{block}{(cc)c}
				0 & \frac{1}{3} \cdot \valR(\state[\mathsf{hide}]) + \frac{1}{3}  & \scriptstyle \mathsf{run}\\
				1 & \valR(\state[\mathsf{hide}]) & \scriptstyle \mathsf{hide}\\
			\end{block}
		\end{blockarray}
	\end{align}\label{eq:matrixGameTrivialEC}
Player $\reach$ is the so called \emph{row player} while Player $\safe$ is the \emph{column player}.
\end{example}

In a matrix game, a player's \emph{strategy} is a distribution over the available actions at a specific state. To distinguish between strategies of a \CSG and strategies of a matrix game, we refer to strategies of a matrix game as \newtarget{def:local}{\emph{local strategies}} and strategies of a \CSG as \emph{global} strategies. The set of all local strategies at a state $\state$ is denoted by $\strategies(\state)$ or $\strategies*(\state)$ for player $\reach$ or $\safe$, respectively. 
The existence of optimal (local) strategies in a matrix game for both players is guaranteed by Nash's Theorem \cite{nashEquilibriumPointsNperson1950a,nashNonCooperativeGames1951a}. The payoff that is attainable with an optimal local strategy is called \emph{value} that we denote by $\val(\mathsf{Z}_{\valR})$ for a matrix game $\mathsf{Z}_{\valR}$. It can be calculated using linear programming (e.g., \cite{hillierIntroductionOperationsResearch2010}, see \ifarxivelse{\cref{apx:LPforMatrixGames}}{\cite[\cref{apx:LPforMatrixGames}]{arxivversion}}).%

\para{End Components} 
A non-empty set of states $\endComponent \subseteq \states$ is called an \emph{end component} \newtarget{def:ec}{(\EC)} if (i)~there exists a pair of strategies $(\strategy,\strategy*) \in \strategies \times \strategies*$ such that for each $\state \in \endComponent$ it holds that $\destination(\state,\strategy,\strategy*) \subseteq \endComponent$; and (ii)~for every pair of states $\state, \state<\prime> \in \endComponent$ there is a play $\state[0]\state[1]\cdots$ such that $\state[0] = \state$ and $\state[n] = \state<\prime>$ for some $n$, and for all $0 \leq i < n$, it holds $\state[i] \in \endComponent$ and $\state[i+1] \in \destination(\state[i], \strategy,\strategy*)$.

Intuitively, an \EC is a set of states where a play can stay forever under some pair of strategies. In other words, the players can cooperate to keep the play inside the \EC (this is the usual way to lift the definition of~\cite{alfarothesis} from MDP to games). Thus, we can compute \ECs in a \CSG by computing \ECs in the corresponding \MDP with both players unified, i.e.,\ every pair of actions is interpreted as an action in the \MDP.
Efficient algorithms for this exist \cite{baierPrinciplesModelChecking2008a, courcoubetisComplexityProbabilisticVerification1995, wijsEfficientGPUAlgorithms2016}. An \EC $\endComponent$ is called \newtarget{def:mec}{\emph{inclusion maximal}} (short maximal) if there exists no \EC $\endComponent<\prime>$ such that $\endComponent \subsetneq \endComponent<\prime>$.

\subsection{Value Iteration} \label{sec:value-iteration}
Value iteration (\VI, e.g.~\cite{chatterjeeValueIteration2008a}) assigns an initial value estimate to each state and then iteratively updates it. 
In classical \VI, which approximates the reachability value from below, the initial estimates are 1 for states in $\success$ and below the actual value otherwise, e.g.~$0$. Each iteration backpropagates the estimate by maximizing the expectation of the value player~$\reach$ can ensure with respect to the previous estimate.

Formally, we capture estimates as valuations, where a \emph{valuation}~$\valuation \colon \states \to [0,1]$ is a function mapping each state~$\state$ to a real number 
representing the (approximate or true) value of the state.
For two valuations $\valuation, \valuation<\prime>$, we write $\valuation \leq \valuation<\prime>$ if $\valuation(\state) \leq \valuation<\prime>(\state)$ for every $\state \in \states$. 

To compute the expected value at a state~$\state$, the matrix game $\mathsf{Z}_{\valuation}(\state)$ has to be solved, meaning its value, $\val(\mathsf{Z}_{\valuation}(\state))$, has to be estimated. This computation is, especially in the turn-based setting, also referred to as \emph{Bellman update}. Formally,
\begin{align*}
\val(\mathsf{Z}_{\valuation}(\state)) \coloneqq \newtarget{def:pre}{\preop}(\valuation)(\state)\coloneqq \adjustlimits\sup_{\strategy \in \strategies(\state)} \inf_{\strategy* \in \strategies*(\state)} \preop(\valuation)(\state, \strategy, \strategy*),
\end{align*}
where $\preop(\valuation){}(\state, \strategy, \strategy*) \coloneqq$
\begin{align*}
	\smashoperator[l]{\sum_{(\action[\reach],\action[\safe]) \in \actions}} \sum_{\state<\prime> \in
		\states} 
	\strategy[\reach](\action[\reach]) \cdot
	\strategy[\safe](\action[\safe])\cdot
	\transitions(\state,\action[\reach],\action[\safe])(\state<\prime>) \cdot
	\valuation(\state<\prime>).
\end{align*}

\para{Convergent Under-approximation}
We recall \emph{VI from below} as in~\cite{Chatterjee2012}:
starting from the initial valuation $\lowerBound<0>$, we perform the Bellman update on every state to obtain a new valuation.
We denote by $\lowerBound<k>$ the valuation obtained in the $k$-th iteration. Formally:
\begin{align} 
   &\lowerBound<0>(\state) \coloneqq 
\begin{cases}
1,\; \mathsf{if} \ \state \in  \success;\\
0,\; \mathsf{else},
\end{cases} \hspace{1cm}\lowerBound<k+1>(\state) \coloneqq 
   \preop(\lowerBound<k>)(\state).
\end{align}

Since $\success$ and $\winning$ are absorbing, for all $k\in\NN$ we have $\lowerBound<k>(\state) = 1$ for all $\state \in \success$, and $\lowerBound<k>(\state) = 0$ for all $\state \in \winning$.
The updated valuation, i.e.\ $\lowerBound<k+1>(\state)$, is computed by solving the corresponding matrix game. 
\begin{theorem}[VI converges from below \protect{\cite[Thm.~1]{Alfaro2004}}]\label{theo:convergentUnderApprox}
VI from below converges to the value, i.e.  $\displaystyle{\lim_{k \to \infty}} \lowerBound<k> = \valR$.
\end{theorem}

\para{Bounded Value Iteration}

\begin{algorithm}[tb]\newtarget{alg:bvi}
	\usetikzlibrary{fit,calc}
	\newcommand{\boxit}[2]{
		\tikz[remember picture,overlay] \node (A) {};\ignorespaces
		\tikz[remember picture,overlay]{\node[yshift=4.5pt,fill=#1,opacity=.3,fit={($(A)+(0,0.15\baselineskip)$)($(A)+(.78\linewidth,-{#2}\baselineskip - 0.25\baselineskip)$)}] {};}\ignorespaces
	}
	
	\begin{algorithmic}[1]
		\Function{BVI}{\CSG $\game$, threshold $\varepsilon > 0$}
		\State $\winning \gets \set{ \state \in \states \mid \valR(\state) = 0}$ \Comment{Winning region for $\safe$}
		\State $\lowerBound<0>,  \upperBound<0>$ initialized by Eq. (1) and (2), respectively 
		\State $\mecs \gets \procname{FIND\_MECs}(\game)$ \Comment{Find all \MECs in the game}
		\State $k \gets 0$ 
		\DoUntil
		\For{$\state \in \states$}\Comment{Standard Bellman update of both bound}
		\State $\lowerBound<k+1>(\state) \gets \preop(\lowerBound<k>)(\state)$
		\State $\upperBound<k+1>(\state) \gets \preop(\upperBound<k>)(\state)$\label{alg:bellman}
		\EndFor
		\For{$\endComponent \in \mecs$}
		\State \boxit{yellow}{0.1}$\upperBound<k+1> \gets$  \DEFLATEALG$(\game, \upperBound<k+1>, \endComponent)$ \label{alg:deflate}
		\EndFor
		\State $k \gets k+1$
		\EndDoUntil{$\upperBound<k+1> - \lowerBound<k+1> \leq \varepsilon$}
		\EndFunction
	\end{algorithmic}
	\caption{Bounded value iteration procedure for CSGs.}
	\label{alg:bvi}
\end{algorithm}

While \cref{theo:convergentUnderApprox} proves that \VI from below converges in the limit, this limit may not be reached in finitely many steps and can be irrational~\cite{Alfaro2004}. Thus, we do not know when to stop the algorithm to guarantee certain precision of the approximation, as there is no practical bound how close any valuation $\lowerBound<k>$ is to the actual value.
We merely have the worst-case bound of~\cite{hansenComplexitySolvingReachability2014}: Running for a number of iterations that is doubly-exponential in the number of states allows to conclude that the lower bound is $\varepsilon$-close to the value.

To obtain a practical stopping criterion, we use the approach of Bounded Value Iteration (BVI), shown in \cref{alg:bvi}.
In addition to the lower bound $\lowerBound$, it maintains an upper bound on the value $\upperBound$ that is meant to converge to the value from above.
Na\"ively, this upper bound is defined as follows:
\begin{align} 
	\upperBound<0>(\state) \coloneqq \begin{cases}
		0,\; \mathsf{if} \ \state \in  \winning;\\
		1,\; \mathsf{else},
	\end{cases} & \hspace{0.1cm} &
	\upperBound<k+1>(\state) \coloneqq \preop(\upperBound<k>)(\state). \label{eq:upperBoundRec}
\end{align}
Given a precision $\varepsilon>0$, the algorithm terminates once the under- and the over-approximations are $\varepsilon$-close, i.e.,\ when both approximations are at most $\varepsilon$-away from the actual value.
However, applying Bellman updates does not suffice for the over-approximation to converge in the presence of \ECs,
as the following example shows:

\begin{example}[Non-convergent Over-approximations]\label{ex:potentialProblemECs}
	Consider the \CSG \hideorrun in \Cref{fig:single_state_example}. To compute $\upperBound<k+1>(\state[\mathsf{hide}])$ according to \Cref{eq:upperBoundRec}, in each iteration we solve the matrix game $\matrixGameMatrix[\upperBound<k>](\state[\mathsf{hide}])$ is given by \cref{eq:matrixGameMatrixVr} where the unknown $\valR$ is replaced by $\upperBound<k>$, i.e.:
	\begin{align}\label{eq:matrixGameMatrix}
		\matrixGameMatrix[\upperBound<k>](\state[\mathsf{hide}]) = \begin{blockarray}{ccc}
			\scriptstyle \mathsf{throw} & \scriptstyle \mathsf{wait}\\
			\begin{block}{(cc)c}
				0 & \frac{1}{3} \cdot \upperBound<k>(\state[\mathsf{hide}]) + \frac{1}{3}  & \scriptstyle \mathsf{run}\\
				1 & \upperBound<k>(\state[\mathsf{hide}]) & \scriptstyle \mathsf{hide}\\
			\end{block}
		\end{blockarray}
	\end{align}\label{eq:matrixGameTrivialEC}
\cref{tab:VI_game_1} shows the updates of the lower and upper bounds, $\lowerBound<k>(\state[\mathsf{hide}])$ and $\upperBound<k>(\state[\mathsf{hide}])$, respectively. While the lower bound converges to 0.5, the upper bound stays at 1.
This is because for player~$\reach$, action $\mathsf{hide}$ always ``promises'' a valuation of 1 in the next step, as the lower row of the matrix yields 1 for all player~$\safe$ strategies.
\end{example}

\begin{table}[b]
	\centering
	\caption{\BVI for the game $\hideorrun$ (\Cref{fig:single_state_example}), where the over-approximations do not converge.}\label{tab:VI_game_1}
	\begin{tabular}{cccccc}\toprule
		$k$ & 0 & 1 & 2 & $\cdots$ & $\infty$ \\
		\midrule
		$\lowerBound<k>(\state[\mathsf{hide}])$ & 0.0 & 0.25 & 0.36 & $\cdots$ & 0.5\\ [0.13cm]
		$\upperBound<k>(\state[\mathsf{hide}])$ & 1.0 & 1.0 & 1.0 & $\cdots$ & 1.0\\ 
		\bottomrule
	\end{tabular} 
\end{table} 

\section{Convergent Over-approximation: Overview}\label{sec:3-vi-title}
Here, we describe the structure of our solution.
\Cref{ex:potentialProblemECs} shows that the na\"ive definition of the over-approximation need not converge to the true value. 
In particular, in the presence of \ECs, Bellman updates do not have a unique fixpoint.
Thus, our goal is to define a function~\procname{DEFLATE} (usage highlighted in \Cref{alg:bvi}, definition in \Cref{alg:deflate_mecs}) that, intuitively, decreases the \enquote{bloated} upper bounds inside each \EC to a realistic value substantiated by a value promised \emph{outside} of this \EC.
Formally, we ensure that \cref{alg:bvi} produces a monotonically decreasing sequence of valuations (i)~over-approximating the reachability value and (ii)~converging to it in the limit. 
This idea has been successfully applied for \TSGs~\cite{eisentrautValueIterationSimple2022, kretinskyStoppingCriteriaValue2023a}.

\begin{remark}[Inflating for Safety]
	Since over-approximations need not converge for $\valR$, dually under-approximations need not converge for $\valS$ (as is the case in \TSGs, see \cite{kretinskyStoppingCriteriaValue2023a}). 
	Thus, to directly solve a safety game, one needs an inflating operation dual to deflating.
	As described when introducing the objectives, we take the conceptually easier route of reducing everything to maximizing reachability objectives.
\end{remark}

We proceed in two steps. First, in \Cref{sec:3-characterize} we prove that indeed \ECs are the source of non-convergence, in particular what we call \emph{Bloated End Components}. 
Intuitively, these are \ECs where at each state both players prefer local strategies which all successor states belong to the \EC.
Second, in \Cref{sec:3-resolve-BEC}, we define the \procname{DEFLATE} algorithm, which essentially ensures that we focus on player $\reach$ strategies that do not make the game stuck in a \EC but rather progress towards the target. To this end, we lift the notion of \emph{best exit} from \TSGs~\cite[Definition 3]{eisentrautValueIterationSimple2022} to \CSGs. 

\section{The Core of the Problem: Characterizing Bloated End Components}\label{sec:3-characterize}

A locally optimal strategy of Player $\reach$ does not coincide with an optimal global strategy of Player $\reach$ because the latter must eventually leave \ECs, while the former is under the illusion that staying is optimal.
Thus, in this section, we want to find properties that local strategies (of both players) must fulfill in order to leave an \EC in a way that is globally optimal. 
Thus, since this section mainly concerns \emph{local} strategies, we use the word strategy to speak about local strategies and explicitly make clear when we talk about global ones.

\begin{remark}
	Throughout the technical sections and the appendix, we always fix a \CSG $\exampleGame$.
\end{remark}

\subsection{Convergence without ECs}
As a first step, we prove that \ECs are the only source of non-convergence, and without them, the na\"ive \BVI using only Bellman updates converges.

\begin{restatable}[Convergence without \ECs~--- Proof in 
	\ifarxivelse{\cref{apx:proof_no_ecs}}{\cite[\cref{apx:proof_no_ecs}]{arxivversion}}]{theorem}{theoBVInoEC}\label{theo:BVInoEC}
	Let $\game$ be a \CSG where all \ECs are trivial, i.e. for every \EC $\endComponent$ we have $\endComponent \subseteq \winning \cup \success$.
	Then, the over-approximation using only \Cref{eq:upperBoundRec} converges, i.e.  $\displaystyle{\lim_{k\to\infty}} \upperBound<k> = \valR$.
\end{restatable}
\begin{proof}[Proof sketch]
	This proof is an extension of the proof of~\cite[Theorem 1]{eisentrautValueIterationSimple2022} for turn-based games to the concurrent setting.
	The underlying idea is the same, and can be briefly summarized as follows:
	We assume towards a contradiction that $\displaystyle{\lim_{k\to\infty}} \upperBound<k> \eqqcolon \upperBound<\star> \neq \valR$, and find a set $\ecStates$ that maximizes the difference between upper bound and value.
	We show that every pair of strategies leaving the set $\ecStates$ decreases the difference $\upperBound<\star> - \valR$.
	However, $\valR$ and $\upperBound<\star>$ are fixpoints of the Bellman update, from \cite[Theorem~1]{Alfaro2004} and \cref{lem:Ustar-is-fixpoint}, respectively.
	Consequently, optimal strategies need to remain in the set.
	However, in the absence of \ECs, optimal strategies have to leave the set, which yields a contradiction and proves that $\upperBound<\star> = \valR$.
	
	The key difference to the proof of~\cite[Theorem 1]{eisentrautValueIterationSimple2022} is that we cannot argue about actions anymore, but have to consider mixed strategies.
	This significantly complicates notation.
	Additionally, and more importantly, the former proof crucially relied on the fact that for a state of Player $\reach$, we know that its valuation is at least as large as that of any action, and dually for a state of Player $\safe$, its valuation is at most as large as that of any action.
	In the concurrent setting, this is not true. The optimal strategies need not be maximizing nor minimizing the valuation and, moreover, they can be maximizing for one valuation and minimizing for another.
	Thus, we found a more general, and in fact simpler, way of proving that \enquote{no state in $\ecStates$ can depend on the outside}~\cite[Statement 5]{eisentrautValueIterationSimple2022} and deriving the contradiction.	
\end{proof}

Interestingly, not all \ECs cause non-convergence of \Cref{eq:upperBoundRec} as the following example illustrates.
\begin{example}[Unproblematic \EC]\label{ex:convergenceWithEC}
	We modify the \CSG \hideorrun (\Cref{fig:single_state_example}) such that the matrix game played at $\state[\mathsf{hide}]$ is $\matrixGameMatrix[\upperBound<k>]<\prime>(\state[\mathsf{hide}])$ below.
	\begin{align*}
		\matrixGameMatrix[\upperBound<k>]<\prime>(\state[\mathsf{hide}]) = \begin{blockarray}{ccc}
			\scriptstyle \mathsf{throw} & \scriptstyle \mathsf{wait}\\
			\begin{block}{(cc)c}
				1 & \frac{1}{3} \cdot \upperBound<k>(\state[\mathsf{hide}]) + \frac{1}{3}  & \scriptstyle \mathsf{run}\\
				0 & \upperBound<k>(\state[\mathsf{hide}]) & \scriptstyle \mathsf{hide}\\
			\end{block}
		\end{blockarray}
	\end{align*}
	The difference is that $\matrixGameMatrix[\upperBound<k>]<\prime>(\state[\mathsf{hide}])(\mathsf{run},\mathsf{throw}) = 1$ and $\matrixGameMatrix[\upperBound<k>]<\prime>(\state[\mathsf{hide}])(\mathsf{hide},\mathsf{throw}) = 0$, switching the values as compared to the original \CSG (see \Cref{eq:matrixGameMatrix}).
	Here, both bounds converge to 0.5 despite the presence of the~\EC $\{\state[\mathsf{hide}]\}$, as shown in \Cref{tab:VI_converges_game_2}.
\end{example}
\begin{table}[tb]
	\centering
	\caption{\BVI for the \CSG in \Cref{ex:convergenceWithEC}, where the over-approximations converge.}\label{tab:VI_converges_game_2}
	\begin{tabular}{ccccccc}\toprule
		$k$ & 0 & 1 & 2 & 3 & $\cdots$ & $\infty$ \\
		\midrule
		$\lowerBound<k>(\state[\mathsf{hide}])$ & 0.0 & $\frac{1}{3}$ & $\frac{4}{9}$ & 0.4815 & $\cdots$ & 0.5\\ [0.1cm]
		$\upperBound<k>(\state[\mathsf{hide}])$ & 1.0 & $\frac{2}{3}$ & $\frac{5}{9}$ &0.5185 & $\cdots$ & 0.5\\ 
		\bottomrule
	\end{tabular} 
\end{table}

\subsection{Towards Characterizing Bloated End Components}
\para{Intuition}
In \Cref{ex:convergenceWithEC}, the best strategy of Player $\reach$ leaves the \EC almost surely against all counter-strategies of Player $\safe$, and hence \BVI converges. In contrast, in \Cref{ex:potentialProblemECs}, the best strategy of Player $\reach$ is one where Player $\safe$ has a counter-strategy that forces the play to stay inside the \EC; this causes non-convergence.
Generalizing these ideas, we see that a problem occurs if Player $\reach$ has a strategy that is locally optimal but non-leaving, i.e.\ Player $\safe$ has a counter-strategy that keeps the play inside an $\EC$.

\para{Outline}
We formalize these ideas in the following definitions:
First, \Cref{def:dominant-set-of-strats-main-body} formalizes optimal (local) strategies using \emph{weakly dominant strategies} in matrix games, extending the standard definition (e.g.~\cite{maschlerGameTheory2020}) %
to sets of strategies.
This extension is not straightforward, and there are several technical intricacies that we comment on.
Next, \Cref{def:leave-and-stay-strats} captures leaving and staying strategies.
We differentiate strategies that are leaving (irrespective of the opponent's strategy), staying (irrespective of the opponent's strategy), and non-leaving (where there exists an opponent's strategy that leads to staying) with respect to a given set of states.
Based on this, we formally describe hazardous strategies in \cref{def:badAction}, which are (locally) optimal strategies of Player $\reach$ that are non-leaving; additionally, to be problematic for convergence, they are better than all leaving strategies. 
Using these, we can precisely characterize the Bloated End Components (\BECs, \Cref{defBec}) that cause non-convergence.

\begin{remark}[Additional Challenges Compared to \TSGs]
	The core problem is the same as in \TSGs: Player $\reach$ is under the illusion that staying inside an \EC yields a better valuation than leaving.
	However, in \TSGs, the definitions of optimality and leaving are straightforward, since every state belongs to a single player and pure strategies are optimal; the definitions of hazardous and trapping strategies are not even necessary. 
	In contrast, the definitions in \CSGs are technically involved, as we have to take into account the interaction of the players and the possibility of optimal mixed strategies.
	In particular,~\cite{eisentrautStoppingCriteriaValue2019} defined a straightforward extension of leaving based only on actions, not strategies. This is incorrect, as we demonstrate in \ifarxivelse{\cref{counter-example:eisentraut}}{\cite[\cref{counter-example:eisentraut}]{arxivversion}}. 
\end{remark}

\begin{definition}[Dominating Sets of Strategies]
	\newtarget{def:dominant}{}\label{def:dominant-set-of-strats-main-body}
	Let $\valuation$ be a valuation, $\state \in \states$ a state, $\strategies[1], \strategies[2], \strategies[\reach]<\prime> \subseteq \strategies(\state)$ and $\mathcal{S}_1, \mathcal{S}_2, \strategies[\safe]<\prime> \subseteq \strategies[\safe](\state)$ sets of local strategies. 
	We now define two notions of domination for sets of strategies, namely \emph{weak domination} and being \emph{not worse}. Both of these depend on the player.
	
	\noindent\textit{Definition for Player $\reach$:}
	We write $\strategies[2] \dominates_{\valuation, \strategies[\safe]<\prime>} \strategies[1]$ to denote that
	$\strategies[1]$ \emph{weakly dominates} $\strategies[2]$ 
	under the set of counter-strategies $\strategies[\safe]<\prime>$ with respect to $\valuation$.
	Formally, $\exists \strategy[1] \in \strategies[1]. \forall \strategy[2] \in \strategies[2]:$
	\begin{enumerate}
		\item[(i)] 	$\inf_{\strategy* \in \strategies[\safe]<\prime>}\preop(\valuation)(\state,\strategy[2], \strategy*) \leq \inf_{\strategy* \in \strategies[\safe]<\prime>}\preop(\valuation)(\state, \strategy[1], \strategy*)$, and
		\item[(ii)] $\exists \sigma^{\prime} \in \strategies[\safe]<\prime>$ such that $\preop(\valuation)(\state, \strategy[2], \sigma^{\prime}) < \preop(\valuation)(\state, \strategy[1], \sigma^{\prime})$.
	\end{enumerate}
	If only Condition (i) is satisfied, we write $\strategies[2] \dominateseq_{\valuation, \strategies[\safe]<\prime>} \strategies[1]$ to denote that the set $\strategies[1]$ is \emph{not worse} than $\strategies[2]$ under $\strategies[\safe]<\prime>$ with respect to $\valuation$.

	\noindent \textit{Definition for Player $\safe$:}
	Dually, we write  $\mathcal{S}_2 \dominates_{\valuation, \strategies<\prime>} \mathcal{S}_1$ to denote that set $\mathcal{S}_1$ \emph{weakly dominates}  $\mathcal{S}_2$ under $ \strategies<\prime>$ with respect to $\valuation$.
	Formally, $\exists \sigma_1 \in \mathcal{S}_1. \forall \sigma_2 \in \mathcal{S}_2:$
	\begin{enumerate}
		\item[(i)] 	$\sup_{\strategy \in \strategies<\prime>}\preop(\valuation)(\state,\strategy, \sigma_2) \geq \sup_{\strategy \in \strategies<\prime>}\preop(\valuation)(\state, \strategy, \sigma_1)$, and
		\item[(ii)] $\exists \strategy<\prime> \in \strategies<\prime>$ such that $\preop(\valuation)(\state, \strategy<\prime>, \sigma_2) > \preop(\valuation)(\state, \strategy<\prime>, \sigma_1)$.
	\end{enumerate}

	If only Condition (i) is satisfied, we write $\mathcal{S}_2 \dominateseq_{\valuation, \strategies<\prime>} \mathcal{S}_1$ to denote that the set $\mathcal{S}_1$ is \emph{not worse} than $\mathcal{S}_2$ under $\strategies<\prime>$ with respect to $\valuation$.

\end{definition}

\begin{example}[Dominating Sets of Strategies]\label{ex:4-dominating-sets-of-strats}
	Consider the matrix game defined in \Cref{eq:matrixGameMatrix} and the valuation $\upperBound<k>(\state[\mathsf{hide}]) = \upperBound<k>(\state[\mathsf{home}]) = 1$ and $\upperBound<k>(\state[\mathsf{wet}]) =0$. 
	Here, for Player $\reach$, the pure strategy $\{\mathsf{hide} \mapsto 1\}$ dominates the pure strategy $\{\mathsf{run} \mapsto 1\}$:
	\[
	\{\mathsf{run} \mapsto 1\} \dominates_{\upperBound<k>, \strategies*(\state[\mathsf{hide}])} \{\mathsf{hide} \mapsto 1\}.
	\]
	This is because when Player $\safe$ throws the ball, hiding yields 1 while running yields 0.
	Note that this is in fact independent of the valuation, so it also holds for $\valR$.
	
	Let $\mathsf{RunPositive} \coloneqq \{(\mathsf{run} \mapsto \varepsilon, \mathsf{hide} \mapsto 1-\varepsilon) \mid \varepsilon>0\}$ be the set of all strategies that put positive probability on running. We have 
	\[
	\mathsf{RunPositive} \dominates_{\upperBound<k>, \strategies*(\state[\mathsf{hide}])} \{\mathsf{hide} \mapsto 1\}.
	\]
	Again, this is true even when using $\valR$ as valuation.
	Note that we have this weak dominance even though the supremum over the set yields the optimum valuation, namely
	$	\sup_{\rho\in\mathsf{RunPositive}} 
	\inf_{\strategy* \in \strategies*(\state[\mathsf{hide}])}\preop(\upperBound<k>)(\state,\strategy, \strategy*) = 1$. This exemplifies the strictness of our notion of dominance.
	It is crucial that our notion of domination can distinguish these sets of strategies:
	The set $\mathsf{RunPostive}$ contains all strategies that leave the \EC. However, none of them is optimal (even though the supremum over all of them is), which is exactly the reason why \VI chooses the staying strategy $\{\mathsf{hide} \mapsto 1\}$ for updating the valuation, and thus is stuck.
\end{example}

We remark on several technicalities of \cref{def:dominant-set-of-strats-main-body}:
\begin{itemize}
	\item \enquote{Weak} dominance: The term \enquote{weak} might be misleading. 
	We choose to use the word for consistency with~\cite[Def. 4.12]{maschlerGameTheory2020}. 
	There, \emph{weak} domination concerns Condition (ii), only requiring that there \emph{exists} a counter-strategy where the inequality is strict; strict domination requires Condition~(ii) \emph{for all} counter-strategies.
	One might be tempted to use weak domination to denote what we call \enquote{not worse}, i.e.\ only require that there is a strategy in the first set that has optimal valuation at least as good as all in the other set; or to think it only refers to the numerical comparators, e.g.\ $\geq$ and $>$ (as is sometimes the case when only comparing single strategies).
	\item Set-related challenges: The definition is challenging since we cannot speak about actions, but have to consider sets of --- possible mixed --- strategies. The exact quantification of the strategies is relevant. Further, it depends not only on the two sets we are comparing, but also on the counter-strategies of the opponent.
	Thus, we provide the definition explicitly for both players, to avoid confusion that could arise from just saying that they are analogous.
	\item All-quantification instead of optima: The definition uses all-quantification instead of optima. 
	Concretely, weak dominance for Player $\reach$ uses $\forall \strategy[2] \in \strategies[2]$ instead of writing $\sup_{\strategy[2] \in \strategies[2]}$.
	The latter definition cannot sufficiently distinguish sets of strategies, since the supremum of a set need not be contained in it, as we exemplified in \cref{ex:4-dominating-sets-of-strats}.
	This fact is extremely important, as in the proof of \cref{thm:nonConvImpliesBEC}, we pick the maximum from a set of strategies, and the existence of this maximum is guaranteed only because of the correct definition of dominance.
	\item Locally optimal strategies are not worse than any other: Formally, this claim is that for all locally optimal strategy $\rho\in\strategies(s)$ with respect to $\valuation$, we have $\strategies(s)\dominateseq_{\valuation,\strategies*(s)} \{\rho\mapsto 1\}$ (and dually for Player $\safe$).
	This is immediate from \cref{def:dominant-set-of-strats-main-body}, since a locally optimal strategy maximizes $\inf_{\strategy* \in \strategies[\safe]<\prime>}\preop(\valuation)(\state,\strategy, \strategy*)$, and thus satisfies Condition (i) when compared to all other strategies.
	We will use this fact throughout the paper.
	\item Notation: When the valuation is clear from the context, we omit it for the sake of readability. Further, if we say that a strategy $\strategy[1]$ weakly dominates another strategy $\strategy[2]$ with respect to a counter-strategy $\strategy*$, then we mean that $\{\strategy[2]\}\dominates_{\{\strategy*\}}\{\strategy[1]\}$.
\end{itemize}

We prove a lemma about the relation of weak domination and not being worse that is useful and instructive.
The proof in \Cref{apx:convergence_no_bad_ecs} works by straightforward unfolding of definitions and rewriting.
\begin{restatable}[Negating Weak Domination ---  Proof in \Cref{apx:convergence_no_bad_ecs}]{lemma}{negateWeakDom}\label{lem:negate-weak-dominance}
	Let $\valuation$ be a valuation, $\state \in \states$ a state, $\strategies[1], \strategies[2], \strategies[\reach]<\prime> \subseteq \strategies(\state)$ and $\mathcal{S}_1, \mathcal{S}_2, \strategies[\safe]<\prime> \subseteq \strategies[\safe](\state)$ sets of local strategies. 
	
	If for some sets of strategies we do \emph{not} have $\strategies[2] \dominates_{\valuation, \strategies[\safe]<\prime>} \strategies[1]$, then we have $\strategies[1] \dominateseq_{\valuation, \strategies[\safe]<\prime>} \strategies[2]$.
	Analogously, not $\mathcal{S}_2 \dominates_{\valuation, \strategies<\prime>} \mathcal{S}_1$ implies $\mathcal{S}_1 \dominateseq_{\valuation, \strategies<\prime>} \mathcal{S}_2$.
\end{restatable}

To complete our intuitive understanding of the definition of domination, we point out a connection to the standard definition of weak domination: \enquote{A rational player does not use a dominated strategy.}~\cite[Asm. 4.13]{maschlerGameTheory2020}
If a strategy is not dominated, by \cref{lem:negate-weak-dominance} it is not worse than any other strategy. This is exactly what we argued above: Locally optimal strategies are not worse than any other.
We often use this fact throughout the paper.

Next, we formally define \emph{leaving and staying strategies}. 
Given a set of states, a leaving strategy ensures that the set of successor states contains states outside the given set of states for all given counter-strategies of the opponent player. 
A strategy is staying if all successor states belong to the given set of states for all given counter-strategies of the opponent player. 
Note that a strategy can be neither leaving nor staying, if the set is exited for some, but not all counter-strategies of the opponent.

\begin{definition}[Leaving and Staying Strategies]\newtarget{def:leavingStratR}{}\newtarget{def:leavingStratS}{}\newtarget{def:allStayingStratR}{}\newtarget{def:allStayingStratS}{}\newtarget{def:stayingStratR}{}\newtarget{def:stayingStratS}{}
	\label{def:leave-and-stay-strats}
	Consider a set of states $\ecStates \subseteq \states$ and a state $\state \in \ecStates$. 
	Let  $\strategies<\prime> \subseteq \strategies(\state)$ and $\strategies[\safe]<\prime> \subseteq \strategies*(\state)$ be sets of strategies of Player $\reach$ and $\safe$, respectively. The set of (local) \emph{leaving strategies} for Player $\reach$ with respect to $\strategies[\safe]<\prime>$, is given by 
	\begin{align*}
		&\leaveStratsR(\strategies[\safe]<\prime>,\ecStates,\state) \coloneqq \{ \strategy \in \strategies(\state) \mid \forall \strategy* \in \strategies[\safe]<\prime>.(\state, \strategy, \strategy*) \; \leaves \; \ecStates\},\\
		&\text{and for Player~$\safe$ with respect to $\strategies<\prime>$ by}\\
		&\leaveStratsS(\strategies<\prime>, \ecStates, \state)\coloneqq \{ \strategy* \in \strategies*(\state) \mid \forall \strategy \in \strategies<\prime>.(\state, \strategy, \strategy*) \; \leaves \; \ecStates\}.
	\end{align*}
	A strategy that is not leaving is called \emph{non-leaving}. The set of all non-leaving Player $\reach$ strategies is denoted by $\stayStratsR(\strategies[\safe]<\prime>, \ecStates, \state)$ (or $\stayStratsS$ for Player $\safe$).
	
	In contrast, the set of \emph{staying} strategies at a state $\state \in \ecStates$ for Player $\reach$ with respect to $\strategies[\safe]<\prime>$, is given by
	\begin{align*}
		&\allStayStratsR(\strategies[\safe]<\prime>,\ecStates,\state) \coloneqq \{ \strategy \in \strategies(\state) \mid \forall \strategy* \in \strategies[\safe]<\prime>.(\state, \strategy, \strategy*) \; \stays \; \ecStates\},\\
		&\text{and for Player~$\safe$ with respect to $\strategies<\prime>$ by}\\
		&\allStayStratsS(\strategies<\prime>, \ecStates, \state)\coloneqq \{ \strategy* \in \strategies*(\state) \mid \forall \strategy \in \strategies<\prime>.(\state, \strategy, \strategy*) \; \stays \; \ecStates\}.
	\end{align*}
\end{definition}

\para{Notation}
If we consider leaving (or staying) strategies with respect to all counter-strategies, then we omit the set of counter-strategies, i.e. instead of $\leaveStratsR(\strategies[\safe](\state),\ecStates,\state)$ (or $\allStayStratsR(\strategies[\safe](\state),\ecStates,\state)$) we write $\leaveStratsR(\ecStates,\state)$ (or $\allStayStratsR(\ecStates,\state)$). We use the same shorthand notion for leaving (staying) strategies of Player $\safe$.

We often speak about a leaving/staying pair of local strategies, so we provide the following shorthand notations: For a tuple $(\state, \strategy, \strategy*) \in \states \times \strategies \times \strategies*$, we say that $(\state, \strategy, \strategy*)\; \newtarget{def:leaves}{\leaves} \;\ecStates$ if and only if $\destination(\state, \strategy, \strategy*) \cap (\states \setminus \ecStates) \neq \emptyset$. Analogously, we say that $(\state, \strategy, \strategy*)\; \newtarget{def:stays}{\stays} \;\ecStates$ if and only if $\destination(\state, \strategy, \strategy*) \cap (\states \setminus \ecStates) = \emptyset$ (or, equivalently, $\destination(\state, \strategy, \strategy*) \subseteq \ecStates$).

\para{Intuition of Hazardous Strategies}
Using the definitions of dominance and leaving or staying, we can now classify strategies of Player $\reach$ that can lead to non-convergence.
Intuitively, a hazardous strategy is one that Player $\reach$ chooses, even though it can be staying for some counter-strategies.
Thus, such a strategy (i) is non-leaving (i.e.\ there exist counter-strategies that make it staying), and (ii) it is not worse than any other strategy so that Player $\reach$ may choose it.
Moreover, to be problematic for convergence, (iii) the strategy weakly dominates all leaving strategies, i.e.\ leaving strategies are not chosen for the update.

\begin{definition}[Hazardous Strategy]\label{def:badAction}
	Let $\ecStates \subseteq \states \setminus (\success \cup \winning)$ be as set of states, $\valuation$ a valuation, and $\state \in \ecStates$. A strategy $\strategy \in \strategies(\state)$ is called hazardous with respect to $\valuation$ if it satisfies:
	\begin{enumerate}
		\item[(i)] $\strategy \in \stayStratsR(\ecStates, \state)$,
		\item[(ii)] $\strategies(\state) \setminus \{\strategy\} \dominateseq_{\strategies*(\state)} \{\strategy\}$, and
		\item[(iii)] $\leaveStratsR(\ecStates, \state) \dominates_{\strategies*(\state)} \{\strategy\}$. %
	\end{enumerate}
	$\newtarget{def:badActs}{\badActs_{\valuation}(\ecStates, \state)}$ denotes the set of all hazardous strategies at state $\state$ with respect to a set of states $\ecStates$ and a valuation $\valuation$.
\end{definition}

We mention a corner case: In a state where Player $\reach$ possesses no leaving strategies, all optimal strategies are hazardous (note in particular that Condition (iii) is trivially satisfied, since the dominated set of strategies $\leaveStratsR(\ecStates, \state)$ is empty).

\begin{example}[Hazardous strategy]\label{ex:hazardousStrats}
	Consider again the matrix game defined in \Cref{eq:matrixGameMatrix} and the initial valuation $\upperBound<0>(\state[\mathsf{hide}]) = \upperBound<0>(\state[\mathsf{home}]) = 1$ and $\upperBound<0>(\state[\mathsf{wet}]) =0$. 
	The strategy $\strategy<\prime> \coloneqq \{\mathsf{hide} \mapsto 1\}$ is hazardous because:
	(i) It is non-leaving.
	(ii) It is an optimal strategy, i.e.\ it is not worse than any other strategy.
	(iii) It weakly dominates the set of all leaving strategies, see \cref{ex:4-dominating-sets-of-strats}. 
\end{example}

\begin{definition}[Bloated End Component (\BEC)]\label{defBec}
	An \EC $\ecStates \subseteq \states \setminus (\success \cup \winning)$ is called \emph{bloated end component} \newtarget{def:bec}{(BEC)} with respect to a valuation $\valuation$ if for all $\state \in \ecStates$ it holds that $\badActs_{\valuation}(\ecStates, \state) \neq \emptyset$.
\end{definition}

\begin{example}[Bloated End Component]\label{ex:bec}
	Consider the \CSG \hideorrun from \cref{fig:single_state_example}.
	As discussed in \cref{ex:hazardousStrats}, there exists a hazardous strategy in state $\state[\mathsf{hide}]$. Moreover, $\{\state[\mathsf{hide}]\}$ is an \EC, since under the pair of strategies that plays $\mathsf{hide}$ and $\mathsf{wait}$, the play stays in it.
	Consequently, $\{\state[\mathsf{hide}]\}$ is a \BEC and therefore \VI does not converge in this state, see \cref{ex:potentialProblemECs}.
\end{example}

We provide a lemma that captures the intuition of what it means to (not) be a \BEC, and that is also useful in several proofs:
\begin{restatable}[Negating Bloated ---  Proof in
	\ifarxivelse{\cref{apx:convergence_no_bad_ecs}}{\cite[\cref{apx:convergence_no_bad_ecs}]{arxivversion}}
	]{lemma}{negateBEC}\label{lem:negate-BEC}
	If an \EC $\ecStates \subseteq \states \setminus (\success \cup \winning)$ is not bloated for a valuation $\valuation$, then there exists a state $s\in\ecStates$ that has a locally optimal strategy that is leaving, formally $\exists \rho \in \leaveStratsR(\ecStates, \state). \strategies(s) \dominateseq_{\valuation,\strategies*(s)} \{\rho\}$.
\end{restatable}

\subsection{Convergence in the Absence of \BECs}
Now we can prove that \BECs indeed are the 
reason that \VI does not converge for over-approximations.

\begin{restatable}[Non-convergence implies \BECs~--- Proof in
	\ifarxivelse{\cref{apx:convergence_no_bad_ecs}}{\cite[\cref{apx:convergence_no_bad_ecs}]{arxivversion}} 
	]{theorem}{nonConvImpliesBEC}\label{thm:nonConvImpliesBEC}
Let $\mathsf{U}^\star: = \lim_{k \rightarrow \infty} \upperBound<k>$ be the limit of the na\"ive upper bound iteration (\cref{eq:upperBoundRec}) on the \CSG $\game$.
If \VI from above does not converge to the value in the limit, i.e.\ $\mathsf{U}^\star > \valR$, then the \CSG $\game$ contains a \BEC in $\states \setminus (\success \cup \winning)$ with respect to $\mathsf{U}^\star$.
\end{restatable}
\begin{proof}[Proof sketch]
	This proof builds on the proof of \cref{theo:BVInoEC}.
	There, we constructed a set $\mathcal{X}$ maximizing the difference between $\upperBound<\star>$ and $\valR$ and showed that if there is a pair of optimal strategies leaving $\mathcal{X}$, then we can derive a contradiction: The upper bound decreases, which contradicts the fact that it is a fixpoint.
	In the context of the other proof, that allowed us to show that without \ECs, \VI converges, because without \ECs it is impossible to have a set of states where all optimal strategies stay in that set.
	
	In the presence of \ECs, states can indeed have a positive difference between $\upperBound<\star>$ and $\valR$, see e.g.\ \cref{ex:potentialProblemECs}. 
	Our goal is to prove that at least one of these \ECs is bloated.
	Thus, we assume for contradiction that no \EC is bloated under $\upperBound<\star>$.
	Then, by \cref{lem:negate-BEC}, there is an optimal leaving strategy for Player $\reach$.
	Using that, we can repeat the argument from \cref{theo:BVInoEC}, showing that in this case $\upperBound<\star>$ would decrease. Again, this is a contradiction because it is a fixpoint of applying Bellman updates (\cref{lem:Ustar-is-fixpoint}).
	Thus, the initial assumption that no \EC is bloated is false, and we can conclude that there exists a \BEC.
\end{proof}

\begin{remark}[Relation to~\cite{eisentrautValueIterationSimple2022}]
	\cref{defBec} of \BEC is more general than the definition of BEC for \TSGs in~\cite[Definition 4]{eisentrautValueIterationSimple2022}.
	The differences are that in~\cite{eisentrautValueIterationSimple2022}, an \EC is only called bloated if it is bloated with respect to $\valR$, whereas we extended that definition to speak about a concrete valuation, similar to~\cite[Def.~3]{kretinskyStoppingCriteriaValue2023a}. 
	Further, the definition for \TSGs speaks about the best exit value, which is the optimum among the available actions; in contrast, in \CSGs the definition of best exit is technically involved and dependent on hazardous strategies, see \cref{def:exitVal}.
	Thus, our definition of \BEC does not analyze the exit values, but instead uses a fundamental analysis of the strategies.
\end{remark}

\para{Key Contribution} 
The key novelty of this section is the correct definition of \BEC that captures 
when \VI from above does not converge.  
We highlight that this definition contains many technical intricacies: 
Lifting the notion of an exiting action~\cite[Section 2.2]{eisentrautValueIterationSimple2022} from \TSGs to \CSGs requires considering sets of local strategies that leave against all opponent-strategies (\Cref{def:leave-and-stay-strats}), and considering the additional complication that strategies can be neither leaving nor staying. 
Further, the exact definition of dominance is very important, as the supremum over all leaving strategies can be a staying strategy, see \cref{ex:4-dominating-sets-of-strats} and the related discussion in the item \enquote{All-quantification instead of optima} after the example.

\section{Resolving Bloated End Components}\label{sec:3-resolve-BEC}
\subsection{Solution in TSGs}
We have identified \BECs as the cause of non-convergence.
Our method for fixing the over-approximation is again based on the ideas for \TSGs.
We explain the intuition of their solution:
Firstly, staying actions yield the valuation that is bloated; thus, we need an additional update of the over-approximation that depends only on leaving actions. The valuation to which we reduce the over-approximation is the \emph{best exit} from the \EC, which in \TSGs simply is the leaving action attaining the highest value over all states of Player $\reach$~\cite[Definition 3]{eisentrautValueIterationSimple2022}.
Secondly, not all states in an \EC need to have the same value, since Player $\safe$ can prevent Player $\reach$ from reaching the state that attains the best exit valuation.
Hence, an \EC is decomposed into parts that share the same value, called \emph{simple} \ECs (SECs) in~\cite[Definition 5]{eisentrautValueIterationSimple2022}.
Repeatedly finding these SECs and \emph{deflating} their valuation by setting it to the best exit from the SEC suffices for convergence in \TSGs.

When generalizing these ideas to \CSGs, we encounter the following problems:
Firstly, the definition of best exit is more involved, since staying and leaving depends not only on actions, but on strategies. 
Additionally, the supremum over all leaving strategies can be a non-leaving strategy, see \cref{ex:4-dominating-sets-of-strats}.
(This is also the reason why globally optimal strategies need not exist in \CSGs, see \cref{ex:2-no-opt-strat}).
This was the fundamental reason why the solution proposed in~\cite{eisentrautStoppingCriteriaValue2019} did not work, as it was based on actions.
Secondly, we need to decompose the \EC into parts. For this, we use a recursive approach, removing states that have been successfully deflated and checking whether there are further problematic states in the remainder of the \EC.

\para{Outline}
In \cref{sec:findingBestExit}, we develop a strategy-based definition of best exit (\Cref{def:bestExitVal}), which relies on identifying the \emph{trapping strategies} (\Cref{def:badCounterStrategy}) that Player $\safe$ uses to keep the play inside the \BEC; and the \emph{deflating strategies} (\cref{def:deflStrats}), the best response of Player $\reach$, namely the leaving strategies that should be played with arbitrarily small probability $\varepsilon$. %
In \cref{sec:findingMBECs}, we provide the \procname{FIND\_MBECs} algorithm that finds all maximal \BECs that are present in a given \MEC. A maximal \BEC is a set of states $\ecStates \subseteq \states \setminus (\success \cup \winning)$ that is a \BEC and there exists no another set of states $\ecStates<\prime> \subseteq \states \setminus (\success \cup \winning)$ that is a \BEC and $\ecStates \subsetneq  \ecStates<\prime>$.
Finally, \cref{sec:convBVI} provides the full deflating procedure for \CSGs in \cref{alg:deflate_mecs}.
In particular, it uses a recursive call to decompose a given \MEC and deflate all relevant parts of it.

\subsection{Defining the Best Exit}\label{sec:findingBestExit}
The key problem of a \BEC is that in all states, all leaving strategies of Player $\reach$ are dominated by hazardous strategies. Player $\safe$ can play a trapping strategy and thereby make the Bellman update self-dependent.
If the current valuation is too high, this prevents convergence.
Intuitively, the \enquote{pressure} inside the \BEC is too high, and we want to \enquote{deflate} it, by adjusting it to the pressure, i.e.\ valuation, outside of the \BEC.
Since non-trivial \BECs neither contain target states nor belong to the winning region of Player $\safe$, there has to exist a state in a \BEC where the supremum over leaving and staying strategies is equal.
To \enquote{equalize the pressure} between the \BEC and the rest of the states, we need to estimate the pressure outside the \BEC. 
To do so, we estimate the valuation attainable upon leaving the \BEC at every state of the \BEC, called exit value of the state. 
The best exit value is the maximum of all exit values.
Reducing the upper bounds of the states inside the \BEC to the best exit value, in case the best exit is smaller than the current valuation, \enquote{decreases the pressure}. 
However, since the valuations of the exiting strategies can depend on the valuations of the states that belong to the \BEC, this procedure may still only converge in the limit, already in \TSGs.
We provide an illustrative example:

\begin{example}[Deflating \BECs]\label{ex:deflateBEC}
	Consider again the \CSG \hideorrun (\cref{fig:single_state_example}). Under the initial valuation (see \Cref{eq:upperBoundRec}) the matrix game played at $\state[\mathsf{hide}]$ is given by 
		\begin{align*}
		\matrixGameMatrix[\upperBound<0>](\state[\mathsf{hide}]) = \begin{blockarray}{ccc}
			\scriptstyle \mathsf{throw} & \scriptstyle \mathsf{wait}\\
			\begin{block}{(cc)c}
				0 & \frac{2}{3}  & \scriptstyle \mathsf{run}\\
				1 & 1 & \scriptstyle \mathsf{hide}\\
			\end{block}
		\end{blockarray}
	\end{align*}
	Due to the hazardous strategy $\{\mathsf{hide} \mapsto 1\}$, the Bellman update cannot improve the initial upper bound of $\state[\mathsf{hide}]$, but remains at 1.
	However, by the same argument as in \cref{ex:2-no-opt-strat}, Player $\reach$ can use the strategy of running with a probability $\varepsilon>0$ that is arbitrarily small.
	This yields a value arbitrarily close to $\frac 2 3$.
	Consequently, we can deflate, i.e.\ decrease the upper bound of $\state[\mathsf{hide}]$ to $\frac{2}{3}$ which is the valuation attainable upon leaving the \BEC at $\state[\mathsf{hide}]$. After the Bellman update, the matrix game played at $\state[\mathsf{hide}]$ is given by
			\begin{align*}
		\matrixGameMatrix[\upperBound<1>](\state[\mathsf{hide}]) = \begin{blockarray}{ccc}
			\scriptstyle \mathsf{throw} & \scriptstyle \mathsf{wait}\\
			\begin{block}{(cc)c}
				0 & \frac{5}{9}  & \scriptstyle \mathsf{run}\\ [0.2em]
				1 & \frac{2}{3}& \scriptstyle \mathsf{hide}\\
			\end{block}
		\end{blockarray}
	\end{align*}
	 The strategy $\{\mathsf{hide} \mapsto 1\}$ is still hazardous, but we can deflate the upper bound of $\state[\mathsf{hide}]$ to $\frac{5}{9}$. 
	 By repeating these steps, the upper bound converges to $\frac{1}{2}$ in the limit, and only in the limit, similar to the lower bound in \cref{tab:VI_game_1}. 
\end{example}
\para{How can we find the best exit value in general?}
In \cref{ex:deflateBEC}, we used an argument about playing a leaving action with vanishingly small probability to figure out which entry in the matrix we choose for deflating.
We provide an alternative intuition:
Player $\reach$ plays a hazardous strategy most of the time. The best response of Player $\safe$ is to play a weakly dominant strategy that stays in the \EC, trapping the play.
Thus, we can ignore the other strategies of Player $\safe$ and consider only the columns of the matrix that correspond to \emph{trapping strategies}.
Now, we need to select a leaving strategy of Player $\reach$ which then is played with vanishingly low probability. Thus, we restrict the matrix further to use only rows corresponding to leaving strategies.
In the example, we end up with the top right entry.
We now formalize how we can construct this sub-matrix, called the exiting sub-game.

\begin{definition}[Trapping Strategy]\newtarget{def:counter}\label{def:badCounterStrategy}
	Let $\ecStates \subseteq \states \setminus (\success \cup \winning)$ be as set of states, $\valuation$ a valuation, and $\state \in \ecStates$. A strategy $\strategy[\safe] \in \strategies[\safe](\state)$ is called \emph{trapping strategy} if two conditions are satisfied:
	\begin{enumerate}
		\item[(i)] $\strategy[\safe] \in \argmin_{\strategy[\safe]<\prime> \in \strategies[\safe](\state)} \max_{\strategy \in \strategies(\state)}\preop(\valuation)(\state, \strategy, \strategy[\safe]<\prime>)$, and
		\item[(ii)] $\forall \strategy \in \badActs_{\valuation}(\ecStates,\state):(\state, \strategy, \strategy[\safe])\;\stays\;\ecStates$.
	\end{enumerate}
		$\counter_{\valuation}(\ecStates, \state)$ denotes the set of all trapping strategies at state $\state$ with respect to a set of states $\ecStates$ and a valuation $\valuation$.
\end{definition}

\begin{definition}[Deflating Strategies]\newtarget{def:deflStrats}\label{def:deflStrats}
	Let $\ecStates \subseteq \states \setminus (\success \cup \winning)$ be a set of states.
	A Player $\reach$ strategy $\strategy \in \strategies(\state)$ is called \emph{deflating} if two conditions are satisfied:
	\begin{enumerate}
		\item[(i)]$\exists \strategy* \in \counter_{\valuation}(\ecStates, \state)$ such that $(\state, \strategy, \strategy*) \; \leaves \; \ecStates$, and
		\item[(ii)] $\support(\strategy) \cap \bigcup_{\strategy<\prime>\in\badActs_{\valuation}(\ecStates, \state)} \support(\strategy<\prime>)= \emptyset$.
	\end{enumerate}
	$\deflStrats_{\valuation}(\ecStates, \state)$ denotes the set of all deflating strategies at state $\state$ with respect to a set of states $\ecStates$ and a valuation $\valuation$.
\end{definition}

\begin{definition}[Exiting sub-game]\label{def:exitingSG}\newtarget{def:exitingSubGame}
		Let $\ecStates \subseteq \states \setminus (\success \cup \winning)$ be a set of states, $s\in\ecStates$ a state, and $\valuation$ a valuation. Further, let $\mathsf{Z}_{\valuation}(\state)$ be the matrix game played at state $\state \in \ecStates$. If $\counter_{\valuation}(\ecStates,\state) \neq \emptyset$
		 then, the \emph{exiting sub-game} played at state $\state$, denoted by $\exitingSG(\state)$, is the matrix game where Player $\reach$ has the actions in $\bigcup_{\strategy \in \deflStrats_{\valuation}(\ecStates, \state)}\support(\strategy)$ 
		  and Player $\safe$ has the actions in $\bigcup_{\strategy* \in \counter_{\valuation}(\ecStates, \state)}\support(\strategy*)$. 
		  The value of the exiting sub-game is given by $\val(\exitingSG(\state)) \coloneqq \max(0, \sup_{\strategy \in \deflStrats_{\valuation}(\ecStates, \state)}\inf_{\strategy* \in \counter_{\valuation}(\ecStates, \state)}\preop(\valuation)(\state, \strategy, \strategy*))$.
\end{definition}

We explain how this exiting sub-game and its value are well-defined:
The set of deflating strategies can be empty, namely in a state which has no leaving strategies. In this case, the value of the exiting sub-game is the supremum over an empty set, i.e.\ the smallest possible value, commonly minus infinity and 0 in our case.
We highlight this possibility by explicitly taking the maximum of 0 and the value of the exiting subgame when we compute it.
If the set of deflating strategies is non-empty, then the set of trapping strategies necessarily is non-empty, too, since the Item (i) of \cref{def:deflStrats} requires existence of a trapping strategy.

\begin{definition}[Exit value]\label{def:exitVal}\newtarget{exitValDef}
	Let $\ecStates \subseteq \states \setminus (\success \cup \winning)$ be a set of states, $s\in\ecStates$ a state, and $\valuation$ an over-approximation.
	Then, the \emph{exit value} from $\ecStates$ attainable at state $\state$ is given by
	\begin{align*}
	&\exitVal[\valuation](\ecStates, \state) \coloneqq\\ 
	&\begin{cases}
		\adjustlimits\sup_{\strategy \in \strategies(\state)}\inf_{\strategy* \in \strategies*(\state)}\preop(\valuation)(\state,\strategy,\strategy*), &\text{ if }\badActs_{\valuation}(\ecStates, \state) = \emptyset \\
		&\quad\hspace{5pt}\lor\, \counter_{\valuation}(\ecStates, \state) = \emptyset;\\ \val(\exitingSG(\state)), &\text{ else.}
	\end{cases}
\end{align*}
\end{definition}

For \BECs that consists of more than one state, the exit values attainable at different states within the \BEC may differ. Consequently, to ensure that deflating does not reduce the value of any of the states in the \BEC below its actual value we estimate the exit values at each state of the \BEC, and finally select the maximal exit value, i.e. the best exit value, for deflating. 
\begin{definition}[Best Exit]\newtarget{defBestExit}\label{def:bestExitVal}\newtarget{def:bestExitVal}
	Let $\ecStates \subseteq (\states \setminus (\success \cup \winning))$ be a set of states and $\valuation$ a valuation.
	The \emph {best exit value} with respect to a set $\ecStates$ and a valuation $\valuation$ is given by
	\begin{align*}
		&\bestExitVal[\valuation](\ecStates) := \max_{\state \in \ecStates} \exitVal[\valuation](\ecStates, \state).
	\end{align*}
	The \emph{best exit}, denoted by $\bestExit[\valuation](\ecStates)$\newtarget{def:bestExits}, is a state obtaining $\bestExitVal[\valuation](\ecStates)$, and the set of all best exits is denoted by $\bestExits[\valuation](\ecStates)$\newtarget{def:bestExits}.
\end{definition}

\begin{lemma}[$\bestExitVal$ is sound]\label{lem:bestExitValSound}
	Let $\ecStates \subseteq \states \setminus (\success \cup \winning)$ be an \EC, and $\upperBound \in [0,1]^{\mid \states \mid}$ be a valuation with $\upperBound \geq \valR$.
	Then, for all states $\state\in\ecStates$, we have $\bestExitVal[\upperBound](\ecStates) \geq \valR(s)$.
\end{lemma}
	\begin{proof}
	This proof relies on the technical \ifarxivelse{\cref{lem:exitPossibleUnderVr} }{{arxivversion}} stating that
	\begin{itemize}
		\item[(i)] $\ecStates<\prime> \coloneqq \{\state \in \ecStates \mid \valR(\state) \leq \exitVal[\valR](\ecStates, \state)\} \neq \emptyset$, and
		\item[(ii)] $\max_{\state \in \ecStates<\prime>}\valR(\state) \geq \max_{\state \in \ecStates \setminus \ecStates<\prime>} \valR(\state)$.
	\end{itemize} 
	Let $\ecStates<\prime> \subseteq \ecStates$ be a set of states satisfying Condition (i).
	Further, choose $e \in \argmax_{t \in \ecStates<\prime>}\valR(t)$ as one of the exits from $\ecStates<\prime>$.
	By Item (ii) of \ifarxivelse{\cref{lem:exitPossibleUnderVr} }{{arxivversion}}, we have that for all $s\in \ecStates$: 
	$\valR(s) \leq \valR(e)$.
	It remains to show $\bestExitVal[\upperBound](\ecStates) \geq \valR(e)$.
	We conclude as follows:
	\begin{align*}
		\valR(e) &\leq \exitVal[\valR](\ecStates,e) \tag{Since $e\in\ecStates<\prime>$}\\
		&=  \bestExitVal[\valR](\ecStates) \tag{By the choice of $e$}\\
		&\leq  \exitVal[\upperBound](\ecStates,e) \tag{By \ifarxivelse{\cref{lem:eVUgeqeVV} }{{arxivversion}}}\\
		&\leq \bestExitVal[\upperBound](\ecStates). \tag{By \cref{def:bestExitVal}}
	\end{align*}	
	We want to highlight that the statement of \ifarxivelse{\cref{lem:eVUgeqeVV} }{{arxivversion}} is non-trivial and the proof is technically involved.
\end{proof}

\subsection{Finding Maximal BECs}\label{sec:findingMBECs} Since a \BEC might contain other \BECs, we want to find \emph{maximal} \BECs. A \BEC $\ecStates$ is maximal if there exists no \BEC $\ecStates<\prime>$ such that $\ecStates \subsetneq \ecStates<\prime>$. The existence of maximal \BECs is proven in \ifarxivelse{\cref{apx:soundness_deflate}}{\cite[App. C-C]{arxivversion}}.

Given a \CSG $\game$, a \MEC $\endComponent$ and the current upper bound estimate $\upperBound$,
\Cref{alg:find_bec} finds all maximal \BECs within $\endComponent$ as follows: The set $B$ contains all states for which hazardous strategies exist with respect to the set $\endComponent$. We distinguish two base cases. First, if $B$ is empty, then no \BEC can exist in $\endComponent$. Second, if $B$ contains all states contained in $\endComponent$, then $B$ is a maximal \BEC. In both cases, the algorithm returns the set $B$. If neither case applies, $B$ might contain multiple disjoint \MECs, which in turn might be \BECs. Therefore, for each \MEC contained in $B$, \procname{FIND\_MBECs} is called recursively.

\begin{algorithm}[tb]\newtarget{alg:find_bec}
	\begin{algorithmic}[1]
		\Function{\procname{FIND\_MBECs}}{CSG $\game$, \MEC $\endComponent$,  upper bound estimate $\upperBound$}
		\State $B \coloneqq \{\state \in \endComponent \mid \badActs_{\upperBound}(\endComponent, \state) \neq \emptyset\}$ 
		\If{$B == \endComponent$ or $B == \emptyset$}\label{findBEC:baseCase} 
		\Return $\{B\}$ %
		\Else
		\State \Return $\{\endComponent<\prime> \in  \procname{FIND\_MECs}(\game, B) \mid \procname{FIND\_MBECs}(\game, \endComponent<\prime>, \upperBound)\}$\label{findMBECs:findallMaxMECS}
		\EndIf
		\EndFunction
	\end{algorithmic}
	\caption{Algorithm for finding maximal \BECs.}
	\label{alg:find_bec}
\end{algorithm}

\begin{restatable}[\procname{FIND\_MBECs} is correct--- Proof in
	\ifarxivelse{\cref{apx:soundness_deflate}}{\cite[App. C-C]{arxivversion}}] 
	{lemma}{findingMBECScorrect}\label{lem_findingMBECScorrect}
	For a \CSG, a \MEC $\endComponent$ and a valid upper bound $\upperBound$, for a non-empty set of states $\ecStates \subseteq \states$, it holds that 
	$\ecStates \in  \procname{FIND\_MBECs}(\game, \endComponent, \upperBound)$ if and only if (i) $\ecStates \subseteq \endComponent$, (ii) $\ecStates$ is a \BEC, and (iii) there exists no $T \subseteq \endComponent$ that is a \BEC and $\ecStates \subsetneq T$.
\end{restatable}

\begin{proof}[Proof sketch]
We prove the ``$\Rightarrow$" direction via structural induction. First, for the two basis cases, we show that if for a non-empty set $\ecStates \subseteq \endComponent$ it holds that $\ecStates \in \procname{FIND\_MBECs}(\game, \endComponent, \upperBound)$, then it has to be a maximal \BEC; otherwise $\ecStates$ would be the empty set. Next, in the recursive step, we show that within the set $B$ there must exist a \MEC $\endComponent<\prime> \subset B$ such that $\endComponent<\prime> = \ecStates$.

To prove the opposite direction, i.e., ``$\Leftarrow$", we assume towards a contradiction that $\ecStates \notin \procname{FIND\_MBECs}(\game, \endComponent, \upperBound)$. Since $\ecStates$ is a maximal \BEC, a non-empty set of states must exist where each state possesses at least one hazardous strategy (concerning $\endComponent$). However, as $\ecStates \notin \procname{FIND\_MBECs}(\game, \endComponent, \upperBound)$, there cannot exist a \MEC $\endComponent<\prime> \subseteq B$ where each state has a hazardous strategy with respect to $\endComponent<\prime>$. Consequently, $\ecStates$ is not an \MEC, which contradicts the assumption that $\ecStates$ is a maximal \BEC.
\end{proof}

\subsection{Convergent Bounded Value Iteration for CSGs and the Recursive Structure of ECs}\label{sec:convBVI} 

\begin{algorithm}[tb]\newtarget{alg:deflate_mecs}
	\begin{algorithmic}[1]
		\Function{\procname{DEFLATE}}{CSG $\game$, upper bound estimate $\upperBound$, \MEC $\endComponent$}
		\For{$\ecStates \in \newlink{alg:find_bec}{\procname{FIND\_MBECs}}(\game,  \endComponent,  \upperBound)$}\label{alg:whileBEC} 
		\State $\mathsf{u} \gets \bestExitVal[\upperBound](\ecStates)$ \label{alg:bestexit}
		\For{$\state \in \ecStates$}\label{alg:for_start}%
		\State $\upperBound(\state ) \gets \min( \upperBound(\state),\mathsf{u})$ \label{line:updateU}
		\EndFor\label{alg:for_end}
		\For{$\mathcal{E} \in \procname{FIND\_MECs}(\ecStates \setminus \bestExits[\upperBound](\ecStates))$}\label{alg:recurs_start}
			\State $\upperBound \gets$ \procname{DEFLATE}$(\game, \upperBound, \mathcal{E})$
			\Statex\Comment{Recursively deflate sub-\BECs}
		\EndFor\label{alg:recurs_end}
		\EndFor\label{alg:end_outer_for}
		\State \Return $\upperBound$
		\EndFunction
	\end{algorithmic}
	\caption{Algorithm for deflating \BECs.}
	\label{alg:deflate_mecs}
\end{algorithm}

Finally, \cref{alg:deflate_mecs} describes our main goal: the deflating procedure for \BECs, to be plugged into \cref{alg:bvi}. The algorithm takes a \CSG $\mathsf{G}$, the current upper bound estimate $\upperBound$, and a \MEC $\endComponent$ as input. 
First, the algorithm searches for all maximal \BECs that might be contained in the \MEC. 
The current upper bound estimate is returned if no \BEC exists in the \MEC. 
Otherwise, at least one \BEC exists and must be deflated. 
If multiple \BECs are found, each maximal \BEC is deflated separately (see the for-loop in Line \ref{alg:whileBEC}). 

To deflate a maximal \BEC $\ecStates$, first, the \emph{best} exit value $\bestExitVal[\upperBound](\ecStates)$ is estimated. Next, each state of the \BEC is considered and if the best exit value is smaller than the current upper bound estimate at that state, then it is reduced to the best exit value (as nothing better can be reached). 

However, notice that within a \BEC $\ecStates$ there might exist another \BEC $\ecStates<'>$.
From $\ecStates<'>$, Player $\reach$ might not be able to get to the best exit of $\ecStates$ (the globally best one in $\ecStates$) but only to a worse one, locally optimal for $\ecStates<'>$.
Hence, for such a sub-\BEC $\ecStates<'>$, more aggressive deflating to $\bestExitVal[\upperBound](\ecStates<'>)$ is due.
In other words, after deflating to $\bestExitVal[\upperBound](\ecStates)$ we have handled all states where $\reach$ can ensure reaching this best exit, and we can ignore these states for the moment;
we can also ignore this best exit and have a fresh look at which states are \emph{now} in a \BEC and can reach the \emph{second best} exit, i.e., the best exit in this remainder.
Subsequently, we continue with the third best option etc.

Consequently, on Lines \ref{alg:recurs_start}-\ref{alg:recurs_end}, \DEFLATEALG is called \emph{recursively} on all \MECs that are contained in $\ecStates \setminus \bestExits[\upperBound](\ecStates)$, i.e. after removing all best exits. 
The procedure is repeated independently for each maximal \BEC contained in $\endComponent$. A full example showing how \DEFLATEALG works on a more complex \BEC is included in \ifarxivelse{\cref{BVIfullExample}}{\cite[App. B-B]{arxivversion}}.

\begin{remark}[Structure of ECs]\label{rem:struct}
	This elucidates the hierarchical structure of \ECs and their corresponding best exits.
	The recursive call of \DEFLATEALG exposes the partial order over the \ECs, their sub-\ECs, and ``internally'' transient states (those not within sub-ECs after removing the best exit since they are bound to the just removed exit or another sub-\EC that is a \BEC with a lower value).
	This hierarchy captures (i) the ordering of exits by their values and (ii) the ``independence'' of exits with possibly different values when visiting one from another cannot be ensured.
\end{remark}

Our goal for the remainder of this work is to show that \cref{alg:bvi} with deflation is correct and converges, i.e.\ that complementing Bellman updates $\preop$ with deflation results in a sequence of upper bounds $\upperBound$ that converges to the value from above. 
For the sake of simplicity, we denote by $\newtarget{DefDeflop}{\deflop}:~[0,1]^{\mid \states\mid} \rightarrow [0,1]^{\mid \states\mid}$ the operator that performs \procname{DEFLATE} on a given valuation for all \MECs in the \CSG (which is reasonable since \cref{alg:bvi} performs \procname{DEFLATE} on all \MECs (in an arbitrary but fixed ordering)).

\begin{remark}[Valid upper bonds]
	In the following, whenever we quantify an upper bound (a.k.a.~over-approximation), we require it to be \emph{valid};
	meaning that it was obtained by iteratively applying deflating and the Bellman update on the initial over-approximation $\upperBound<0>$ from \cref{eq:upperBoundRec}, i.e. $\upperBound = (\deflop \circ \preop)^k(\upperBound<0>)$ for some $k \in \NN$. 
	The reason for this restriction is that for convergence, we require $\deflop$ to be order-preserving.
	While $\deflop$ is order-preserving on valid upper bounds (see \ifarxivelse{\cref{DBorderPreserving}}{\cite[\cref{DBorderPreserving}]{arxivversion}}), in general it is not monotonic for arbitrary $\upperBound \in \RR^{\mid \states \mid}$. 
	We illustrate this in \ifarxivelse{\cref{app:ex-defl-not-monotonic}}{\cite[App. B-C]{arxivversion}}.
\end{remark}

\begin{definition}[Valid upper bound]\label{def:validUpperBound}
	An upper bound $\upperBound \in [0,1]^{\mid \states \mid}$ is called \emph{valid} if there exists $k \in \NN$ such that $\upperBound = (\deflop \circ \preop)^k(\upperBound<0>)$, or if $\upperBound = \valR$.
\end{definition}

We proceed as follows: After proving fundamental properties of both operators $\preop$ and $\deflop$ in \cref{lem:fundamentals}, we use these properties to show that valid upper bounds are indeed upper bounds, i.e.\ they are always greater or equal than the value in \cref{lemValidDB}.
With correctness established, \cref{DBorderPreserving} shows that on valid upper bounds, $\deflBell$ is order-preserving, which is a necessary ingredient for convergence.
Finally, \cref{theo:soundness_completeness} concludes by proving soundness and completeness of the full algorithm.

\begin{restatable}[Properties of $\deflop$ and $\preop$]{lemma}{fundamentals}\label{lem:fundamentals}
	Let $\valuation\in [0,1]^{\mid \states \mid}$ be a valuation.
	If $\valuation \geq \valR$, then:
	\begin{enumerate}
		\item[(i)] $\preop(\valuation) \leq \valuation$ and $\deflop(\valuation) \leq \valuation$.
		\item[(ii)] $\preop(\valuation) \geq \valR$ and $\deflop(\valuation) \geq \valR$.
	\end{enumerate}
\end{restatable}
\begin{proof}
	For the Bellman operator, both properties follow from the fact that the value $\valR$ is the least fixpoint of the Bellman operator, see \cite[Thm.~1]{Alfaro2004}. Thus, given a valuation greater than the value, it cannot increase, and it cannot become smaller than the value.
	
	For $\deflop$, observe that deflation only updates the valuation in Line~\ref{line:updateU}  of \cref{alg:deflate_mecs} when setting the valuation of a state to $\min(\upperBound(s),\bestExitVal[\upperBound](\ecStates))$ for some \EC $\ecStates$.
	Item (i) holds because of taking the minimum with the current valuation, so it can only decrease.
	For (ii), we show in \ifarxivelse{\cref{lem:bestExitValSound} }{\cite[\cref{lem:bestExitValSound}]{arxivversion}} that for every \EC $\ecStates$ and state $s\in\ecStates$, we have that $\bestExitVal[\upperBound](\ecStates) \geq \valR(s)$.
	This proves our goal, as the only update of \procname{DEFLATE} keeps the valuation greater than $\valR$ in every state.
	The proof of \ifarxivelse{\cref{lem:bestExitValSound} }{\cite[\cref{lem:bestExitValSound}]{arxivversion}} is technically involved, having to unfold many definitions to show the following intuitive fact: No state can have a larger value than that of some exit from the \EC it is contained in.
\end{proof}

\begin{lemma}[Soundness of valid upper bounds]\label{lemValidDB}
	For all $k \in \NN$ it holds that $\deflBell^k(\upperBound<0>) \geq \valR$.
\end{lemma}
\begin{proof}
		We proceed by induction over $k$.
		\begin{description}
			\item[Base case:] $k=0$, thus $\deflBell^0(\upperBound<0>) = \upperBound<0> \geq \valR$.
			\item[Induction hypothesis:] For all $k\geq0$, we assume that\\ $\deflBell^k(\upperBound<0>) \geq \valR$.
			\item[Induction step:] To show: $\deflBell^{k+1}(\upperBound<0>) \geq \valR$.
			We know by induction hypothesis that $\deflBell^k(\upperBound<0>) \geq \valR$.
			Applying (ii) of \ifarxivelse{\cref{lem:fundamentals}  }{\cite[\cref{lem:fundamentals} ]{arxivversion}} for $\preop$, we obtain $\preop(\deflBell^k(\upperBound<0>)) \geq \valR$.
			From this and (ii) for $\deflop$, we get $\deflop(\preop(\deflBell^k(\upperBound<0>))) \geq \valR$, proving our goal.
		\end{description}
\end{proof}

\begin{restatable}[$\deflBell$ is order-preserving on valid upper bounds%
]{lemma}{DBorderPre}\label{DBorderPreserving}
	Let $\valuation[1], \valuation[2]$ be valid upper bounds with $\valuation[1] \geq \valuation[2]$.
	It holds that $\deflBell(\valuation[1]) \geq \deflBell(\valuation[2])$.
\end{restatable}	
\begin{proof}
	We know by \cref{lemValidDB} that all valid upper bounds are greater or equal to the value (including the value itself).
	Thus, if $\valuation[2] = \valR$, the statement holds, and if $\valuation[1]=\valR$, then $\valuation[2]=\valR$ as well, since $\valuation[1]\geq \valuation[2]$.
	It remains to show that both valuations come from repeated application of the deflation and Bellman operators, i.e.\ $\valuation[1] = \deflBell^i(\upperBound<0>)$ and $\valuation[2] = \deflBell^j(\upperBound<0>)$.
	We assume $\valuation[1] \neq \valuation[2]$, since otherwise the claim trivially holds.
	
	By (i) of \cref{lem:fundamentals}, every application of the operators can only decrease the resulting value; the item remains applicable, since the resulting valuations are always greater than or equal to the value.
	Consequently, $i < j$, as $\valuation[1] > \valuation[2]$. 
	Using this and applying (i) of \ifarxivelse{\cref{lem:fundamentals} }{\cite[\cref{lem:fundamentals}]{arxivversion}} again, we conclude by stating
	\begin{align*}
		\deflBell(\valuation[1]) = &\deflBell^{i+1}(\upperBound<0>) \geq\\ &\deflBell^{j+1}(\upperBound<0>) = \deflBell(\valuation[2]).
	\end{align*}
\end{proof}

\begin{restatable}[Soundness and completeness - Proof in
	\ifarxivelse{\cref{apx:soundness_and_completeness}}{\cite[App. C-D]{arxivversion}}]{theorem}{soundnessCompleteness}\label{theo:soundness_completeness}
For \CSGs \cref{alg:bvi}, using \cref{alg:deflate_mecs} as \DEFLATEALG, produces monotonic sequences $\lowerBound$ under- and $\upperBound$ over-approximating $\valR$, and terminates for every $\varepsilon > 0$.
\end{restatable}

\begin{proof}[Proof sketch]
Soundness and convergence of lower bounds is classical~\cite[Thm.~1]{Alfaro2004}, and our algorithm does not modify the computation of under-approximations.
The soundness of the upper bounds is immediate from \cref{lemValidDB}, since all upper bounds computed by the algorithm are valid, and thus greater or equal than the value.
Proving the convergence of the upper bounds is the main challenge.
First, in \ifarxivelse{\cref{lemma:convergence_to_fixpoint} }{\cite[Lem. 52]{arxivversion}} we provide the proof that the upper bound in \cref{alg:bvi} indeed converges to a fixpoint, using that the operators are order-preserving (\cref{DBorderPreserving}) and arguments from lattice theory.
Then, we use the same idea we have utilized in the proofs of \cref{theo:BVInoEC,thm:nonConvImpliesBEC}:
We assume for contradiction that there exists a state where the difference between the fixpoint of the upper bound in \cref{alg:bvi} and the true value is strictly greater than zero. The states with the largest such difference contain a \BEC (by \cref{thm:nonConvImpliesBEC}), that eventually will be found and deflated; since deflation depends on the outside of the \BEC, this decreases the upper bound. 
This is a contradiction to the fact that the upper bound has converged to a fixpoint.
Consequently, there can be no state with a positive difference, and the upper bounds converge, too. 
\end{proof}

\section{Conclusion and Future Work}
\label{sec:conclusion}
We have introduced a convergent over-approximation for concurrent stochastic games with reachability and safety objectives, thus giving value iteration the first sound stopping criterion and turning it into an anytime algorithm. Since the games are concurrent and ($\varepsilon$-)optimal strategies may need to be randomized, we could not use the technique of simple end components of \cite{eisentrautValueIterationSimple2022}. Instead, we identify bloated end components where the play can get stuck forever and recursively deflate the over-approximations of these states to the best possible value attainable upon leaving the end component. We leave an efficient implementation for future work as an extension 
of the standard model checker \textsc{prism-games}~\cite{Kwiatkowska2018}.

\bibliographystyle{plainurl}
\addcontentsline{toc}{chapter}{Bibliography}
\bibliography{MyLibrary}

\newpage
\onecolumn
\appendices

\crefalias{section}{appendix} %
\crefalias{subsection}{appendix} %
\section{Further Definitions and Concepts}

\subsection{Solving Matrix Games}\label{apx:LPforMatrixGames}
Given a valuation $\valuation$ for all states in a two-player \CSG, we can update the valuation at state $\state$ by solving a Linear Program (LP). Let the matrix game played at $\state$ be given by a matrix $\matrixGameMatrix \in \QQ^{l\times m}$, where $l$ (resp. $m$) is the number of actions available to player $\reach$ (resp. $\safe$) at the state $\state$. Then the LP that yields the value is the following \cite{santosAutomaticVerificationStrategy2020}: Maximize $\valuation(\state)$ subject to the constraints:
\begin{align*}
	\valuation(\state) &\leq x_1 \cdot z_{1j} + \dots + x_l \cdot z_{lj} \text{ for } 1 \leq j \leq m\\
	x_i &\geq 0 \text{ for } 1 \leq i \leq l\\
	1 &= x_1 + \dots + x_l
\end{align*}
where $z_{ij}= \matrixGameMatrix(\state)(i,j) = \sum_{\state<\prime> \in \states} \transitions(\state,\NFGaction[1](i),\NFGaction[2](j))(\state<\prime>) \cdot \valuation(\state<\prime>)$, and $x_i$ is the probability that player $\reach$ will take action $i$. Thus, by solving the LP, we not only obtain the value but also the optimal local strategy for player $\reach$.

\subsection{Domination of Strategies}\label{def:dominantStrategies}
\begin{definition}[Domination \protect{\cite[Def. 4.12]{maschlerGameTheory2020}}]\label{def:standardDominance}
	Given a valuation $\valuation$, a state $\state \in \states$ and sets of strategies $\strategies(\state)$ and $\strategies*(\state)$ available at the state $\state$ for player $\reach$ and $\safe$, respectively.
	\begin{description}
		\item[-] A strategy $\strategy \in \strategies(\state)$ is \emph{weakly dominated} if there exists another strategy $\strategy<\prime> \in \strategies(\state)$ satisfying the following two conditions:
		\begin{itemize}
			\item $\inf_{\strategy* \in \strategies[\safe](\state)}\preop(\valuation)(\state, \strategy,\strategy*) \leq \inf_{\strategy* \in \strategies[\safe](\state)} \preop(\valuation)(\state, \strategy<\prime>,\strategy*)$, and
			\item $\exists \strategy[\safe]<\prime> \in \strategies[\safe](\state)$ such that $\preop(\valuation)(\state, \strategy,\strategy[\safe]<\prime>) < \preop(\valuation)(\state, \strategy<\prime>,\strategy[\safe]<\prime>)$.
		\end{itemize}
		
		\item[-]\newtarget{def:dominantStrategies}{Dually}, a strategy $\strategy[\safe] \in \strategies*(\state)$ is \emph{weakly dominated} if there exists another strategy $\strategy[\safe]<\prime> \in \strategies*(\state)$ satisfying the following two conditions:
		\begin{itemize}
			\item $\sup_{\strategy \in \strategies(\state)}\preop(\valuation)(\state, \strategy,\strategy*) \geq \sup_{\strategy \in \strategies(\state)}\preop(\valuation)(\state, \strategy,\strategy[\safe]<\prime>)$, and
			\item $\exists \strategy<\prime> \in \strategies(\state)$ such that $\preop(\valuation)(\state, \strategy<\prime>,\strategy*) > \preop(\valuation)(\state, \strategy<\prime>,\strategy[\safe]<\prime>)$.
		\end{itemize}
	\end{description}
\end{definition}

\section{Further Examples}\newtarget{sec:furtherExamples}\label{AppfurtherExamples}

\subsection{Best Exit}\label{ex:BestExit}

\begin{figure}[b]
	\centering
	\resizebox{0.6\textwidth}{!}{\begin{tikzpicture}[
state/.style={draw,circle, minimum size = 0.5cm, align=center},
stochSt/.style={draw,circle, fill=black, minimum size = 0.1cm, inner sep=0pt},
transition/.style={->, black, -{Stealth[length=2mm]}, align=center},
exit/.style={->, black, -{Stealth[length=2mm]}, align=center},
invalid/.style={draw,rectangle,fill=black!30,minimum width=2cm,minimum height=1cm},
every label/.append style = {font=\small}
]

	\node (start) at (0.5, 1) {};
    \node[state] (S0) at (0,0) {\scalebox{1.3}{$s_0$}};
    \node[state,right of=S0, xshift=4cm, yshift=2.5cm] (S1)  {\scalebox{1.3}{$s_1$}};
    \node[state,right of=S0, xshift=4cm, yshift=-2.5cm] (S2)  {\scalebox{1.3}{$s_2$}};
    \node[stochSt, above of=S2, xshift=-0.8cm, yshift=0cm] (stoch) {};
    \node[stochSt, below of=S1, xshift=0cm, yshift=-1.5cm] (stoch3) {};
    \node[stochSt, below of=S0, xshift=-1.5cm, yshift=0cm] (stoch4) {};
    \node[cloud, left of=S0, xshift=-3.5cm, draw =black, fill = lipicsGray!20, minimum width = 2cm, minimum height = 1cm] (c1) {\scalebox{1.3}{$\alpha$}};
    \node[cloud, right of=S2, xshift=3cm, draw =black, fill = lipicsGray!20, minimum width = 2cm, minimum height = 1cm] (c2) {\scalebox{1.3}{$\gamma$}};
    \node[cloud, right of=S1, xshift=3cm, draw =black, fill = lipicsGray!20, minimum width = 2cm, minimum height = 1cm] (c3) {\scalebox{1.3}{$\beta$}};
    
    \draw[transition] (start) to (S0);

    \draw[transition] (S0) edge[bend left=10] node[above,align=center, yshift=0.1cm]{\scalebox{1.3}{$\mathsf{a}_1 \mathsf{d}_2$}} (S1);
    \draw[transition] (S0) edge[in=160,out=110,looseness=8] node[above]{\scalebox{1.3}{$\mathsf{a}_1 \mathsf{d}_1$}} (S0);
    \draw[exit] (S0) to node[above,align=center]{\scalebox{1.3}{$\mathsf{a}_2 \mathsf{d}_2$}} (c1);
    \draw[-] (S0) edge[bend left=15]  node[above,xshift=-0.3cm] {\scalebox{1.3}{$\mathsf{a}_2 \mathsf{d}_1$}} (stoch4); 
    	\draw[exit] (stoch4) edge[bend left=15] (c1.south);
    	\draw[exit] (stoch4) edge[bend right=70] (S0.south);
    
    \draw[exit] (S1) to node[above,align=center]{\scalebox{1.3}{$\mathsf{b}_1 \mathsf{e}_2$}} (c3);
    \draw[transition] (S1) edge[bend left=10]  node[below,yshift=-0.1cm] {\scalebox{1.3}{$\mathsf{b}_2 \mathsf{e}_2$}} (S0);
    \draw[-] (S1) edge[bend right=0]  node[right,xshift=-0.13cm,yshift=0.3cm] {\scalebox{1.3}{$\mathsf{b}_2 \mathsf{e}_1$}} (stoch3);
    	\draw[transition] (stoch3) edge[bend left=0] (S2);
    	\draw[transition] (stoch3) edge[bend right=75] ([yshift=-0.2cm, xshift=-0.04cm]S1.east);
    \draw[exit] (S1) edge[bend left=35] node[right,align=center,xshift=0.1cm]{\scalebox{1.3}{$\mathsf{b}_1 \mathsf{e}_1$}} (c2.north);	
    	
	\draw[exit] (S2) to node[above,align=center]{\scalebox{1.3}{$\mathsf{c}_2 \mathsf{f}_1$}} (c2);   
	\draw[transition] (S2) edge[bend left=20]  node[left,yshift=-0.2cm] {\scalebox{1.3}{$\mathsf{c}_1 \mathsf{f}_1$}} (S0); 	
    \draw[-] (S2) edge[bend right=0]  node[left] {\scalebox{1.3}{$\mathsf{c}_1 \mathsf{f}_2$}} (stoch);
    	\draw[transition] (stoch) edge[bend right=10] (S0.east);
    	\draw[transition] (stoch) edge[bend left=10] (S1);
    	\draw[transition] (S2) edge[in=20,out=70,looseness=8] node[above, yshift=0cm]{\scalebox{1.3}{$\mathsf{c}_2 \mathsf{f}_2$}} (S2);

\end{tikzpicture}}
	\captionof{figure}{\CSG with non-single-state \EC. The clouds represent the irrelevant parts of the game.}\label{fig:non_trivial_ec}
\end{figure}
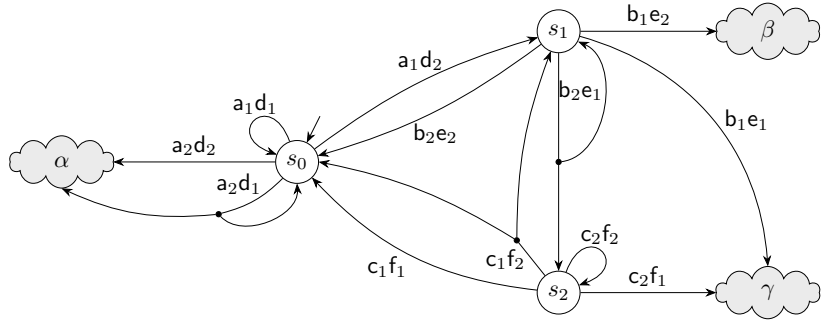

\begin{example}[Deflating BECs]\label{apx:deflBEC}
	Consider the \CSG depicted in \Cref{fig:non_trivial_ec}. At each state both players can choose among two actions. The game contains one \MEC $\endComponent \coloneqq \{\state[0], \state[1],\state[2]\}$. The matrix games played at the states $\state[0], \state[1]$ and $\state[2]$ with respect to an upper bound $\upperBound<k>$ (where $k \in \NN$) are defined as follows.
	\begin{align*}
		\mathsf{Z}_{\upperBound<k>}(\state[0]):=
		 \begin{blockarray}{ccc}
			\scriptstyle \mathsf{d}_1 & \scriptstyle  \mathsf{d}_2 \\
			\begin{block}{(cc)c}
			\upperBound<k>(\state[0]) & \upperBound<k>(\state[1])  & \scriptstyle \action[\reach](1) \\
\frac{1}{2}(\upperBound<k>(\state[0]) + \alpha) & \alpha & \scriptstyle \action[\reach](2)\\
			\end{block}
		\end{blockarray},\hspace{1cm}
		\mathsf{Z}_{\upperBound<k>}(\state[1]) = \begin{blockarray}{ccc}
			\scriptstyle \mathsf{e}_1 & \scriptstyle  \mathsf{e}_2 \\
			\begin{block}{(cc)c}
			\gamma   & \beta  & \scriptstyle \mathsf{b}_1 \\
			\frac{1}{2}(\upperBound<k>(\state[1])+\upperBound<k>(\state[2])) &  \upperBound<k>(\state[0]) & \scriptstyle \mathsf{b}_2\\
			\end{block}
		\end{blockarray},
			\end{align*}
			\vspace{-0.5cm}
				\begin{align*}
		\mathsf{Z}_{\upperBound<k>}(\state[2]) =\begin{blockarray}{ccc}
			\scriptstyle \mathsf{f}_1 & \scriptstyle  \mathsf{f}_2 \\
			\begin{block}{(cc)c}
			\upperBound<k>(\state[0]) & \frac{1}{2}(\upperBound<k>(\state[0])+\upperBound<k>(\state[1]))  & \scriptstyle \mathsf{c}_1\\
			\gamma & \upperBound<k>(\state[2])  & \scriptstyle \mathsf{c}_2\\
					\end{block}
	\end{blockarray}.
	\end{align*}
	For $k=0$, $\alpha=0.2, \beta=0.7$ and $\gamma=0.9$ and the initialization $\upperBound<0>(\state[0])=\upperBound<0>(\state[1])=\upperBound<0>(\state[2])=1$, the matrix games look as follows.
	\begin{align*}
		\mathsf{Z}_{\upperBound<0>}(\state[0]) = \begin{blockarray}{ccc}
			\scriptstyle \mathsf{d}_1 & \scriptstyle  \mathsf{d}_2 \\
\begin{block}{(cc)c}
			1 & 1  & \scriptstyle \action[\reach](1) \\
			0.6 & 0.2 & \scriptstyle \action[\reach](2)\\
			\end{block}
\end{blockarray}, \hspace*{1cm}
		\mathsf{Z}_{\upperBound<0>}(\state[1]) = \begin{blockarray}{ccc}
			\scriptstyle \mathsf{e}_1 & \scriptstyle  \mathsf{e}_2 \\
			\begin{block}{(cc)c}
			0.9  & 0.7  & \scriptstyle \mathsf{b}_1 \\
			1 & 1 & \scriptstyle \mathsf{b}_2\\
		\end{block}
\end{blockarray}, \hspace*{1cm}
	\mathsf{Z}_{\upperBound<0>}(\state[2]) = \begin{blockarray}{ccc}
			\scriptstyle \mathsf{f}_1 & \scriptstyle  \mathsf{f}_2 \\
			\begin{block}{(cc)c}
			1 & 1  & \scriptstyle \mathsf{c}_1 \\
			0.9 & 1 & \scriptstyle \mathsf{c}_2\\
		\end{block}
\end{blockarray}.
	\end{align*} 
	Since at each state, there exist hazardous and trapping strategies, the three states form a \BEC.
	To estimate the values of each exit, we need to solve three sub-matrix games played at each state of the \BEC. The best exit value is then the maximum of the three solutions. The three linear programs that solve the three sub-matrix games are the following.
	
	\begin{minipage}{0.2\textwidth}
		\begin{align*}
			&\max \upperBound<1>(\state[0]) \text{ s.t.}\\
			&\upperBound<1>(\state[0]) \leq 0.6 \cdot x_2\\
			&\upperBound<1>(\state[0])  \leq 0.2 \cdot x_2\\
			&x_2 = 1
		\end{align*}
	\end{minipage}
	\begin{minipage}{0.2\textwidth}
		\begin{align*}
			&\max \upperBound<1>(\state[1])  \text{ s.t.}\\
			&\upperBound<1>(\state[1]) \leq 0.9 \cdot x_1\\
			&\upperBound<1>(\state[1]) \leq 0.7 \cdot x_1\\
			&x_1 = 1
		\end{align*}
	\end{minipage}
	\begin{minipage}{0.2\textwidth}
		\begin{align*}
			&\max \upperBound<1>(\state[2]) \text{ s.t.}\\
			& \upperBound<1>(\state[2]) \leq 0.9 \cdot x_2\\
			& \upperBound<1>(\state[2]) \leq 1 \cdot x_2\\
			&x_2 = 1
		\end{align*}
	\end{minipage}
	
	Here $x_1$ and $x_2$ are the probabilities that player $\reach$ chooses the first or second action available at the corresponding state. The best exit is $\max \{0.2, 0.7, 0.9\} = 0.9$, therefore, the upper bound of all the three states can be safely reduced to 0.9.
\end{example}

\subsection{BVI Algorithm - Full Example}\label{BVIfullExample}

	\begin{figure}[tb]
	\centering
	\captionsetup{type=table}
	\captionof{table}{Full example of the BVI algorithm for the CSG depicted in \Cref{fig:non_trivial_ec}. $i$ is the $i$-th iteration of the main loop in \cref{alg:bvi}, $j$ is the $j$-th iteration of the while-loop in \cref{alg:deflate_mecs}. $\endComponent$ is the set of states among which a \BEC $\ecStates$ is searched. The best exit from $\ecStates$ is underlined.}\label{tab:bviFull}
	\begin{tabular}{cccccll}
		\toprule
		$i$ & $j$ & $\upperBound<i>(\state[0])$ & $\upperBound<i>(\state[1])$ & $\upperBound<i>(\state[2])$ & \multicolumn{1}{c}{$\endComponent$}      & \multicolumn{1}{c}{$\ecStates$} \\ \midrule
		\multirow{4}{*}{0}   & 0   & 1                        & 1                        & 1                        & $\{\state[0], \state[1], \state[2]\}$ & $\{\state[0], \state[1], \underline{\state[2]}\}$   \\
		& 1   & 0.9                      & 0.9                      & 0.9                      & $\{\state[0], \state[1]\}$            & $\{\state[0], \underline{\state[1]}\}$              \\
		& 2   & 0.7                      & 0.7                      & 0.9                      & $\{\state[0]\}$                       & $\{\underline{\state[0]}\}$                         \\
		& 3   & 0.2                      & 0.7                      & 0.9                      & $\emptyset$                           & $\emptyset$                             \\ \midrule
		\multirow{3}{*}{1}   & 0   & 0.2                      & 0.7                      & 0.9                      & $\{\state[0],\state[1], \state[2]\}$  & $\{\underline{\state[0]}\}$                         \\
		& 1   & 0.2                      & 0.7                      & 0.9                      & $\{\state[1], \state[2]\}$            & $\{\underline{\state[2]}\}$                         \\
		& 2   & 0.2                      & 0.7                     & 0.45                      & $\{\state[1]\}$                       & $\emptyset$    
		\\ \bottomrule                      
	\end{tabular}
\end{figure}

\begin{example}[BVI]\label{ex:bviFull}
	Consider the \CSG depicted in \Cref{fig:non_trivial_ec}, where $\alpha=0.2, \beta=0.7$, and $\gamma=0.9$. The matrix games played at each state are given by the matrices $\mathsf{Z}_{\upperBound<k>}(\state[0]), \mathsf{Z}_{\upperBound<k>}(\state[1])$ and $\mathsf{Z}_{\upperBound<k>}(\state[2])$ defined as in Example \ref{apx:deflBEC}. We choose $\varepsilon=0.001$.
	
	\cref{tab:bviFull} summarizes the steps of the algorithm. Initially it holds that $\upperBound<0>(\state[0]) = \upperBound<0>(\state[1]) = \upperBound<0>(\state[2]) = 1$. Since $\endComponent = \{\state[0], \state[1], \state[2]\}$ is a \MEC, it will be found by \cref{alg:deflate_mecs}. Within $\endComponent$ the \BEC $\ecStates = \endComponent$ is found. The best exit value from $\ecStates$ is calculated, i.e. the three sub-matrix games are solved and the best exit value needed for deflating the \BEC is the maximum among the solutions: 
	\begin{align*}
		\mathsf{Z}^{\mathsf{exit}}_{\upperBound<0>}(\state[0]):= \begin{blockarray}{ccc}
			\scriptstyle \mathsf{d}_1 & \scriptstyle  \mathsf{d}_2 \\
			\begin{block}{(cc)c}
			0.6 & 0.2 & \scriptstyle \action[\reach](2)\\
			\end{block}
		\end{blockarray}&\hspace{1cm}
		\mathsf{Z}^{\mathsf{exit}}_{\upperBound<0>}(\state[1]) = \begin{blockarray}{ccc}
			\scriptstyle \mathsf{e}_1 & \scriptstyle  \mathsf{e}_2 \\
\begin{block}{(cc)c}
			0.9 & 0.7  & \scriptstyle \mathsf{b}_1 \\
		\end{block}
\end{blockarray}&\hspace{-0.2cm}
	\mathsf{Z}^{\mathsf{exit}}_{\upperBound<0>}(\state[2]) = \begin{blockarray}{ccc}
			\scriptstyle \mathsf{f}_1 & \scriptstyle  \mathsf{f}_2 \\
	\begin{block}{(cc)c}
			0.9 & 1  & \scriptstyle \mathsf{c}_1\\
			\end{block}
			\end{blockarray}
	\end{align*}
	The best exit from $\ecStates$ is $\state[2]$ (underlined in \cref{tab:bviFull}) and the value of the exit is 0.9, so in the next step the over-approximations of those three states are deflated to 0.9 and the best exit is removed from $\endComponent$.
	\begin{figure}[tb]
		\centering
		\resizebox{0.6\textwidth}{!}{\begin{tikzpicture}[
state/.style={draw,circle, minimum size = 0.5cm, align=center},
stochSt/.style={draw,circle, fill=black, minimum size = 0.1cm, inner sep=0pt},
transition/.style={->, black, -{Stealth[length=2mm]}, align=center},
exit/.style={->, black, -{Stealth[length=2mm]}, align=center},
invalid/.style={draw,rectangle,fill=black!30,minimum width=2cm,minimum height=1cm},
every label/.append style = {font=\small}
]

    \node[state] (S0) at (0,0) {\scalebox{1.3}{$s_0$}};
    \node[state,right of=S0, xshift=4cm] (S1)  {\scalebox{1.3}{$s_1$}};
    \node[stochSt, below of=S0, xshift=-1.5cm, yshift=0cm] (stoch4) {};
    \node[cloud, left of=S0, xshift=-2.5cm, draw =black, fill = lipicsGray!20, minimum width = 2cm, minimum height = 1cm] (c1) {\scalebox{1.3}{$\alpha$}};
    \node[cloud, right of=S1, xshift=3cm, yshift=-1cm, draw =black, fill = lipicsGray!20, minimum width = 2cm, minimum height = 1cm] (c2) {\scalebox{1.3}{$\gamma$}};
    \node[cloud, right of=S1, xshift=3cm, yshift=1cm, draw =black, fill = lipicsGray!20, minimum width = 2cm, minimum height = 1cm] (c3) {\scalebox{1.3}{$\beta$}};

    \draw[transition] (S0) edge[bend left=10] node[above,align=center, yshift=0.1cm]{\scalebox{1.3}{$\mathsf{a}_1 \mathsf{d}_2$}} (S1);
    \draw[transition] (S0) edge[in=160,out=110,looseness=8] node[above]{\scalebox{1.3}{$\mathsf{a}_1 \mathsf{d}_1$}} (S0);
    \draw[exit] (S0) to node[above,align=center]{\scalebox{1.3}{$\mathsf{a}_2 \mathsf{d}_2$}} (c1);
    \draw[-] (S0) edge[bend left=15]  node[above,xshift=-0.2cm] {\scalebox{1.3}{$\mathsf{a}_2 \mathsf{d}_1$}} (stoch4); 
    	\draw[exit] (stoch4) edge[bend left=15] (c1.south);
    	\draw[exit] (stoch4) edge[bend right=70] (S0.south);
    
    \draw[exit] (S1) to node[above,align=center]{\scalebox{1.3}{$\mathsf{b}_1 \mathsf{e}_2$}} (c3);
    \draw[transition] (S1) edge[bend left=10]  node[below,yshift=-0.1cm] {\scalebox{1.3}{$\mathsf{b}_2 \mathsf{e}_2$}} (S0);
    \draw[exit] (S1) edge[bend left=0] node[above,align=center]{\scalebox{1.3}{$\mathsf{b}_1 \mathsf{e}_1$}} (c2);

\end{tikzpicture}}
		\captionof{figure}{The resulting game after removing $\state[2]$ from the game illustrated in \Cref{fig:non_trivial_ec}.}\label{apxfig:non_trivial_ec_deflated_1}
	\end{figure}
	The resulting sub-game is depicted in \Cref{apxfig:non_trivial_ec_deflated_1}. For $\upperBound<\prime>(\state[0]) = \upperBound<\prime>(\state[1]) = \upperBound<\prime>(\state[2]) =0.9$, the set $\endComponent=\{\state[0], \state[1]\}$ contains the \BEC $\ecStates = \endComponent$. Now, two linear programs need to be solved to solve the two sub-matrix games.
	\begin{align*}
		\mathsf{Z}^{\mathsf{exit}}_{\upperBound<\prime>}(\state[0]):= \begin{blockarray}{ccc}
			\scriptstyle \mathsf{d}_1 & \scriptstyle  \mathsf{d}_2 \\
	\begin{block}{(cc)c}
			0.55 & 0.2 & \scriptstyle \action[\reach](2)\\
			\end{block}
\end{blockarray}&\hspace{1cm}
		\mathsf{Z}^{\mathsf{exit}}_{\upperBound<\prime>}(\state[1]) = \begin{blockarray}{ccc}
			\scriptstyle \mathsf{e}_1 & \scriptstyle  \mathsf{e}_2 \\
				\begin{block}{(cc)c}
			0.9 & 0.7  & \scriptstyle \mathsf{b}_1 \\
\end{block}
\end{blockarray}
	\end{align*}
	The best exit is $\state[1]$ and the value of the exit 0.7, thus, the upper bounds of $\state[0]$ and $\state[1]$ are reduced to 0.7 and $\state[1]$ is removed from $\endComponent$. The resulting sub-game is depicted in \Cref{apxfig:non_trivial_ec_deflated_2}.
	\begin{figure}[tb]
		\centering
		\resizebox{0.2\textwidth}{!}{\begin{tikzpicture}[
state/.style={draw,circle, minimum size = 0.5cm, align=center},
stochSt/.style={draw,circle, fill=black, minimum size = 0.1cm, inner sep=0pt},
transition/.style={->, black, -{Stealth[length=2mm]}, align=center},
exit/.style={->, black, -{Stealth[length=2mm]}, align=center},
invalid/.style={draw,rectangle,fill=black!30,minimum width=2cm,minimum height=1cm},
every label/.append style = {font=\small}
]

    \node[state] (S0) at (0,0) {\scalebox{1.3}{$s_0$}};
    \node[stochSt, below of=S0, xshift=-1.5cm, yshift=0cm] (stoch4) {};
    \node[cloud, left of=S0, xshift=-2.5cm, draw =black, fill = lipicsGray!20, minimum width = 2cm, minimum height = 1cm] (c1) {\scalebox{1.3}{$\alpha$}};

    \draw[transition] (S0) edge[in=160,out=110,looseness=8] node[above]{\scalebox{1.3}{$\mathsf{a}_1 \mathsf{d}_1$}} (S0);
    \draw[exit] (S0) to node[above,align=center]{\scalebox{1.3}{$\mathsf{a}_2 \mathsf{d}_2$}} (c1);
    \draw[-] (S0) edge[bend left=15]  node[above,xshift=-0.2cm] {\scalebox{1.3}{$\mathsf{a}_2 \mathsf{d}_1$}} (stoch4); 
    	\draw[exit] (stoch4) edge[bend left=15] (c1.south);
    	\draw[exit] (stoch4) edge[bend right=70] (S0.south);
\end{tikzpicture}}
		\captionof{figure}{The resulting game after removing $\state[1]$ from the game illustrated in \Cref{apxfig:non_trivial_ec_deflated_1}.}\label{apxfig:non_trivial_ec_deflated_2}
	\end{figure}
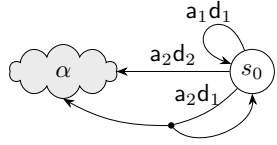
	Finally, for $\upperBound<\prime\prime>(\state[0])=0.7$, the set $\endComponent = \{\state[0]\}$ contains the \BEC $\ecStates=\{\state[0]\}$. To deflate it, we need to solve the following exiting sub-game.
	\begin{align*}
		\mathsf{Z}^{\mathsf{exit}}_{\upperBound<\prime\prime>}(\state[0]):=\begin{blockarray}{ccc}
			\scriptstyle \mathsf{d}_1 & \scriptstyle  \mathsf{d}_2 \\
			\begin{block}{(cc)c}
			0.45 & 0.2 & \scriptstyle \action[\reach](2)\\
\end{block}
\end{blockarray}
	\end{align*}
	The best exit value is 0.2. The upper upper bound of $\state[0]$ is reduced to 0.2. After removing the state from $\ecStates$ we obtain an empty set and the algorithm can proceed with the next \MEC. Since, we assumed that the game has only one \MEC, the deflating phase is finished. As $\upperBound<0> - \lowerBound<0> > \varepsilon$ holds, the next iteration of the algorithm is executed.
	
	The Bellman update returns the same valuation for all states, i.e. $\upperBound<1>(\state[0]) =0.2$, $\upperBound<1>(\state[1]) =0.7$ and $\upperBound<1>(\state[2]) = 0.9$. \cref{alg:deflate_mecs} again finds the \MEC $\endComponent = \{\state[0],\state[1],\state[2]\}$ that contains the two \BECs $\ecStates[1] \coloneqq \{\state[0]\}$ and $\ecStates[1] \coloneqq \{\state[2]\}$. First $\ecStates = \{ \state[0]\}$ is deflated to 0.2 and next the \BEC $\ecStates=\{\state[2]\}$ is deflated. For this, the following exiting sub-game needs to be solved.
	\begin{align*}
		\mathsf{Z}^{\mathsf{exit}}_{\upperBound<1>}(\state[0]):=\begin{blockarray}{cc}
			\scriptstyle  \mathsf{d}_2 \\
			\begin{block}{(c)c}
			0.45 & \scriptstyle \action[\reach](2)\\
\end{block}
\end{blockarray}
	\end{align*}
	Notice, that here we only consider staying strategies for player $\safe$, which is why we only consider action $\mathsf{d}_2$. Therefore, the best exit value is 0.45. After removing $\state[0]$ and $\state[2]$ from $\ecStates$, no further \BECs are contained in the \MEC. Since now $\upperBound<1> - \lowerBound<1> < \varepsilon$ holds, the \BVI algorithm terminates.
\end{example}

\subsection{Non-Monotonicity of Deflation}\label{app:ex-defl-not-monotonic}
\begin{figure}
	\centering
	\resizebox{0.4\textwidth}{!}{\begin{tikzpicture}[
	state/.style={draw,circle, minimum size = 0.5cm, align=center},
	stochSt/.style={draw,circle, fill=black, minimum size = 0.1cm, inner sep=0pt},
	transition/.style={->, black, -{Stealth[length=2mm]}, align=center},
	exit/.style={->, black, -{Stealth[length=2mm]}, align=center},
	invalid/.style={draw,rectangle,fill=black!30,minimum width=2cm,minimum height=1cm},
	every label/.append style = {font=\small}
	]
	
	\node (start) at (-1,0){};
	\node[state] (S0) at (0,0) {\scalebox{1.3}{$s_0$}};
	\node[state,right of=S0, xshift=3cm] (S1)  {\scalebox{1.3}{$s_1$}};
	\node[cloud, right of=S1, xshift=3cm, yshift=1cm, draw =black, fill = lipicsGray!20, minimum width = 2cm, minimum height = 1cm] (c2) {\scalebox{1.3}{$\beta$}};
		\node[cloud, above of=S1, yshift=1cm, draw =black, fill = lipicsGray!20, minimum width = 2cm, minimum height = 1cm] (c1) {\scalebox{1.3}{$\alpha$}};
		\node[cloud, below of=c2, xshift=0cm, yshift=-1cm, draw =black, fill = lipicsGray!20, minimum width = 2cm, minimum height = 1cm] (c3) {\scalebox{1.3}{$\gamma$}};
	
	\draw[transition] (start) to (S0);
	
	\draw[transition] (S0) edge[bend left=10] node[above,align=center, yshift=0.0cm]{\scalebox{1.3}{$(\square, \square)$}} (S1);
	
	\draw[transition] (S1) edge[bend left=0] node[left,align=center, yshift=0cm]{\scalebox{1.3}{($\mathsf{a}, \mathsf{d}$)}} (c1);

	\draw[exit] (S1) to node[above,align=center,xshift=-0.15cm, yshift=0cm]{\scalebox{1.3}{$(\mathsf{b}, \mathsf{e})$}\\\scalebox{1.3}{$(\mathsf{c}, \mathsf{e})$}} (c2);
	
	  \draw[transition] (S1) edge[in=280,out=230,looseness=5] node[below]{\scalebox{1.3}{$(\mathsf{a},\mathsf{e})$}} (S1);
	
	\draw[transition] (S1) edge[bend left=10]  node[below,yshift=0cm] {\scalebox{1.3}{$(\mathsf{b}, \mathsf{d})$}} (S0);
	
	\draw[exit] (S1) edge[bend left=0] node[above,yshift=-0.15cm, xshift=0.5cm,align=center]{\scalebox{1.3}{$(\mathsf{c}, \mathsf{d})$}} (c3);

\end{tikzpicture}}
	\caption{Monotonicity is not guaranteed in general.}
	\label{fig:non_monotonic}
\end{figure}
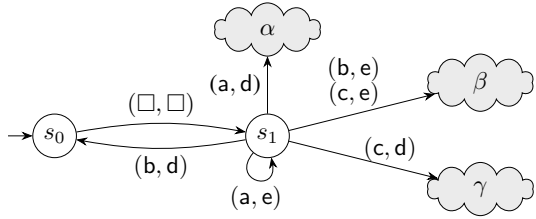

\begin{example}
	Consider the \CSG depicted in \cref{fig:non_monotonic}. The variables $\alpha, \beta$ and $\gamma$ are placeholders indicating that upon leaving a certain valuation is obtained. 
	The set of states $\{\state[0], \state[1]\}$ is an \EC. We consider two valuations, $\upperBound$ and $\upperBound<\prime>$, such that $\upperBound \geq \upperBound<\prime>$ and the \EC $\{\state[0], \state[1]\}$ is for both upper bounds a \BEC. The upper bound $\upperBound$ assigns the following valuations: $\upperBound(\state[0]) = 0.6, \upperBound(\state[1]) = 0.6, \upperBound(\alpha) =0.8, \upperBound(\beta) = 0.5,\upperBound(\gamma=0.55)$. The upper bound $\upperBound<\prime>$ assigns the following valuations: $\upperBound<\prime>(\state[0]) = 0.6, \upperBound<\prime>(\state[1]) = 0.45,\upperBound<\prime>(\alpha) =0.5, \upperBound<\prime>(\beta) = 0.5,\upperBound<\prime>(\gamma=0.55)$. Then, the matrix games played for the two valuations at state $\state[1]$ are given by the following matrices:
	\begin{align*}
		\matrixGameMatrix[\upperBound](\state[1]) \coloneqq \begin{blockarray}{ccc}
			\scriptstyle \mathsf{d} & \scriptstyle \mathsf{e}\\
			\begin{block}{(cc)c}
				0.8 & 0.6 & \scriptstyle \mathsf{a}\\
				0.6 & 0.5 & \scriptstyle \mathsf{b}\\
				0.55 & 0.5 & \scriptstyle \mathsf{c}\\
			\end{block}
		\end{blockarray}, \hspace{0.5cm}
		\matrixGameMatrix[\upperBound<\prime>](\state[1]) \coloneqq \begin{blockarray}{ccc}
			\scriptstyle \mathsf{d} & \scriptstyle \mathsf{e}\\
			\begin{block}{(cc)c}
				0.5 & 0.45 & \scriptstyle \mathsf{a}\\
				0.6 & 0.5 & \scriptstyle \mathsf{b}\\
				0.55 & 0.5 & \scriptstyle \mathsf{c}\\
			\end{block}
		\end{blockarray} 
	\end{align*}
	Then, under $\upperBound$ the strategy $\{\mathsf{a} \mapsto 1\}$ is a hazardous strategy and the exit value is 0.5. In contrast for $\upperBound<\prime>$ the strategy $\{\mathsf{b} \mapsto 1\}$ is hazardous and the exit value is 0.55. Thus, for arbitrary $\upperBound, \upperBound<\prime> \in \RR^{\mid \states \mid}$ %
	it might happen that $\exitVal[\upperBound](\ecStates) < \exitVal[\upperBound<\prime>](\ecStates)$ for some \BEC $\ecStates$ although $\upperBound \geq \upperBound<\prime>$ holds.
	Intuitively, this arises because the sub-EC which forms a \BEC changes when the relative ordering of exits is modified.
	However, the problem cannot occur when considering valid over-approximations (see \cref{DBorderPreserving}) because then the upper bounds decrease in a well-behaved way when the relative ordering of the exits changes.
\end{example}

\section{Proofs for \Cref{sec:3-vi-title}}\newtarget{sec:correctnesProofsforBVI}
\label{app:proofs-for-vi}

Throughout the whole Appendix \ref{app:proofs-for-vi}, when proving convergence of \BVI, we utilize definitions and theorems from lattice and fixpoint theory.
Thus, we first briefly recall some necessary definitions (adjusting notation to our work) and theorems.

\begin{definition}[Ordered set {\cite[Chapter 1.2]{daveyIntroductionLatticesOrder2002}}]
	A set $P$ equipped with a relation $\preceq: P \times P$ is called an \emph{ordered set} if and only if $\preceq$ is reflexive, antisymmetric and transitive.
\end{definition}

\begin{definition}[Directed set {\cite[Chapter 7.7]{daveyIntroductionLatticesOrder2002}}]\newtarget{def:directedSet}
	Let $P$ be an ordered set.
	A non-empty set $D \subseteq P$ is \emph{directed} if and only if for every pair of elements $x,y \in \mathsf{D}$ there exists $z \in \mathsf{D}$ that is a lower bound for both, formally $z \preceq x$ and $z \preceq y$.
\end{definition}

\begin{definition}[Complete partial order (CPO) {\cite[Chapter 8.1]{daveyIntroductionLatticesOrder2002}}] %
	An ordered set $P$ is a complete partially ordered set (CPO) if and only if
	\begin{itemize}
		\item[(i)] $P$ has a top element, $\top \coloneq \inf_P \emptyset$, and
		\item[(ii)] for every directed set $D \subseteq P$, we have $\inf_P D$ exists. 
	\end{itemize}
\end{definition}

\begin{definition}[Continuity {\cite[Chapter 8.6]{daveyIntroductionLatticesOrder2002}}]\newtarget{def:scottContinuity}
	\label{def:continuous}
	Let $P$ and $Q$ be two CPOs. 
	A mapping $\phi: P \rightarrow Q$ is \emph{continuous} if
	\begin{enumerate}
		\item[(i)] for every \newlink{def:directedSet}{directed set} $D \subseteq P$, the subset $\phi(D)$ of $Q$ is also directed, and
		\item[(ii)] it holds that $\phi(\inf D) = \inf \phi(D) := \inf \{\phi(x) | x \in D\}$.
	\end{enumerate}
\end{definition}

\begin{definition}[Order-preserving {\cite[Chapter 1.34]{daveyIntroductionLatticesOrder2002}}]
	Let $P$ and $Q$ be ordered sets. A map $\phi: P \rightarrow Q$ is \emph{order-preserving} (also called monotone) if $x \preceq y$ in $P$ implies $\phi(x) \preceq \phi(y)$ in~$Q$.
\end{definition}

\begin{theorem}[Fixpoint Theorem {\cite[Chapter 8.15]{daveyIntroductionLatticesOrder2002}}]\label{thm:fixpoint-theorem}
	Let $P$ be a complete partial order, let $F$ be an order-preserving and continuous self-map on $P$ and define $\alpha := \sup_{n \geq 0} F^n(\top)$.
	Then $\alpha$ is the greatest fixpoint of $F$, i.e.\ the largest element of $P$ satisfying $F(\alpha)=\alpha$.
\end{theorem}
We remark that we inverted the definitions and the theorem: This is because we are interested in a greatest fixpoint, whereas the textbook~\cite{daveyIntroductionLatticesOrder2002} only speaks about least fixpoints. Inverting the comparator and replacing $\inf$ with $\sup$ yields the original definitions. With only these changes, the proof of \cite[Chapter 8.15]{daveyIntroductionLatticesOrder2002} yields our modified claim \Cref{thm:fixpoint-theorem}.

\subsection{Convergence in the absence of end components}\label{apx:proof_no_ecs}

We start by proving a technical lemma that is also useful for several future proofs:
The over-approximation computed using only Bellman updates converges to a fixpoint.

\begin{lemma}[Upper bound converges to a fixpoint]\label{lem:Ustar-is-fixpoint}
	Let $(\upperBound<k>)_{k\in\mathbb{N}_0}$ be the sequence of upper bounds computed by applying \Cref{eq:upperBoundRec} on a \CSG $\game$.
	Let $\upperBound<\star> \coloneqq \displaystyle{\lim_{k\to\infty}} \upperBound<k>$ be the limit of the sequence.
	This limit is a fixpoint of the Bellman update, i.e.\ for all $\state \in \states$, $\preop(\upperBound<\ast>)(\state) = \upperBound<\star>(\state)$.
\end{lemma}
\begin{proof}
	This lemma is a consequence of the fixpoint theorem we just recalled.
	Thus, we proceed as follows: We explain that the domain of $\preop$ is a CPO and prove that $\preop$ is order-preserving and continuous. Then, \Cref{thm:fixpoint-theorem} yields that $\upperBound<\star>$ is a (namely the greatest) fixpoint.
	
	\subparagraph{Complete partial order}	
	The domain of $\preop$ are valuations, i.e.\ vectors $[0,1]^{|\states|}$ mapping every state to a number. 
	Thus, we define the set $P$ to be the set of all valuations.
	We use the standard point-wise comparisons as relation, i.e.\ $\valuation[1] \preceq \valuation[2]$ if and only if for all states $s\in\states$ we have $\valuation[1](s) \leq \valuation[2](s)$.
	Thus, the top element $\top$ is the function that maps all states to $1$.
	For every directed set $D$, a greatest lower bound $\ell = \inf_P D$ exists:
	Set $\ell(s) = \inf_{d \in D} d(s)$ for all $s \in \states$.
	It is a lower bound, as by point-wise comparison, it is smaller than all valuations in $D$; it is the greatest lower bound, since picking a larger number for any state would not be a lower bound any more.
	Thus, the set consisting of valuations $[0,1]^{|\states|}$ with this relation is a CPO.

	\subparagraph{Order-preserving}
	Recall that the Bellman operator on a state is defined as follows:
	${\preop}(\valuation)(\state)\coloneqq \adjustlimits\sup_{\strategy \in \strategies(\state)} \inf_{\strategy* \in \strategies*(\state)} \preop(\valuation)(\state, \strategy, \strategy*)$, where
	\begin{align*}
		\preop(\valuation){}&(\state, \strategy, \strategy*) \coloneqq \sum_{(\action[\reach],\action[\safe]) \in \actions} \sum_{\state<\prime> \in
			\states} \valuation(\state<\prime>) \cdot
		\transitions(\state,\action[\reach],\action[\safe])(\state<\prime>) \cdot
		\strategy[\reach](\state)(\action[\reach]) \cdot
		\strategy[\safe](\state)(\action[\safe]).
	\end{align*}
	\Cref{eq:upperBoundRec} lifts it to valuations by applying it state-wise.
	Hence, for every state, we apply an operation consisting of multiplications and summations, which are order-preserving.
	Thus, overall, the Bellman operator is order-preserving.
		
	\subparagraph{Continuous}
	We just showed that the Bellman operator on valuations is an order-preserving self-map on the set $P$ of valuations.
	Then, \cite[Lemma 8.7 (i)]{daveyIntroductionLatticesOrder2002} yields that for every directed subset $D \subseteq P$, we have that $\preop(D) \coloneqq \{\preop(d)\mid d\in D\}$ is a directed subset, which is Condition (i) of \Cref{def:continuous}.
	It remains to show Condition (ii): $\inf_{d \in D} \preop(d) = \preop(\inf D)$. %
	Since the comparisons by the relation $\preceq$ are performed point-wise, we have to prove that for all states $s\in\states$, we have $\inf_{d \in D} \preop(d)(s) = \preop(\inf D)(s)$. 
	Thus, fix an arbitrary state $s \in \states$, and conclude using the following chain of equations.
	\begin{align*}
		&\preop(\inf D)(s)\\
		&= 
		\adjustlimits\sup_{\strategy \in \strategies(\state)} \inf_{\strategy* \in \strategies*(\state)} \sum_{(\action[\reach],\action[\safe]) \in \actions} \sum_{\state<\prime> \in \states}\\ 
		&\quad\quad 
		\left(\inf_{d\in D} d (\state<\prime>)\right) \cdot 		\transitions(\state,\action[\reach],\action[\safe])(\state<\prime>) \cdot 		\strategy[\reach](\state)(\action[\reach]) \cdot 		\strategy[\safe](\state)(\action[\safe]) 
		\tag{Unfolding definition of Bellman operator}\\
		&= 
		\adjustlimits 
		\sup_{\strategy \in \strategies(\state)} \inf_{\strategy* \in \strategies*(\state)} \inf_{d \in D} \sum_{(\action[\reach],\action[\safe]) \in \actions} \sum_{\state<\prime> \in \states}\\ 
		&\quad\quad 
		d (\state<\prime>) \cdot 		\transitions(\state,\action[\reach],\action[\safe])(\state<\prime>) \cdot 		\strategy[\reach](\state)(\action[\reach]) \cdot 		\strategy[\safe](\state)(\action[\safe]) 
		\tag{\Cref{arg:UstarFixpointArg1}}\\
		&= 
		\adjustlimits 
		\inf_{d \in D} \sup_{\strategy \in \strategies(\state)} \inf_{\strategy* \in \strategies*(\state)}  \sum_{(\action[\reach],\action[\safe]) \in \actions} \sum_{\state<\prime> \in \states}\\ 
		&\quad\quad 
		d (\state<\prime>) \cdot 		\transitions(\state,\action[\reach],\action[\safe])(\state<\prime>) \cdot 		\strategy[\reach](\state)(\action[\reach]) \cdot 		\strategy[\safe](\state)(\action[\safe]) 
		\tag{\Cref{arg:UstarFixpointArg2}}\\
		&=\inf_{d \in D} \preop(d)(s).
		\tag{Collapsing the Bellman operator definition}		
	\end{align*}

	\begin{argument}\label{arg:UstarFixpointArg1}
		This step moves the $\inf_{d\in D}$ out of the summation, which is correct, since addition is a continuous operation.
	\end{argument}

	\begin{argument}\label{arg:UstarFixpointArg2}
		This step moves $\inf_{d\in D}$ to the front, first utilizing that infima can be switched.
		Then, to switch $\inf_{d \in D}$ and $\sup_{\strategy \in \strategies(\state)}$, we make use of the Minimax Theorem \cite{simonsMinimaxTheoremsTheir1995} which states that for a function $f: X \times Y \rightarrow \mathbb{R}$ that is concave-convex, it holds that $\sup_{x \in X}\inf_{y \in Y} f(x,y) = \inf_{y \in Y} \sup_{x \in X} f(x,y)$.
		$f$ is concave-convex if $f$ is concave for a fixed $y \in Y$ and convex for a fixed $x \in X$. 
		This holds, in particular, for bilinear functions, i.e. functions that are linear in both arguments. 		
		The function considered at this point is the following:
		\allowdisplaybreaks\begin{align*}
			f&(\strategy, d) = \inf_{\strategy* \in \strategies*(\state)}  \sum_{(\action[\reach], \action[\safe]) \in \actions} \sum_{\state<\prime> \in \states} \\
			&\quad\quad d(\state<\prime>) \cdot \transitions(\state, \action[\reach], \action[\safe])(\state<\prime>) \cdot \strategy(\state)(\action[\reach]) \cdot \strategy*(\state)(\action[\safe]).
		\end{align*}
		This function is indeed bilinear, since addition and multiplication are linear functions.
	\end{argument}

	Overall, we have shown that the sequence $(\upperBound<k>)_{k\in\mathbb{N}_0}$ is the result of applying an order-preserving, continuous function to the top-element of a complete partial order, and thus it converges to a (the greatest) fixpoint.

\end{proof}

\theoBVInoEC*

\begin{proof}
	This proof is an extension of the proof of~\cite[Theorem 1]{eisentrautValueIterationSimple2022} for turn-based games to the concurrent setting.
	The underlying idea is the same, and can be briefly summarized as follows:
	We assume towards a contradiction that $\upperBound<\star> \neq \valR$, and find a set $\ecStates$ that maximizes the difference between upper bound and value.
	Every pair of strategies leaving the set decreases the difference.
	However, $\valR$ and $\upperBound<\star>$ are fixpoints of the Bellman updates, from \cite[Theorem~1]{Alfaro2004} and \cref{lem:Ustar-is-fixpoint}, respectively.
	Consequently, optimal strategies need to remain in the set.
	However, in the absence of \ECs, optimal strategies have to leave the set, which yields a contradiction and proves that $\upperBound<\star> = \valR$.
	
	\subparagraph{Main challenge}
	The key difference to the proof of~\cite[Theorem 1]{eisentrautValueIterationSimple2022} is that we cannot argue about actions anymore, but have to consider mixed strategies.
	This significantly complicates notation.
	Additionally, and more importantly, the former proof crucially relied on the fact that for a state of player $\reach$, we know that its valuation is at least as large as that of any action, and dually for a state of player $\safe$, its valuation is at most as large as that of any action.
	In the concurrent setting, this is not true. The optimal strategies need not be maximizing nor minimizing the valuation and, moreover, they can be maximizing for one valuation and minimizing for another.
	Thus, we found a more general, and in fact simpler, way of proving that \enquote{no state in $\ecStates$ can depend on the outside}~\cite[Statement 5]{eisentrautValueIterationSimple2022} and deriving the contradiction.
	The crucial insight is that we can fix locally optimal strategies and then apply \Cref{arg:replaceStratOk}.
	
	\subparagraph{Notation for Bellman operator}
	Before we begin the formal proof, we establish a condensed notation for a number of terms in the Bellman operator:
	\[
	\rest(s,a,b,s',\strategy,\strategy*) = \transitions(\state,\action[\reach],\action[\safe])(\state<\prime>) \cdot 		\strategy[\reach](\state)(\action[\reach]) \cdot 		\strategy[\safe](\state)(\action[\safe]).
	\]
	Thus, the Bellman operator for some valuation $\valuation$ and pair of strategies $(\rho,\sigma)$ simplifies to 
	\[
	 \preop(\valuation)(s,\rho,\sigma) = \sum_{(\action[\reach],\action[\safe]) \in \actions} \sum_{\state<\prime> \in \states} 
	\valuation(s') \cdot \rest(s,a,b,s',\strategy,\strategy*).
	\]

	\subparagraph{The set $\ecStates$ with maximum difference}	
	We define the \emph{difference of a state} $s\in\states$ as $\Delta(s) \coloneqq \upperBound<\star>(s) - \valR(s)$. 
	Recall that $\valR$ is the least fixpoint and $\upperBound<\star>$ the greatest fixpoint  of $\preop$.
	Hence, we know that $\Delta(s) \geq 0$ for all states.
	Further, since we assume for contradiction that $\upperBound<\star> \neq \valR$, there exist states with $\Delta(s) > 0$.
	Thus, we can find a non-empty set of states with maximum difference:
	$\ecStates \coloneqq \{ \state \in \states | \Delta(\state) = \max_{\state \in \states} \Delta(\state)\}$.
	
	\subparagraph{A leaving pair of strategies decreases the difference}	
	Let $s \in \ecStates$ be a state in $\ecStates$.
	Let $(\strategy,\strategy*) \in (\strategies(\state) \times \strategies*(\state))$ be a pair of strategies such that $(s,\strategy,\strategy*) \; \leaves \;\ecStates$.
	Then, following this pair of strategies for one step decreases the difference, formally
	\begin{equation}\label{eq:differenceInDifference}
		\preop(\Delta)(s,\strategy,\strategy*) < \Delta(s).
	\end{equation}
		
	We prove this using the following chain of equations:
	\begin{align*}
		\allowdisplaybreaks
		&\preop(\Delta)(s,\strategy,\strategy*) \\
		&=\sum_{(\action[\reach],\action[\safe]) \in \actions} \sum_{\state<\prime> \in \states} 
		\Delta(s') \cdot \rest(s,a,b,s',\strategy,\strategy*) \tag{Definition of Bellman operator}\\
		&< \Delta(s). \tag{\Cref{arg:forStatement6ofEKKW}}
	\end{align*}
	\begin{argument}\label{arg:forStatement6ofEKKW}
		By assumption, we have that there exists a $t \in \states \setminus \ecStates$ such that this $t$ is reached with positive probability under the exiting strategies, i.e.\ 
		$\sum_{(\action[\reach],\action[\safe]) \in \actions} \rest(s,a,b,t,\strategy,\strategy*) > 0$.
		For this $t$ outside of $\ecStates$, we have $\Delta(t) < \Delta(s)$, since $\ecStates$ is defined as the set of  all states with maximum difference.
		Further, no state can have a difference larger than $\Delta(s)$.
		Furthermore, the remaining terms in the sum that the differences $\Delta(s')$ are multiplied with are a probability distribution, formally $\sum_{(\action[\reach],\action[\safe]) \in \actions} \sum_{\state<\prime> \in \states} 
		\rest(s,a,b,s',\strategy,\strategy*) = 1$.
		Thus, if all differences were equal to $\Delta(s)$, the sum would yield $\Delta(s)$. As one of the summands is smaller than $\Delta(s)$ and all others are at most $\Delta(s)$, we get that the sum has to be smaller than $\Delta(s)$.
	\end{argument}	
	
	\subparagraph{Without non-trivial \ECs, $\ecStates$ must be left}
	We have that $\ecStates \cap (\winning \cup \success) = \emptyset$: The difference is 0 for target states because both the value and the upper bound are equal to 1; and the difference is 0 for the sure winning region of player $\safe$, since the upper bound and value are equal to 0 (see \Cref{eq:upperBoundRec}). 
	Thus, since by assumption there are no \ECs in $\states \setminus (\winning \cup \success)$, the set $\ecStates$ cannot contain an \EC.
	Consequently, there exists a state $s \in \ecStates$ such that for all pairs of available actions $(a,b) \in \actionAssignment[\reach](s) \times \actionAssignment[\safe](s)$, we have $\support(\transitions(s,a,b)) \cap (\states \setminus \ecStates) \neq \emptyset$, i.e.\ there is a successor state outside of $\ecStates$.
	This is the case because, if all states had a pair of actions that stays in $\ecStates$, then there exists a pair of strategies that keeps the play inside a subset of $\ecStates$, which would then form an \EC. For a formal proof, we refer to \cite[Lemma 2]{eisentrautValueIterationSimple2022}. Note that, while their proof is for turn-based games, the definition of \EC is a graph theoretic notion where, intuitively, the players \enquote{work together} (formally, it is only about the existence of an edge in the underlying hypergraph), and thus the proof is applicable to \CSGs, too.
	In the following, we let $s$ denote such a state where all strategies are leaving.

	\subparagraph{Notation for locally optimal strategies}
	For any state and valuation, locally optimal strategies for both players exist.
	We establish a shorthand for the locally optimal strategies in state $s$ (the one obtained in the previous step) with respect to $\upperBound<\star>$ and $\valR$.
	For player $\reach$ and valuation $\upperBound<\star>$, we denote a locally optimal strategy by
	\[
	\rhoU \in \adjustlimits\argmax_{\rho\in\strategies(\state)} \inf_{\strategy* \in \strategies*(\state)} \preop(\upperBound<\star>)(s,\rho,\sigma).
	\]
	Similarly we denote a locally optimal strategy of player $\safe$ with respect to $\upperBound<\star>$ by 
	\[
	\sigmaU \in \adjustlimits\argmin_{\strategy* \in \strategies*(\state)} \sup_{\rho\in\strategies(\state)} \preop(\upperBound<\star>)(s,\rho,\sigma).
	\]
	Analogously, we define locally optimal strategies with respect to $\valR$, namely $\rhoV$ and $\sigmaV$, obtained by replacing $\upperBound<\star>$ with $\valR$ in the above definition.
	
	\subparagraph{Deriving the contradiction}
	Recall that $s \in \ecStates$ is a state where all available pairs of actions, and thus all strategies, leave $\ecStates$.
	We derive the contradiction $\Delta(s) < \Delta(s)$. %
	\begin{align*}
		\allowdisplaybreaks
		\Delta(s) &= \upperBound<\star>(s) - \valR(s) \tag{Definition of $\Delta$}\\
		&= \preop(\upperBound<\star>)(s) - \valR(s)  \tag{$\upperBound<\star>$ is fixpoint by \Cref{lem:Ustar-is-fixpoint}}\\
		&= \preop(\upperBound<\star>)(s,\rhoU,\sigmaU) - \valR(s) \tag{$(\rhoU,\sigmaU)$ locally optimal w.r.t. $\upperBound<\star>$}\\
		&\leq \preop(\upperBound<\star>)(s,\rhoU,\sigmaV) - \valR(s) \tag{\Cref{arg:replaceStratOk}}\\
		&= \preop(\upperBound<\star>)(s,\rhoU,\sigmaV) - \preop(\valR)(s) \tag{$\valR$ is fixpoint~\cite[Theorem~1]{Alfaro2004}}\\
		&= \preop(\upperBound<\star>)(s,\rhoU,\sigmaV) -\preop(\valR)(s,\rhoV,\sigmaV) \tag{$(\rhoV,\sigmaV)$ locally optimal w.r.t. $\valR$}\\
		&\leq \preop(\upperBound<\star>)(s,\rhoU,\sigmaV) -\preop(\valR)(s,\rhoU,\sigmaV) \tag{\Cref{arg:replaceStratOk}}\\
		&= \preop(\Delta)(s,\rhoU,\sigmaV) \tag{\Cref{arg:rewriteDelta}}\\
		&< \Delta(s). \tag{$(s,\rhoU,\sigmaV) \; \leaves \; \ecStates$ and \Cref{eq:differenceInDifference}}
	\end{align*}
	
	\begin{argument}\label{arg:replaceStratOk}
		This argument is used in two steps in the above chain of equations.
		For the first usage, observe that $\sigmaV$ can be at most as good as the optimal $\sigmaU$.
		More formally, recall $\sigmaU$ was chosen as the $\argmin$ of the Bellman operator with respect to $\upperBound<\star>$. 
		Thus, $\preop(\upperBound<\star>)(s,\rhoU,\sigmaU) \leq \preop(\upperBound<\star>)(s,\rhoU,\sigmaV)$.
		
		For the second usage, by the analogous argument we have $\preop(\valR)(s,\rhoV,\sigmaV) \geq \preop(\valR)(s,\rhoU,\sigmaV)$. 
		Since this term is the subtrahend of the subtraction, the overall expression can only become greater.		
	\end{argument}
	
	\begin{argument}\label{arg:rewriteDelta}
		This step follows from expanding the definition of the Bellman operator, rearranging the sums and collapsing the definition of Bellman operator.
		Formally, for all states q and strategy pairs $(\rho,\sigma)$ it holds that 
		\begin{align*}
			\allowdisplaybreaks
			\preop(\Delta)&(q,\rho,\sigma) = \sum_{(\action[\reach],\action[\safe]) \in \actions} \sum_{\state<\prime> \in \states} 
			\Delta(s') \cdot \rest(s,a,b,s',\strategy,\strategy*)
			\tag{Definition of Bellman operator}\\
			&=  \sum_{(\action[\reach],\action[\safe]) \in \actions} \sum_{\state<\prime> \in \states} 
			\left(\upperBound<\star>(s') - \valR(s')\right) \cdot \rest(s,a,b,s',\strategy,\strategy*)
			\tag{Definition of $\Delta$}
			\\
			&= \left(\sum_{(\action[\reach],\action[\safe]) \in \actions} \sum_{\state<\prime> \in \states} \upperBound<\star>(s') \cdot \rest(s,a,b,s',\strategy,\strategy*)\right) ~ - \\
			&\phantom{=} \left(\sum_{(\action[\reach],\action[\safe]) \in \actions} \sum_{\state<\prime> \in \states} \valR(s') \cdot \rest(s,a,b,s',\strategy,\strategy*)\right)
			\tag{Splitting the sum}
			\\
			&= \preop(\upperBound<\star>)(q,\rho,\sigma) - \preop(\valR)(q,\rho,\sigma).
			\tag{Definition of Bellman operator}
		\end{align*}
	\end{argument}
	
	\subparagraph{Summary}
	Starting from the assumption that $\upperBound<\star> \neq \valR$, we derived that there exists a set of states $\ecStates$ where the difference $\Delta$ between upper bound and value is maximized. 
	Further, a pair of strategies leaving $\ecStates$ decreases this difference. 
	However, since there are no ECs in $\ecStates$, there has to be a state where the optimal strategies for $\valR$ and $\upperBound<\star>$ leave, which allows us to derive a contradiction.
	Thus, the initial assumption is false, and we have $\upperBound<\star> = \valR$.
\end{proof}

\subsection{Convergence without Bloated End Components}\label{apx:convergence_no_bad_ecs}

\negateWeakDom*
\begin{proof}
	We only provide the proof for Player $\reach$, as the other one is analogous by exchanging the names of the strategy sets, replacing $\strategies$ with $\strategies*$ and vice versa.

	We assume that we do not have $\strategies[2] \dominates_{\valuation, \strategies[\safe]<\prime>} \strategies[1]$ under the set of counter-strategies $\strategies[\safe]<\prime>$ with respect to $\valuation$.
	Writing out the definition, this means that $\forall \strategy[1] \in \strategies[1]. \exists \strategy[2] \in \strategies[2]:$
	\begin{enumerate}
		\item[(i)] 	$\inf_{\strategy* \in \strategies[\safe]<\prime>}\preop(\valuation)(\state,\strategy[2], \strategy*) > \inf_{\strategy* \in \strategies[\safe]<\prime>}\preop(\valuation)(\state, \strategy[1], \strategy*)$, or
		\item[(ii)] $\forall \sigma^{\prime} \in \strategies[\safe]<\prime>$ we have $\preop(\valuation)(\state, \strategy[2], \sigma^{\prime}) \geq \preop(\valuation)(\state, \strategy[1], \sigma^{\prime})$.
	\end{enumerate}

	Our goal is to have that $\strategies[1] \dominateseq_{\valuation, \strategies[\safe]<\prime>} \strategies[2]$, formally:
	Formally, $\exists \strategy[2] \in \strategies[2]. \forall \strategy[1] \in \strategies[1]:$
	\[
	\inf_{\strategy* \in \strategies[\safe]<\prime>}\preop(\valuation)(\state,\strategy[1], \strategy*) \leq \inf_{\strategy* \in \strategies[\safe]<\prime>}\preop(\valuation)(\state, \strategy[2], \strategy*).
	\]
	
	Both conditions of negated weak dominance imply our goal.
	The only remaining problem is the order of quantifiers. 
	However, the choice of $\rho_2$ does not depend on $\rho_1$, and we can always pick $\rho_2$ as the strategy that maximizes $\inf_{\strategy* \in \strategies[\safe]<\prime>}\preop(\valuation)(\state,\strategy[2], \strategy*)$.	
	Thus, we can exchange the order of quantifiers and prove our goal.
\end{proof}

\negateBEC*
\begin{proof}
	
	Since $\ecStates<\prime>$ is not a \BEC, we know that at some $\state \in \ecStates<\prime>$ it holds that $\badActs_{\upperBound<\star>}(\ecStates<\prime>, \state)= \emptyset$. 
	Fix $\state$ to be such a state.
	Every strategy $\strategy<\prime> \in \strategies(\states)$ must violate at least one of the 3 conditions of \cref{def:badAction}.
	We write out the negations:
	\begin{enumerate}
		\item[(i)] $\strategy<\prime> \notin \stayStratsR(\ecStates<\prime>, \state)$, i.e.\ the strategy is leaving $\strategy<\prime> \in \leaveStratsR(\ecStates<\prime>,\state)$. 
		\item[(ii)] We do not have $\strategies(\state) \setminus \{\strategy<\prime>\} \dominateseq_{\strategies*(\state)} \{\strategy<\prime>\}$. By contraposition of \cref{lem:negate-weak-dominance}, this implies $\{\strategy<\prime>\} \dominates_{\strategies*(\state)}\strategies(\state) \setminus \{\strategy<\prime>\}$, so the strategy is sub-optimal.
		\item[(iii)] We do not have $\leaveStratsR(\ecStates<\prime>, \state) \dominates_{\strategies*(\state)} \{\strategy<\prime>\}$. By \cref{lem:negate-weak-dominance}, this implies $\{\strategy<\prime>\} \dominateseq_{\strategies*(\state)}\leaveStratsR(\ecStates<\prime>, \state)$, i.e.\ there are leaving strategies that are not worse than $\rho'$. Note that in particular, this implies that $\leaveStratsR(\ecStates<\prime>, \state)$ is non-empty.
	\end{enumerate}

	Our assumption gives us the disjunction over the three violated conditions.
	We proceed by a case distinction, always assuming that all strategies violate a certain condition, which allows us to prove our goal, or, if there exists a strategy satisfying the condition, we continue with the next one.
	
	\textbf{Case \enquote{Not (i)}:} 
	If all strategies violate Condition (i), that means all strategies are leaving, i.e.\ $\leaveStratsR(s) = \strategies(s)$. Thus, since there always exist optimal strategies, by picking an optimal $\rho \in \strategies$ we naturally have $\strategies(s) \dominateseq_{\valuation,\strategies*(s)} \{\rho\}$.
	
	\textbf{Case \enquote{(i), but not (iii)}:}
	We assume there exists strategies that satisfies Condition (i), but all strategies that satisfy it violate Condition (iii).
	This means that for every non-leaving strategy, the set of leaving strategies is not worse than it. 
	As this holds for all non-leaving strategies, we have $\stayStratsR(s) \dominateseq_{\valuation,\strategies*(s)} \leaveStratsR(s)$.
	Using that $\strategies(s) = \stayStratsR(s) \cup \leaveStratsR(s)$ and the set of leaving strategies trivially is not worse than itself, we obtain:
	$\strategies(s) \dominateseq_{\valuation,\strategies*(s)} \leaveStratsR(s)$.
	Moreover, the set of leaving strategies is non-empty, since the definition of not worse requires that there exists a strategy in the right-hand set.
	This proves our goal.
	
	\textbf{Case \enquote{(i) and (iii), but not (ii)}:}
	We assume there exists strategies that satisfy Condition (i) and (iii), but all these strategies violate Condition (ii). 
	This case cannot happen, and below we derive a contradiction.
	This completes our cast distinction, since every strategy has to violate at least one of the three conditions.
	
	Our assumption is that there exists a non-leaving strategy that weakly dominates the set of all leaving strategies. 
	However, every such strategy is suboptimal, as by violating Condition (ii) it is weakly dominated by all other strategies.
	This is a contradiction, because then there are no optimal strategies.
	More formally, if we assume the optimal strategy is leaving, this is a contradiction, because there exists a non-leaving strategy dominating the set of leaving strategies.
	And if we assume the optimal strategy is non-leaving, this is a contradiction, because every non-leaving strategy is weakly dominated by the set of others.
	
\end{proof}

\nonConvImpliesBEC*
\begin{proof}
	\subparagraph{Intuition and outline:}
This proof builds on the proof of \cref{theo:BVInoEC}.
There, we constructed a set $\mathcal{X}$ maximizing the difference between $\upperBound<\star>$ and $\valR$ and showed that if there is a pair of optimal strategies leaving leaving $\mathcal{X}$, then we can derive a contradiction: The upper bound decreases, which contradicts the fact that it is a fixpoint.
In the context of the other proof, that allowed us to show that without \ECs, \VI converges, because without \ECs it is impossible to have a set of states where all optimal strategies stay in that set.

In the presence of \ECs, states can indeed have a positive difference between $\upperBound<\star>$ and $\valR$, see e.g.\ \cref{ex:potentialProblemECs}. 
Our goal is to prove that at least one of these \ECs is bloated.
Thus, we assume for contradiction that no \EC is bloated under $\upperBound<\star>$.
Thus, by \cref{lem:negate-BEC}, there is an optimal leaving strategy for player $\reach$.
Using that, we can repeat the argument from \cref{theo:BVInoEC}, showing that in this case $\upperBound<\star>$ would decrease. Again, this is a contradiction because it is a fixpoint of applying Bellman updates (\cref{lem:Ustar-is-fixpoint}).
Thus, the initial assumption that no \EC is bloated is false, and we can conclude that there exists a \BEC.

	\subparagraph{Establishing the Context}
	As in the proof of \cref{theo:BVInoEC}, let $\mathcal{X}:= \{ \state \in \states \;|\; \Delta(\state) = \max_{\state \in \states} \Delta(\state)\}$ be the set of states with maximum difference.
	We denote the maximum by $\Delta^\textsf{max}$, and our assumption yields that $\Delta^\textsf{max}>0$.
	Note that this implies that $\mathcal{X} \cap (\winning \cup \success) = \emptyset$, since for those states, their value is set correctly by initialization, and their difference is 0. Thus, if we find a \BEC that is a subset of $\mathcal{X}$, it also satisfies the additional condition of being non-trivial, i.e.\ not in $(\winning \cup \success)$. 
	The contraposition of \cref{theo:BVInoEC} yields that there has to be an \EC in $\ecStates$.

	\subparagraph{Bottom MECs}
	To derive the contradiction, in the following, we consider \ECs with a particular property, namely \ECs that are \emph{bottom in $\ecStates$}. A \MEC $\ecStates<\prime>$ is bottom in $\ecStates$ if the successors of a pair of strategies that leaves the \MEC reaches states outside of $\ecStates$ with positive probability. Intuitively, a bottom \MEC in $\ecStates$ is a \MEC, such that after leaving it, none of the successors is part of another \MEC in $\ecStates$. One can compute such \ECs using the \MEC decomposition of $\ecStates$, ordering them topologically and picking one at the end of a chain.
	
	Let $\ecStates<\prime>$ be a bottom \MEC with $\ecStates<\prime> \subseteq \ecStates$. $\ecStates<\prime>$ exists because by assumption $\ecStates$ contains at least one \EC, thus, there also has to exist an \EC that is bottom in $\ecStates$.

	\subparagraph{Optimal Leaving Strategies in Non-\BECs}
	We use the assumption for contradiction to say that $\ecStates<\prime>$ is not bloated with respect to~$\upperBound<\star>$.
	Then, using \cref{lem:negate-BEC}, we know that there exists a state $s \in \ecStates<\prime>$ where an optimal strategy $\rhoU$ exists that is leaving $\ecStates<\prime>$.
	Moreover, since $\ecStates<\prime>$ is a bottom MEC in $\ecStates$, we also have that it is leaving with respect to $\ecStates$.

	\subparagraph{Deriving the Contradiction}
	Using these facts, we can exactly repeat the argument used in the proof of \cref{theo:BVInoEC} under the paragraph-heading \enquote{Deriving the Contradiction}.
	Recall we denote locally optimal strategies with respect to $\upperBound<\star>$ by $\rhoU, \sigmaU$ (and we just proved $\rhoU$ is leaving for all counter-strategies), and analogously locally optimal strategies with respect to $\valR$, by $\rhoV$ and $\sigmaV$.
	We highlight that \cref{eq:differenceInDifference}, \Cref{arg:replaceStratOk} and \Cref{arg:rewriteDelta} from \cref{theo:BVInoEC} are applicable in the context of this proof, too.
	\begin{align*}
		\allowdisplaybreaks
		\Delta(s) &= \upperBound<\star>(s) - \valR(s) \tag{Definition of $\Delta$}\\
		&= \preop(\upperBound<\star>)(s) - \valR(s)  \tag{$\upperBound<\star>$ is fixpoint by \Cref{lem:Ustar-is-fixpoint}}\\
		&= \preop(\upperBound<\star>)(s,\rhoU,\sigmaU) - \valR(s) \tag{$(\rhoU,\sigmaU)$ locally optimal w.r.t. $\upperBound<\star>$}\\
		&\leq \preop(\upperBound<\star>)(s,\rhoU,\sigmaV) - \valR(s) \tag{\Cref{arg:replaceStratOk} in \cref{theo:BVInoEC}}\\
		&= \preop(\upperBound<\star>)(s,\rhoU,\sigmaV) - \preop(\valR)(s) \tag{$\valR$ is fixpoint~\cite[Theorem~1]{Alfaro2004}}\\
		&= \preop(\upperBound<\star>)(s,\rhoU,\sigmaV) -\preop(\valR)(s,\rhoV,\sigmaV) \tag{$(\rhoV,\sigmaV)$ locally optimal w.r.t. $\valR$}\\
		&\leq \preop(\upperBound<\star>)(s,\rhoU,\sigmaV) -\preop(\valR)(s,\rhoU,\sigmaV) \tag{\Cref{arg:replaceStratOk} in \cref{theo:BVInoEC}}\\
		&= \preop(\Delta)(s,\rhoU,\sigmaV) \tag{\Cref{arg:rewriteDelta} in \cref{theo:BVInoEC}}\\
		&< \Delta(s). \tag{$(s,\rhoU,\sigmaV) \; \leaves \; \ecStates$ and \Cref{eq:differenceInDifference}}
	\end{align*}

	Now that we have derived a contradiction, our initial assumption that all \ECs are not \BECs is wrong, so we know there exists a \BEC $\mathcal{X'} \subseteq (\states \setminus (\success \cup \winning))$ with respect to~$\upperBound<\star>$.
	This concludes the proof.
	
	As a side note, we remark that it is indeed possible that there is only one \BEC that causes many states, even some not in an \EC, to have a positive difference $\Delta$.
	For an example, we refer to~\cite[Fig.~4]{eisentrautValueIterationSimple2022}.

\end{proof}

\subsection{Soundness of \procname{DEFLATE}}\label{apx:soundness_deflate}

	The following is a technical lemma that is needed to show the soundness of deflation.
	Intuitively, it says that the value of all states in an EC needs to depend on some exit.
	Note that there can be states whose value is higher than their own exit value, namely if they can reach a better exit.
	However, this cannot be the case for all states, but there must be some whose value is less than or equal to their exit value (first condition), and in fact no state can have a higher value than these states that actually depend on exiting (second condition).
	The proof is very technical, as essentially it requires unfolding all the definitions, and thereby also unfolding all the included case distinctions.

	\begin{lemma}[No state has a larger value than that of an exit from its EC]\label{lem:exitPossibleUnderVr}
		Let $\ecStates \subseteq \states \setminus (\success \cup \winning)$ be an \EC. Then, it holds that 
		\begin{itemize}
			\item[(i)] $\ecStates<\prime> \coloneqq \{\state \in \ecStates \mid \valR(\state) \leq \exitVal[\valR](\ecStates, \state)\} \neq \emptyset$, and
			\item[(ii)] $\max_{\state \in \ecStates<\prime>}\valR(\state) \geq \max_{\state \in \ecStates \setminus \ecStates<\prime>} \valR(\state)$.
		\end{itemize}
	\end{lemma}
	
	\begin{proof}
		We prove the lemma by contradiction, i.e. we assume that one of the two conditions posed by the lemma is violated. We make the following case distinction: 
		\begin{enumerate}[leftmargin=2cm, resume]
			\item[(i)] $\ecStates<\prime> = \emptyset$; and 
			\item[(ii)] $\ecStates<\prime> \neq \emptyset$ but $\max_{\state \in \ecStates<\prime>}\valR(\state) < \max_{\state \in \ecStates \setminus \ecStates<\prime>} \valR(\state)$.
		\end{enumerate}
		
		\begin{description}[style=nextline]
			\item[Case (i)] In this case it holds that $\ecStates<\prime> = \emptyset$, i.e. for all $\state \in \ecStates$ we have $\valR(\state) > \exitVal[\valR](\ecStates, \state)$. Recall that in \cref{def:exitVal} if $\badActs_{\valR}(\ecStates, \state) = \emptyset$ is true at some state $\state \in \ecStates$, then the exit value is given by the value of the matrix game played at that state.
			
			Consequently, as for all $\state \in \ecStates$ it holds that $\valR(\state) > \exitVal[\valR](\ecStates, \state)$, at all states $\state \in \ecStates$ it must hold that $\badActs_{\valR}(\ecStates, \state) \neq \emptyset$ because otherwise we would obtain the following contradiction: $\exitVal[\valR](\ecStates, \state) = \sup_{\strategy \in \strategies(\state)}\inf_{\strategy* \in \strategies*(\state)}\preop(\valR)(\state, \strategy, \strategy*) = \valR(\state)$.
			
			We proceed with the assumption that for all  $\state \in \ecStates$ it holds that $\badActs_{\valR}(\ecStates, \state) \neq \emptyset$. Since by the case assumption at all  $\state \in \ecStates$ it holds that $\valR(\state) > \exitVal[\valR](\ecStates, \state)$ the true values of the states in $\ecStates$ are attainable with the hazardous and trapping strategies. More formally, at each $\state \in \ecStates$ the following chain of equations holds.
			\begin{align*}
				\valR(\state) &= \adjustlimits\sup_{\strategy \in \strategies(\state)} \inf_{\strategy* \in \strategies*(\state)}\preop(\valR)(\state, \strategy, \strategy*)\tag{$\valR$ is a fixpoint \& case assumption: $\valR(\state) > \exitVal[\valR](\ecStates, \state)$}\\
				&= \adjustlimits\sup_{\strategy \in \badActs_{\valR}(\ecStates, \state)} \inf_{\strategy* \in \strategies*(\state)}\preop(\valR)(\state, \strategy, \strategy*) \tag{$\badActs_{\valR}(\ecStates, \state)$ are optimal by \Cref{def:badAction}}\\
				&= \adjustlimits\sup_{\strategy \in \badActs_{\valR}(\ecStates, \state)} \inf_{\strategy* \in \counter_{\valR}(\ecStates, \state))}\preop(\valR)(\state, \strategy, \strategy*). \tag{$\counter_{\valR}(\ecStates, \state)$ are optimal by \Cref{def:badCounterStrategy}}\\
			\end{align*}
			Thus, for player $\reach$ staying in $\ecStates$ is optimal and since no target state is contained in $\ecStates$ it has to hold that $\valR(\state) = 0$ for all $\state \in \ecStates$. However, this is a contradiction to the assumption that $\valR(\state) > \exitVal[\valR](\ecStates, \state)$ since $\exitVal[\valR](\ecStates, \state) \geq 0$ (as all possible valuations are non-negative and the value of the exiting sub-game is given by the maximum between 0 and $\exitVal$ - see \cref{def:exitVal} and \cref{def:exitingSG}). Thus, the case assumption that $\ecStates<\prime> = \emptyset$ must be false.
			
			\item[Case (ii)] In this case it holds that $\ecStates<\prime> \neq \emptyset$ but 
			$\max_{\state \in \ecStates<\prime>}\valR(\state) < \max_{\state \in \ecStates \setminus \ecStates<\prime>} \valR(\state)$ is true.
			
			We make the following case distinction: (ii.a) $\ecStates \setminus \ecStates<\prime>$ is a \BEC; and (ii.b) $\ecStates \setminus \ecStates<\prime>$ is not a \BEC.
			
			\begin{description}[style=nextline]
				\item[Case (ii.a)] In this case $\ecStates \setminus \ecStates<\prime>$ is a \BEC, i.e. for all $\state \in \ecStates \setminus \ecStates<\prime>$ it holds that $\badActs_{\valR}(\ecStates \setminus \ecStates<\prime>, \state) \neq \emptyset$.
				
				In case it holds that $\valR(\state) > \exitVal[\valR](\ecStates \setminus \ecStates<\prime>, \state)$ for all $\state \in \ecStates \setminus \ecStates<\prime>$, then since no target state is contained in $\ecStates$ and therefore neither in $\ecStates \setminus \ecStates<\prime>$, it must be true that $\valR(\state) = 0$ for all $\state \in \ecStates \setminus \ecStates<\prime>$. However, similarly as in Case (i), this is a contradiction to the assumption that $\valR(\state) > \exitVal[\valR](\ecStates \setminus \ecStates<\prime>, \state)$ for all $\state \in \ecStates \setminus \ecStates<\prime>$ as $\exitVal[\valR](\ecStates \setminus \ecStates<\prime>, \state) \geq 0$ for all $\state \in \ecStates \setminus \ecStates<\prime>$.
				
				Consequently, there has to exist $\state \in \ecStates \setminus \ecStates<\prime>$ such that $\valR(\state) \leq \exitVal[\valR](\ecStates \setminus \ecStates<\prime>, \state)$.

				We make another case distinction:
				\begin{enumerate}[leftmargin=2cm,resume]
					\item[(ii.a.1)] for all $\state<\prime> \in \argmax_{\state \in \ecStates \setminus \ecStates<\prime>} \valR(\state)$ it holds that $\valR(\state<\prime>) > \exitVal[\valR](\ecStates \setminus \ecStates<\prime>, \state<\prime>)$; and
					\item[(ii.a.2)] there exists $\state<\prime> \in \argmax_{\state \in \ecStates \setminus \ecStates<\prime>} \valR(\state)$ such that $\valR(\state<\prime>) \leq \exitVal[\valR](\ecStates \setminus \ecStates<\prime>, \state<\prime>)$.
				\end{enumerate}
				
				\begin{description}
					\item[Case (ii.a.1)] In this case for all $\state<\prime> \in \argmax_{\state \in \ecStates \setminus \ecStates<\prime>} \valR(\state)$ it holds that $\valR(\state<\prime>) > \exitVal[\valR](\ecStates \setminus \ecStates<\prime>, \state<\prime>)$.
					
					Thus, at each $\state<\prime> \in \argmax_{\state \in \ecStates \setminus \ecStates<\prime>} \valR(\state<\prime>)$ the following chain of equations holds.
					\begin{align*}
						\valR(\state<\prime>) &= \adjustlimits\sup_{\strategy \in \strategies(\state<\prime>)} \inf_{\strategy* \in \strategies*(\state<\prime>)}\preop(\valR)(\state<\prime>, \strategy, \strategy*)\tag{$\valR$ is a fixpoint}\\
						&= \adjustlimits\sup_{\strategy \in \badActs_{\valR}(\ecStates \setminus \ecStates<\prime>, \state<\prime>)} \inf_{\strategy* \in \strategies*(\state<\prime>)}\preop(\valR)(\state<\prime>, \strategy, \strategy*) \tag{$\badActs_{\valR}(\ecStates \setminus \ecStates<\prime>, \state<\prime>)$ are optimal by \Cref{def:badAction}}\\
						&= \adjustlimits\sup_{\strategy \in \badActs_{\valR}(\ecStates \setminus \ecStates<\prime>, \state<\prime>)} \inf_{\strategy* \in \counter_{\valR}(\ecStates \setminus \ecStates<\prime>, \state<\prime>))}\preop(\valR)(\state<\prime>, \strategy, \strategy*). \tag{$\counter_{\valR}(\ecStates \setminus \ecStates<\prime>, \state<\prime>)$ are optimal by \Cref{def:badCounterStrategy}}\\
					\end{align*}				
					Thus, at all states that attain the highest value among $\ecStates \setminus \ecStates<\prime>$, it is optimal for both players to choose strategies that together are staying in $\ecStates \setminus \ecStates<\prime>$. However, since $\ecStates$ and thus $\ecStates \setminus \ecStates<\prime>$ does not belong to the winning region of Player $\safe$, leaving $\ecStates \setminus \ecStates<\prime>$ has to be possible. Thus, from each state $\state<\prime> \in \argmax_{\state \in \ecStates \setminus \ecStates<\prime>} \valR(\state)$, Player $\reach$ has to possess an optimal strategy that leads to a state $\state<\prime\prime> \in \ecStates \setminus \ecStates<\prime>$ where leaving $\ecStates \setminus \ecStates<\prime>$ is possible, i.e. where $\valR(\state<\prime\prime>) \leq \exitVal[\valR](\ecStates \setminus \ecStates<\prime>, \state<\prime\prime>)$ holds. Thus, at state $\state<\prime\prime>$ the highest value also has to be attainable. However, this is a contradiction to the case assumption that for all $\state<\prime> \in \argmax_{\state \in \ecStates \setminus \ecStates<\prime>} \valR(\state)$ it holds that $\valR(\state<\prime>) > \exitVal[\valR](\ecStates \setminus \ecStates<\prime>, \state<\prime>)$. 
					\item[Case (ii.a.2)] In this case  there exists $\state<\prime> \in \argmax_{\state \in \ecStates \setminus \ecStates<\prime>} \valR(\state)$ such that $\valR(\state<\prime>) \leq \exitVal[\valR](\ecStates \setminus \ecStates<\prime>, \state<\prime>)$. 
					Let $\ecStates<\prime\prime> \coloneqq \{\state<\prime> \in \argmax_{\state \in \ecStates \setminus \ecStates<\prime>}\valR(\state) \mid \valR(\state<\prime>) \leq \exitVal[\valR](\ecStates \setminus \ecStates<\prime>, \state<\prime>)\}$. Then, we need to make another case distinction:	
					\begin{enumerate}[leftmargin=2cm,resume]
						\item[(ii.a.2.1)] for all $\state \in \ecStates<\prime\prime>$ it holds that $\deflStrats_{\valR}(\ecStates \setminus \ecStates<\prime>, \state) = \emptyset$; and
						\item[(ii.1.2.2)] there exists $\state<\prime> \in \ecStates<\prime\prime>$ such that $\deflStrats_{\valR}(\ecStates \setminus \ecStates<\prime>, \state<\prime>) \neq \emptyset$.
					\end{enumerate}
					
					\begin{description}
						\item[Case (ii.a.2.1)] In this case for all $\state \in \ecStates<\prime\prime>$ it holds that $\deflStrats_{\valR}(\ecStates \setminus \ecStates<\prime>, \state) = \emptyset$. Then, since $\ecStates$ and so $\ecStates \setminus \ecStates<\prime>$ do not belong to the winning region of Player $\safe$, there has to exist a state $\state<\prime\prime> \in \ecStates \setminus \ecStates<\prime>$ such that $\deflStrats_{\valR}(\ecStates \setminus \ecStates<\prime>, \state<\prime\prime>) \neq \emptyset$ so leaving $\ecStates \setminus \ecStates<\prime>$ is possible. However, then $\state<\prime\prime> \in \argmax_{\state \in \ecStates \setminus \ecStates<\prime>}\valR(\state)$ must hold which is a contradiction to the case assumption.
						
						\item[Case (ii.a.2.2)] In this case there exists $\state<\prime> \in \ecStates<\prime\prime>$ such that $\deflStrats_{\valR}(\ecStates \setminus \ecStates<\prime>, \state<\prime>) \neq \emptyset$. Then, the following chain of equations holds.
						\begin{align*}
							\valR(\state<\prime>) &\leq \exitVal[\valR](\ecStates \setminus \ecStates<\prime>, \state<\prime>) \tag{Case assumption}\\
							&= \max(0, \val(\hat{\mathsf{Z}}(\state<\prime>))) \tag{By case assumption (ii.a): $\badActs_{\valR}(\ecStates \setminus \ecStates<\prime>, \state<\prime>) \neq \emptyset$}\\
							&= \val(\hat{\mathsf{Z}}(\state<\prime>)) \tag{Case assumption: $\deflStrats_{\valR}(\ecStates \setminus \ecStates<\prime>, \state<\prime>) \neq \emptyset$}\\
							&= \adjustlimits\sup_{\strategy \in \deflStrats_{\valR}(\ecStates \setminus \ecStates<\prime>, \state<\prime>)}\inf_{\strategy* \in \counter_{\valR}(\ecStates \setminus \ecStates<\prime>, \state<\prime>)}\preop(\valR)(\state<\prime>, \strategy, \strategy*) \tag{Value of the exiting sub-game --- see \Cref{def:exitVal}}\\
							&= \adjustlimits\sup_{\strategy \in \deflStrats_{\valR}(\ecStates \setminus \ecStates<\prime>, \state<\prime>)}\inf_{\strategy* \in \counter_{\valR}(\ecStates \setminus \ecStates<\prime>, \state<\prime>)} \sum_{(\actionR, \actionS) \in \actions} \sum_{\state \in \ecStates \setminus \ecStates<\prime>} \underbrace{\valR(\state)}_{\leq \max_{\state<\prime\prime> \in \ecStates \setminus \ecStates<\prime>} \valR(\state<\prime\prime>)} \cdot \transitions(\state<\prime>, \actionR, \actionS)(\state) \cdot \strategy(\actionR) \cdot \strategy*(\actionS) \\
							&+ \sum_{\state \in \ecStates<\prime>} \underbrace{\valR(\state)}_{< \max_{\state<\prime\prime> \in \ecStates \setminus \ecStates<\prime>} \valR(\state<\prime\prime>)} \cdot \transitions(\state<\prime>, \actionR, \actionS)(\state) \cdot \strategy(\actionR) \cdot \strategy*(\actionS) \tag{Def. of $\preop$ and case assumption (ii): $\max_{\state \in \ecStates<\prime>}\valR(\state) < \max_{\state \in \ecStates \setminus \ecStates<\prime>} \valR(\state)$} \\
							&< \max_{\state \in \ecStates \setminus \ecStates<\prime>} \valR(\state). \tag{Everything sums up to 1}
						\end{align*}
						Thus, from the case assumption we derived a contradiction because $\state<\prime> \in \argmax_{\state \in \ecStates \setminus \ecStates<\prime>} \valR(\state)$.
					\end{description}
					
				\end{description}
				
				\item[Case (ii.b)] In this case $\ecStates \setminus \ecStates<\prime>$ is not a \BEC, i.e. there exists $\state<\prime> \in \ecStates \setminus \ecStates<\prime>$ such that $\badActs_{\valR}(\ecStates \setminus \ecStates<\prime>, \state<\prime>) = \emptyset$.
				
				We need another case distinction: 
				\begin{enumerate}[leftmargin=2cm, resume]
					\item[(ii.b.1)] for all $\state<\prime> \in \argmax_{\state \in \ecStates \setminus \ecStates<\prime>} \valR(\state)$ it holds that $\badActs_{\valR}(\ecStates \setminus \ecStates<\prime>, \state<\prime>) \neq \emptyset$; and 
					\item[(ii.b.2)] there exists  $\state<\prime> \in \argmax_{\state \in \ecStates \setminus \ecStates<\prime>} \valR(\state)$ such that $\badActs_{\valR}(\ecStates \setminus \ecStates<\prime>, \state<\prime>) =\emptyset$.
				\end{enumerate}
				
				\begin{description}
					\item[Case (ii.b.1)] In this case we have that for all $\state<\prime> \in \argmax_{\state \in \ecStates \setminus \ecStates<\prime>} \valR(\state)$ it holds that $\badActs_{\valR}(\ecStates \setminus \ecStates<\prime>, \state<\prime>) \neq \emptyset$.
					
					Then, at each $\state<\prime> \in \argmax_{\state \in \ecStates \setminus \ecStates<\prime>} \valR(\state)$ the following chain of equations holds.
					\begin{align*}
						\valR(\state<\prime>) &=  \adjustlimits\sup_{\strategy \in \badActs_{\valR}(\ecStates \setminus \ecStates<\prime>, \state<\prime>)} \inf_{\strategy* \in \strategies*(\state<\prime>)}\preop(\valR)(\state<\prime>, \strategy, \strategy*) \tag{$\badActs_{\valR}(\ecStates \setminus \ecStates<\prime>, \state<\prime>)$ are optimal by \Cref{def:badAction}}\\
						&= \adjustlimits\sup_{\strategy \in \badActs_{\valR}(\ecStates \setminus \ecStates<\prime>, \state<\prime>)} \inf_{\strategy* \in \counter_{\valR}(\ecStates \setminus \ecStates<\prime>, \state<\prime>))}\preop(\valR)(\state<\prime>, \strategy, \strategy*). \tag{$\counter_{\valR}(\ecStates \setminus \ecStates<\prime>, \state<\prime>)$ are optimal by \Cref{def:badCounterStrategy}}\\
					\end{align*}				
					Thus, at all states that attain the highest value among $\ecStates \setminus \ecStates<\prime>$ it is optimal for both players to choose strategies that together are staying in $\ecStates \setminus \ecStates<\prime>$. However, since $\ecStates$ and thus $\ecStates \setminus \ecStates<\prime>$ does not belong to the winning region of Player $\safe$, leaving $\ecStates \setminus \ecStates<\prime>$ has to be possible. Thus, from each state $\state<\prime> \in \argmax_{\state \in \ecStates \setminus \ecStates<\prime>} \valR(\state)$, Player $\reach$ has to possess an optimal strategy that leads to a state $\state<\prime\prime> \in \ecStates \setminus \ecStates<\prime>$ where leaving $\ecStates \setminus \ecStates<\prime>$ is possible, i.e. where $\badActs_{\valR}(\ecStates \setminus \ecStates<\prime>, \state<\prime\prime>) = \emptyset$ holds. Thus, at state $\state<\prime\prime>$ the highest value also has to be attainable. However, this is a contradiction to the case assumption that for all $\state<\prime> \in \argmax_{\state \in \ecStates \setminus \ecStates<\prime>} \valR(\state)$ it holds that $\badActs_{\valR}(\ecStates \setminus \ecStates<\prime>, \state<\prime>) \neq \emptyset$. 
					
					\item[Case (ii.b.2)] In this case there exists  $\state<\prime> \in \argmax_{\state \in \ecStates \setminus \ecStates<\prime>} \valR(\state)$ such that $\badActs_{\valR}(\ecStates \setminus \ecStates<\prime>, \state<\prime>) =\emptyset$.
					
					Then, the following chain of equations holds.
					\begin{align*}
						\valR(\state<\prime>) &= \adjustlimits\sup_{\strategy \in \strategies(\state<\prime>)}\inf_{\strategy* \in \strategies*(\state<\prime>)} \preop(\valR)(\state<\prime>, \strategy, \strategy*) \tag{$\valR$ is a fixpoint}\\
						&= \adjustlimits\sup_{\strategy \in \strategies(\state<\prime>)}\inf_{\strategy* \in \strategies*(\state<\prime>)} \sum_{(\actionR, \actionS)\in \actions} \sum_{\state \in \ecStates\setminus\ecStates<\prime>} \underbrace{\valR(\state)}_{\leq \max_{\state<\prime\prime> \in \ecStates\setminus\ecStates<\prime>}\valR(\state<\prime\prime>)} \cdot \transitions(\state<\prime>, \actionR, \actionS)(\state) \cdot \strategy(\actionR) \cdot \strategy*(\actionS)\\
						&+ \sum_{\state \in\ecStates<\prime>} \underbrace{\valR(\state)}_{< \max_{\state<\prime\prime> \in \ecStates\setminus\ecStates<\prime>}\valR(\state<\prime\prime>)}  \cdot \transitions(\state<\prime>, \actionR, \actionS)(\state) \cdot \strategy(\actionR) \cdot \strategy*(\actionS)\\
						&< \max_{\state<\prime\prime> \in \ecStates\setminus\ecStates<\prime>}\valR(\state<\prime\prime>). \tag{Everything sums up to 1 and $\badActs_{\valR}(\ecStates \setminus \ecStates<\prime>, \state<\prime>) =\emptyset$ thus at least one successor state is in $\ecStates<\prime>$}
					\end{align*}
					Thus, from the case assumption we derived a contradiction because $\state<\prime> \in \argmax_{\state \in \ecStates \setminus \ecStates<\prime>} \valR(\state)$.
				\end{description}
			\end{description}
		\end{description}
		
		Thus, every case leads to a contradiction which concludes the proof.
	\end{proof}

	Using \cref{lem:exitPossibleUnderVr}, we now prove that for every EC $\ecStates$, we have that for an upper bound $\upperBound$, the best exit value of the EC is also an upper bound.
	This is the crucial ingredient for showing that deflation cannot decrease a valuation below the value $\valR$.
	
	\begin{lemma}\label{lem:eVUgeqeVV}
		Let $\ecStates \subseteq \states \setminus (\success \cup \winning)$ be an \EC, and $\upperBound \in [0,1]^{\mid \states \mid}$ be a valuation with $\upperBound \geq \valR$.
		Then, for all states $\state\in\bestExits[\valR](\ecStates)$, we have $\exitVal[\upperBound](\ecStates, \state) \geq \exitVal[\valR](\ecStates, \state) $.
	\end{lemma}
	\begin{proof}
		Let $\state \in \bestExits[\valR](\ecStates)$. Our goal is to show that $\exitVal[\upperBound](\ecStates,s) \geq \exitVal[\valR](\ecStates, \state)$.
		Since for the estimation of $\exitVal[\upperBound](\ecStates,s)$ and $\exitVal[\valR](\ecStates, \state)$ the sets of strategies, $\counter_{\valR}(\ecStates, \state), \counter_{\upperBound}(\ecStates, \state),\badActs_{\valR}(\ecStates, \state)$, and $ \badActs_{\upperBound}(\ecStates, \state)$, which can be empty or non-empty, we have to consider all possible combinations.
		We consider the following four main cases, which we further analyse with respect to the $\upperBound$ sets of strategies where necessary:
		\begin{description}
			\item[Case (I)] $\counter_{\valR}(\ecStates, \state) = \emptyset$ and $\badActs_{\valR}(\ecStates, \state) = \emptyset$,
			\item[Case (II)] $\counter_{\valR}(\ecStates, \state) = \emptyset$ and $\badActs_{\valR}(\ecStates, \state) \neq \emptyset$,
			\item[Case (III)] $\counter_{\valR}(\ecStates, \state) \neq \emptyset$ and $\badActs_{\valR}(\ecStates, \state) = \emptyset$, and
			\item[Case (IV)] $\counter_{\valR}(\ecStates, \state) \neq \emptyset$ and $\badActs_{\valR}(\ecStates, \state) \neq \emptyset$.
		\end{description}
		For each case we show that either the case is impossible or that the statement of the lemma holds.
		
		\begin{description}
			\item[Case (I)] $\counter_{\valR}(\ecStates, \state) = \emptyset$ and $\badActs_{\valR}(\ecStates, \state) = \emptyset$.
			
			Let $\sigma^\prime \in \strategies*(\state)$ be an optimal Player $\safe$ strategy under $\valR$. Since $\counter_{\valR}(\ecStates, \state) = \emptyset$, at least on of the two conditions posed by the definition of trapping strategies (\cref{def:badCounterStrategy}) must be violated under $\valR$.
			
			Since $\badActs_{\valR}(\ecStates, \state) = \emptyset$, Condition (ii) of \cref{def:badCounterStrategy}, i.e. $\forall \strategy \in \badActs_{\valR}(\ecStates, \state): (\state, \strategy, \sigma^\prime)\; \stays\; \ecStates$, is trivially satisfied. Consequently, Condition (i) must be violated, i.e., it must hold that there exists no optimal strategy for Player $\safe$ which cannot be true. Thus, this case is impossible.
			
			\item[Case (II)] $\counter_{\valR}(\ecStates, \state) = \emptyset$ and $\badActs_{\valR}(\ecStates, \state) \neq \emptyset$.

			At state $\state$ it holds that $\valR(\state) = \exitVal[\valR](\ecStates, \state)$, since we chose it to be in $\bestExits[\valR]$. 
			Thus, it is in the set $\ecStates[\prime]$ constructed in Condition (i) of \cref{lem:exitPossibleUnderVr}.
			By Condition (ii) of \cref{lem:exitPossibleUnderVr}, we know that all states in $\ecStates \setminus \{\state\}$ attain a value that is either smaller or equal $\valR(\state)$. 
			Using this and the fact that no staying strategy can be optimal (as the set of trapping strategies is empty), we can derive a contradiction as follows.
			\begin{align*}
				\valR(\state) &= \adjustlimits\sup_{\strategy \in \strategies(\state)}\inf_{\strategy* \in \strategies*(\state)} \preop(\valR)(\state, \strategy, \strategy*) \tag{$\valR$ is fixpoint of $\preop$}\\
				&= \adjustlimits\sup_{\strategy \in \badActs_{\valR}(\ecStates, \state)}\inf_{\strategy* \in \strategies*(\state)} \preop(\valR)(\state, \strategy, \strategy*)\tag{$\badActs_{\valR}(\ecStates, \state)$ are optimal under $\valR$}\\
				&< \adjustlimits\sup_{\strategy \in \badActs_{\valR}(\ecStates, \state)}\inf_{\strategy* \in \stayStratsS(\badActs_{\valR}(\ecStates, \state), \ecStates,\state)} \preop(\valR)(\state, \strategy, \strategy*) \tag{$\counter_{\valR}(\ecStates, \state) =  \emptyset$}\\
				&= \adjustlimits\sup_{\strategy \in \badActs_{\valR}(\ecStates, \state)}\inf_{\strategy* \in \stayStratsS(\badActs_{\valR}(\ecStates, \state), \ecStates,\state)} \sum_{(\actionR, \actionS) \in \actions} \sum_{\state<\prime> \in \ecStates } \underbrace{\valR(\state<\prime>)}_{\leq \valR(\state)} \cdot \transitions(\state, \strategy, \strategy*) \cdot \strategy(\actionR) \cdot \strategy*(\actionS) \tag{Unfolding definition of $\preop$}\\
				&\leq \adjustlimits\sup_{\strategy \in \badActs_{\valR}(\ecStates, \state)}\inf_{\strategy* \in \stayStratsS(\badActs_{\valR}(\ecStates, \state), \ecStates,\state)} \sum_{(\actionR, \actionS) \in \actions} \sum_{\state<\prime> \in \ecStates }  \valR(\state) \cdot \transitions(\state, \strategy, \strategy*) \cdot \strategy(\actionR) \cdot \strategy*(\actionS) \tag{$\valR(\state<\prime>) \leq \valR(\state)$ by (ii) in \cref{lem:exitPossibleUnderVr}}\\
				&= \preop(\valR)(\state)\\
				&= \valR(\state). \tag{$\valR$ is fixpoint of $\preop$}
			\end{align*}
			Thus, overall $\valR(\state) < \valR(\state)$, a contradiction.
			
			\item[Case (III):] $\counter_{\valR}(\ecStates, \state) \neq \emptyset$ and $\badActs_{\valR}(\ecStates, \state) = \emptyset$.

			Since $\badActs_{\valR}(\ecStates, \state) = \emptyset$ we know by \cref{lem:negate-weak-dominance} that there exists $\strategy[L] \in \leaveStratsR(\ecStates, \state). \strategies(\state) \dominateseq_{\valR,\strategies*(\state)} \{\strategy[L]\}$. Our goal is to show that there exists $\strategy<\ast> \in \strategies(\state)$ that is optimal under $\valR$ and $\strategy<\ast> \in \deflStrats_{\upperBound}(\ecStates, \state)$. We distinguish two cases:
			\begin{itemize}
				\item (III.a): For all $\strategy[M] \in \leaveStratsR(\ecStates, \state)$ that are optimal under $\valR$, there exists $\strategy[H] \in \badActs_{\upperBound}(\ecStates, \state)$ such that $\support(\strategy[M]) \cap \support(\strategy[H]) \neq \emptyset$, and
				\item (III.b): There exists $\strategy[L] \in \leaveStratsR(\ecStates, \state)$ that is optimal under $\valR$ and it holds that $\support(\strategy[L]) \cap \bigcup_{\strategy \in \badActs_{\upperBound}(\ecStates, \state)} \support(\strategy) = \emptyset$.
			\end{itemize}
				
				\begin{description}
					\item[Case (III.a)] In this case, for all $\strategy[M] \in \leaveStratsR(\ecStates, \state)$ that are optimal under $\valR$, there exists $\strategy[H] \in \badActs_{\upperBound}(\ecStates, \state)$ such that $\support(\strategy[M]) \cap \support(\strategy[H]) \neq \emptyset$.
					
					Let $\strategy[H] \in \badActs_{\upperBound}(\ecStates, \state)$ and $\strategy[L] \in \{\strategy \in \leaveStratsR(\ecStates, \state) \mid \support(\strategy[H]) \cap \support(\strategy[L]) =\emptyset\}$ such that $\support(\strategy[M]) \cap \support(\strategy[H]) \neq \emptyset$ and $\support(\strategy[M]) \cap \support(\strategy[L]) \neq \emptyset$. In other words, only a strategy that mixes the supports of $\strategy[H]$ and $\strategy[L]$ is optimal under $\valR$. 
					
					Due to the case assumption that only strategies that mix with some hazardous strategy are optimal, $\strategy[L]$ must be sub-optimal under $\valR$, as well as $\strategy[H]$. 
					
					More formally, for all $\strategy[M] \in \leaveStratsR(\ecStates, \state)$, such that $\strategies(\state) \setminus \{\strategy[M]\} \dominateseq \{\strategy[M]\}$ it holds that
					\begin{align*}
						\exists\strategy[H] \in \badActs_{\upperBound}(\ecStates, \state).\exists\strategy[L] &\in \{\strategy \in \leaveStratsR(\ecStates, \state) \mid \support(\strategy[H]) \cap \support(\strategy[L]) =\emptyset\}:\\
						\sup_{\distributions(\support(\strategy[M]) \setminus \support(\strategy[L]) )} \inf_{\strategy* \in \strategies*(\state)} &\preop(\valR)(\state, \strategy, \strategy*) \dominates \{\strategy[M]\}\\
						&\text{and}\\
						\sup_{\distributions(\support(\strategy[M]) \setminus \support(\strategy[H]) )} \inf_{\strategy* \in \strategies*(\state)} &\preop(\valR)(\state, \strategy, \strategy*) \dominates \{\strategy[M]\}.
					\end{align*}
					
					Let $\strategies[\safe]<\ast> \subseteq \strategies*(\state)$ be the optimal Player $\safe$ strategies with respect to $\strategy[M]$. Since $\strategy[M]$ has to mix $\strategy[L]$ and $\strategy[H]$, the following must be true:
					\begin{itemize}
						\item there exists $\sigma_1 \in \strategies[\safe]<\ast>$ such that 
						\begin{align}\label{phsi}
							\preop(\valR)(\state, \strategy[L],\sigma_1) &< \preop(\valR)(\state, \strategy[H], \sigma_1), \text{ and}
						\end{align}
						\item there exists $\sigma_2 \in \strategies[\safe]<\ast>$ such that 
						\begin{align}\label{phsii}
							\preop(\valR)(\state, \strategy[L],\sigma_2) &> \preop(\valR)(\state, \strategy[H], \sigma_2),
						\end{align}
					\end{itemize}
					because otherwise either $\strategy[L]$ or $\strategy[H]$ would be optimal under $\valR$ thus mixing would not be necessary. Further, Player $\safe$ also needs to mix the two strategies $\sigma_1$ and $\sigma_2$ to ensure optimality. 
					
					Our goal now is to show that $\state \notin \bestExits[\valR](\ecStates)$ which would lead to a contradiction to the assumption that $\state \in \bestExits[\valR](\ecStates)$, showing that at a best exit there cannot exist only optimal strategies that fulfill the same properties as $\strategy[M]$. 
					
					To estimate the best exit value from $\ecStates$ under $\valR$, all exists from $\state \in \ecStates$ have to be estimated.  By \cref{phsi} and \cref{phsii}, we know that the part of the \EC $\ecStates$ that is reachable with $(\strategy[H], \sigma_1)$, say $\ecStates<\prime> \subset \ecStates$, attains a higher value than the one that is reachable with $(\strategy[H], \sigma_2)$. 
					
					More formally, the following holds. Let $\hat{\state} \in \max_{\state<\prime> \in \ecStates<\prime>} \exitVal[\valR](\ecStates, \state<\prime>)$. At $\hat{\state}$ leaving $\ecStates<\prime>$ has to be possible, because otherwise, the values at all $\state<\prime> \in \ecStates<\prime>$ would be equal to $0$, which would be a contradiction to the assumption that $\ecStates \subseteq \states \setminus (\success \cup \winning)$. Then, the following chain of equations holds.
					\begin{align*}
						\exitVal[\valR](\ecStates, \state) &= \adjustlimits\sup_{\strategy \in \strategies(\state)}\inf_{\strategy* \in \strategies*(\state)}\preop(\valR)(\state, \strategy, \strategy*) \tag{$\badActs_{\valR}(\ecStates, \state) = \emptyset$}\\
						&=\inf_{\strategy* \in \strategies*(\state)}\preop(\valR)(\state, \strategy[M], \strategy*)\tag{$\strategy[M]$ is optimal under $\valR$}\\
						&=\preop(\valR)(\state, \strategy[M], \sigma_1)\tag{$\sigma_1$ is optimal under $\valR$ with respect to $\strategy[M]$}\\
						&<\preop(\valR)(\state, \strategy[H], \sigma_1).\tag{Under $\sigma_1$, $\strategy[H]$ is optimal by \cref{phsi}}\\
						&=\sum_{(\actionR, \actionS) \in \actions}\sum_{\state<\prime> \in \destination(\state, \strategy[H], \sigma_1)}\valR(\state<\prime>) \cdot \transitions(\state, \actionR, \actionS) \cdot \strategy(\actionR) \cdot \strategy*(\actionS) \tag{Unfolding the def. of $\preop$}\\
						&\leq \sum_{(\actionR, \actionS) \in \actions}\sum_{\state<\prime> \in \destination(\state, \strategy[H], \sigma_1)}\exitVal[\valR](\ecStates, \hat{\state}) \cdot \transitions(\state, \actionR, \actionS) \cdot \strategy(\actionR) \cdot \strategy*(\actionS) \tag{At $\hat{\state}$ one can attain the highest value upon leaving}\\
						&=  \exitVal[\valR](\ecStates, \hat{\state}) \tag{Everything sums-up to 1}.
					\end{align*}
					Thus, under $\valR$, the state $\hat{\state}$ attains a higher value upon leaving than $\state$. This is a contradiction to the assumption that $\state \in \bestExits[\valR](\ecStates)$.
					
					\item[Case (III.b)] There exists $\strategy[L] \in \leaveStratsR(\ecStates, \state)$ that is optimal under $\valR$ and it holds that $\support(\strategy[L]) \cap \bigcup_{\strategy \in \badActs_{\upperBound}(\ecStates, \state)} \support(\strategy) = \emptyset$.
					
					Then, $\strategy[L]$ must belong to the set of deflating strategies under $\upperBound$, as it satisfies both conditions posed by the definition of deflating strategies (see  \cref{def:deflStrats}).
				
				If $\counter_{\upperBound}(\ecStates, \state) \neq \emptyset$, then the following chain of equations holds.
				\begin{align*}
					\exitVal[\valR](\ecStates, \state) &= \adjustlimits \sup_{\strategy \in \strategies(\state)}\inf_{\strategy* \in \strategy*(\state)} \preop(\valR)(\state, \strategy, \strategy*) \tag{By \cref{def:badAction} in case $\badActs_{\valR}(\ecStates, \state) = \emptyset$}\\
					&= \inf_{\strategy* \in \strategy*(\state)} \preop(\valR)(\state, \strategy[L], \strategy*) \tag{$\strategy[L]$ is optimal}\\
					&\leq \adjustlimits \sup_{\strategy \in \deflStrats_{\upperBound}(\ecStates, \state)}\inf_{\strategy* \in \strategy*(\state)} \preop(\valR)(\state, \strategy, \strategy*) \tag{$\{\strategy[L]\} \subseteq \deflStrats_{\upperBound}(\ecStates, \state)\}$}\\
					&\leq \adjustlimits \sup_{\strategy \in \deflStrats_{\upperBound}(\ecStates, \state)}\inf_{\strategy* \in \counter_{\upperBound}(\ecStates, \state)} \preop(\valR)(\state, \strategy, \strategy*) \tag{$\counter_{\upperBound}(\ecStates, \state) \subseteq \strategies*(\state)$}\\
					&\leq \adjustlimits \sup_{\strategy \in \deflStrats_{\upperBound}(\ecStates, \state)}\inf_{\strategy* \in \counter_{\upperBound}(\ecStates, \state)} \preop(\upperBound)(\state, \strategy, \strategy*) \tag{$\preop$ is order-preserving}\\
					&= \exitVal[\upperBound](\ecStates, \state).
				\end{align*}
				
				If $\counter_{\upperBound}(\ecStates, \state) = \emptyset$, then the following chain of equations holds.
				\begin{align*}
					\exitVal[\valR](\ecStates, \state) &= \adjustlimits \sup_{\strategy \in \strategies(\state)}\inf_{\strategy* \in \strategy*(\state)} \preop(\valR)(\state, \strategy, \strategy*) \tag{By \cref{def:badAction} in case $\badActs_{\valR}(\ecStates, \state) = \emptyset$}\\
					&= \inf_{\strategy* \in \strategy*(\state)} \preop(\valR)(\state, \strategy[L], \strategy*) \tag{$\strategy[L]$ is optimal}\\
					&\leq \adjustlimits \sup_{\strategy \in \deflStrats_{\upperBound}(\ecStates, \state)}\inf_{\strategy* \in \strategies*( \state)} \preop(\upperBound)(\state, \strategy, \strategy*) \tag{$\preop$ is order-preserving}\\
					&= \exitVal[\upperBound](\ecStates, \state).
				\end{align*}

			\end{description}%

			\item[Case (IV):] $\counter_{\valR}(\ecStates, \state) \neq \emptyset$ and $\badActs_{\valR}(\ecStates, \state) \neq \emptyset$.
			
			It holds that $\state \in \bestExits[\valR](\ecStates)$, thus there exist no other state at which leaving $\ecStates$ attains better value. By \cref{lem:exitPossibleUnderVr} we have that for all $\state<\prime> \in \bestExits[\valR](\ecStates)$ it holds that $\valR(\state<\prime>) \leq \exitVal[\valR](\ecStates, \state<\prime>)$. Consequently, at a best exit it can only hold that $\valR(\state<\prime>) = \exitVal[\valR](\ecStates, \state<\prime>)$ because otherwise $\valR$ would not be a fixpoint of $\preop$. Therefore, the following chain of equations holds at $\state$.
			\begin{align*}
				\valR(\state) &=  \adjustlimits \sup_{\strategy \in \strategies(\state)} \inf_{\strategy* \in \strategies*(\state)}\preop(\valR)(\state, \strategy, \strategy*)\tag{$\valR$ is fixpoint of $\preop$}\\
				&= \adjustlimits \sup_{\strategy \in \badActs_{\valR}(\state)} \inf_{\strategy* \in \strategies*(\state)}\preop(\valR)(\state, \strategy, \strategy*)\tag{$\badActs_{\valR}(\ecStates, \state)$ are optimal under $\valR$}\\
				&= \adjustlimits \sup_{\strategy \in \badActs_{\valR}(\state)} \inf_{\strategy* \in  \counter_{\valR}(\state)}\preop(\valR)(\state, \strategy, \strategy*)\tag{$\counter_{\valR}(\ecStates, \state)$ are optimal under $\valR$}\\
				&= \exitVal[\valR](\ecStates, \state). \tag{By \cref{lem:exitPossibleUnderVr} we have $\valR(\state) \leq \exitVal[\valR](\ecStates, \state)$ however only "=" can hold as $\valR$ is fixpoint of $\preop$}
			\end{align*}
			Thus, this case is equal to Case (III) where $\exitVal[\valR](\ecStates, \state) =  \sup_{\strategy \in \strategies(\state)} \inf_{\strategy* \in \strategies*(\state)}\preop(\valR)(\state, \strategy, \strategy*)$, as the assumptions of cases (I) and (II) are not possible.
		\end{description}
		\end{proof}

		\begin{lemma}[Existence of maximal \BECs]\label{lem:existenceMaxBECs}
			Given a \CSG $\game$ and let $\endComponent \subseteq (\states \setminus (\success \cup \winning))$ be an \EC and let $\upperBound$ be a valuation. If under $\upperBound$ there exists a \BEC $ \ecStates \subseteq \endComponent$, then there also exists a maximal \BEC, i.e. there exists a \BEC $\ecStates<\max> \subseteq \endComponent$ such that $\ecStates \subseteq \ecStates<\text{max}>$ and it holds that $\ecStates<\text{max}> \cup \{\state\}$ is not a \BEC for all $\state \in \endComponent \setminus \ecStates<\max>$.
		\end{lemma}
		
		\begin{proof}
			To prove the lemma it suffices to show the following: Let $\ecStates[1] \subseteq (\states \setminus (\success \cup \winning))$ and $\ecStates[2]\subseteq (\states \setminus (\success \cup \winning))$ be two \ECs that are \BEC under $\upperBound$ with $\ecStates[1] \cap \ecStates[2] \neq \emptyset$. Then, it holds that $\ecStates[1] \cup \ecStates[2]$ is also a \BEC.
			
			Let $\state<\cap> \in (\ecStates[1] \cap \ecStates[2])$, $\state<\prime> \in \ecStates[1]$, and $\state<\prime\prime> \in \ecStates[2]$. Since $\ecStates[1]$ is a \BEC at each state $\state \in \ecStates[1]$ it holds that $\badActs_{\upperBound}(\ecStates[1], \state) \neq \emptyset$ and $\counter_{\upperBound}(\ecStates[1], \state)\neq \emptyset$ for all $\state \in \ecStates[1]$. Further, since $\ecStates[1]$ is an \EC, there exists a play $\state[0]\state[1]\dots$ such that $\state[0]=\state<\prime>$ and $\state[n]=\state<\cap>$ for some $n$ and for all $0 \leq i < n$ it holds that $\state[i+1] \in \destination(\state[i], \strategy[i]<\ast>, \sigma_i^\ast)$ where $\strategy[i]<\ast> \in \badActs_{\upperBound}(\ecStates[1], \state[i])$ and $\sigma_i^\ast \in \counter_{\upperBound}(\ecStates[1],\state[i])$.
			
			Similarly, since $\ecStates[2]$ is also a \BEC, there exists a play $\state[0]<\prime>\state[1]<\prime>\dots$ such that $\state[0]<\prime> = \state<\cap>$ and $\state[m]<\prime> = \state<\prime\prime>$ for some $m$ and for all $0 \leq j < m$ it holds that $\state[j+1]<\prime> \in \destination(\state[j]<\prime>, \strategy[j]<\prime>, \sigma_j^\ast)$ where $\strategy[j]<\prime> \in \badActs_{\upperBound}(\ecStates[2], \state[j]<\prime>)$ and $\sigma_j^\prime \in \counter_{\upperBound}(\ecStates[2],\state[j]<\prime>)$.
			
			Therefore, for any $\state, \state<\prime> \in \ecStates[1] \cup \ecStates[2]$ there exist a play $\state[0]<\prime\prime>\state[1]<\prime\prime>\dots$ such that $\state[0]<\prime\prime> = \state<\prime>$ and $\state[l]<\prime\prime> = \state<\prime\prime>$ for some $l$ and for all $0 \leq k < l$ it holds that $\state[k+1]<\prime\prime> \in \destination(\state[k]<\prime\prime>, \strategy[k]<\prime\prime>, \sigma_k^{\prime\prime})$ where $\strategy[k]<\prime\prime> \in \badActs_{\upperBound}(\ecStates[1] \cup \ecStates[2], \state[k]<\prime\prime>)$ and $\sigma_k^{\prime\prime} \in \counter_{\upperBound}(\ecStates[1] \cup \ecStates[2],\state[k]<\prime\prime>)$. 
			
			Thus, we have shown that the $\ecStates[1] \cup \ecStates[2]$ is an \EC. Further, since at all states of $\ecStates[1]$ and $\ecStates[2]$ the two conditions (i) and (ii) of \Cref{defBec} are fulfilled, also $\ecStates[1] \cup  \ecStates[2]$ fulfills them. Consequently, $\ecStates[1] \cup  \ecStates[2]$ is a \BEC which in turn proves that there exist maximal \BECs.
		\end{proof}
		
		Now that we have proven that maximal \BECs indeed exist, the next step is to prove the correctness of \procname{FIND\_MBECs}, i.e. the algorithm that can find maximal \BECs. 
		
		\findingMBECScorrect*

		\begin{proof}
			We prove the lemma in two steps. First, if for a non-empty set of states $\ecStates \subseteq \states$ it holds that $\ecStates \in \procname{FIND\_MBECs}(\game, \endComponent, \upperBound)$ then it holds that $\ecStates \subseteq \endComponent$ and $\ecStates$ is a maximal \BEC in $\endComponent$. Second, we will show that if a set of states $\ecStates\subseteq \endComponent$ is a maximal \BEC in $\endComponent$, then $\ecStates \in \procname{FIND\_MBECs}(\game, \endComponent, \upperBound)$.
			
			\begin{description}
				\item[1. Direction ``$\Rightarrow$"] Let $\ecStates \subseteq \endComponent$ be a non-empty set of states and let $\ecStates \in \procname{FIND\_MBEC}(\game, \endComponent, \upperBound)$. We prove that $\ecStates$ is a maximal \BEC in $\endComponent$ via structural induction.
				
				\begin{description}
					\item[Basis step:] Let $B \coloneqq \{\state \in \endComponent \mid \badActs_{\upperBound}(\endComponent, \state) \neq \emptyset\}$. The recursive function $\procname{FIND\_MBECs}$ has two base cases that we consider now.
									
					If $B = \endComponent$ holds, then $\procname{FIND\_MBECs}(\game, \endComponent, \upperBound) =\{B\}$ and thus, by the case assumption, $B= \ecStates$ must be true. Since $\endComponent$ is a \MEC, also $B$ and thus $\ecStates$ have to be \MECs. Further, as for all states $\state \in B$ there exist at least one hazardous strategy, then $B$, and thus also $\ecStates$, must be maximal \BECs. $\ecStates \subseteq \endComponent$ follows directly from the case assumption that $\endComponent = B$.
					
					Otherwise it must hold that $B = \emptyset$ which in turn would mean that $\procname{FIND\_MBECs}(\game, \endComponent, \upperBound) = \{\emptyset\}$. Since by the case assumption we have that $\ecStates \in \procname{FIND\_MBECs}(\game, \endComponent, \upperBound)$ we can conclude that $\ecStates = \emptyset $ which is a contradiction to the assumption that $\ecStates$ is a non-empty set of states. 
					
					\item[Recursive step:] It holds that $B$ is a non-empty set and $B \neq \endComponent$. In case  there exists no \MEC in $B$, then it holds that $\procname{FIND\_MBECs}(\game, \endComponent, \upperBound) = \{\emptyset\}$ which in turn means that $\ecStates = \emptyset$ which is a contradiction to the assumption that $\ecStates$ is a non-empty set of states.
									
					Otherwise, there exists $\endComponent<\prime> \in \procname{FIND\_MECs}(\game, B)$ such that $\ecStates \in \procname{FIND\_MBECs}(\game, \endComponent<\prime>, \upperBound)$. In case $B^\prime \coloneqq \{ \state \in \endComponent<\prime> \mid \badActs_{\upperBound}(\endComponent<\prime>, \state) \neq \emptyset\} = \endComponent<\prime>$ or $B^\prime = \emptyset$ we end up in the basis step. Otherwise, we end up in the recursive step, that will eventually lead to a basis step, as there is only a finite number of states.
				\end{description}
								
				\item[2. Direction ``$\Leftarrow$"] Let $\emptyset \neq \ecStates \subseteq \endComponent$ be a maximal \BEC in the \MEC $\endComponent$. We need to show that $\ecStates \in \procname{FIND\_MBEC}(\game, \endComponent, \upperBound)$ holds. 
				
				Let $B \coloneqq \{\state \in \endComponent \mid \badActs_{\upperBound}(\endComponent, \state) \neq \emptyset\}$. Assume towards a contradiction that $\ecStates  \notin \procname{FIND\_MBEC}(\game, \endComponent, \upperBound)$. As $\ecStates$ is a maximal \BEC in $\endComponent$, it must be true that $B \coloneqq \{\state \in \endComponent \mid \badActs_{\upperBound}(\endComponent, \state) \neq \emptyset\} \neq \emptyset$ and $\ecStates \subseteq B$. Since by assumption  $\ecStates  \notin \procname{FIND\_MBEC}(\game, \endComponent, \upperBound)$ holds, for all $\endComponent<\prime> \in \procname{FIND\_MECs}(\game, B)$ it must be true that $\ecStates \notin \procname{FIND\_MBECs}(\game, \endComponent<\prime>, \upperBound)$. Consequently, $\ecStates$ is not a \MEC in $B$ which is a contradiction to the assumption that $\ecStates$ is a maximal \BEC.
			\end{description}
		\end{proof}

\subsection{Soundness and Completeness}\label{apx:soundness_and_completeness}
In order to prove the soundness and completeness of \cref{alg:bvi} with \cref{alg:deflate_mecs} as \procname{DEFLATE} routine, we need to prove that the sequence of lower and upper bounds converges to a fixpoint. Therefore, before we can prove \cref{theo:soundness_completeness}, first we need to prove the following lemma. 

\begin{lemma}[\cref{alg:bvi} converges to a fixpoint]\label{lemma:convergence_to_fixpoint}
	The BVI algorithm (\cref{alg:bvi}) converges to a fixpoint, i.e. $\lim_{k\rightarrow \infty}(\preop<k>(\lowerBound<0>), \deflBell^k\left( \upperBound<0>\right)) = (\preop(\lim_{k\rightarrow \infty}\preop<k>(\lowerBound<0>)),  \deflBell(\lim_{k\rightarrow \infty}\deflBell^k\left( \upperBound<0>\right))).$ 
\end{lemma}

\begin{proof}
 We consider the domain $\mathbb{V} := [0,1]^{|\states|} \times [0,1]^{|\states|}$, i.e. every element consists of two vectors of real numbers, the under- and over-approximation. The bottom element of the domain, denoted by $\bot$, is $(\vec{0},\vec{1})$, where for $a \in [0,1]$, $\vec{a}$ denotes the function that assigns $a$ to all states. We further restrict the domain to exclude elements of the domain that are trivially irrelevant for the computation. In particular, we exclude all tuples $(\lowerBound, \upperBound)$ where $\lowerBound(\state) < 1$ for a target state $\state \in \success$ or $\upperBound(\state) > 0$ for a state with no path to the target state $\state \in \winning$. Then the bottom  element is $\bot = (\lowerBound<0>, \upperBound<0>)$, i.e. the vector that we have before the first iteration of the main loop of \cref{alg:bvi}. Concretely, $\lowerBound<0>(\state)$ is 1 for all $\state \in \success$, i.e. target states, and 0 everywhere else, and $\upperBound<0>(\state)$ is 0 for all $\state \in \winning$, i.e. states where $\safe$ can surely win, and 1 everywhere else.

We define a comparator $\sqsubseteq$ on $\mathbb{V}$, to compare two elements of the domain. We write $(\lowerBound<k>, \upperBound<k>) \sqsubseteq (\lowerBound<k+1>, \upperBound<k+1>)$ if and only if both $\lowerBound<k> \leq \lowerBound<k+1>$ and $\upperBound<k> \geq \upperBound<k+1>$ hold with component-wise comparison. Intuitively, $(\lowerBound<k>, \upperBound<k>) \sqsubseteq (\lowerBound<k+1>, \upperBound<k+1>)$ holds if $(\lowerBound<k+1>, \upperBound<k+1>)$ is a more precise approximation than $(\lowerBound<k>, \upperBound<k>)$. The comparator $\sqsubseteq$ induces a complete partial order over the domain, since we have a bottom element and every direct subset has a supremum; the latter claim holds, because $\sqsubseteq$ reduces to component-wise comparison between real numbers from [0,1], where suprema exist. For more details on the definition of directed set and complete partial orders, we refer to \cite{daveyIntroductionLatticesOrder2002}.

\cref{alg:bvi} first applies the Bellman operator on the over- and under-approximation and subsequently applies the deflate operator on the over-approximation (i.e. upper bound). Thus, the operator that mimics the behavior of the algorithm is the following $\mathsf{BVI}(\lowerBound<k>, \upperBound<k>) = (\preop(\lowerBound<k>), \deflBell(\upperBound<k>))$.

From \cref{theo:convergentUnderApprox} we know that the under-approximation converges. Also, by \cref{DBorderPreserving} we know that $\deflBell$ is order-preserving. Thus, for the final argument it remains to show that $\deflop$ is \newlink{def:scottContinuity}{(Scott-)continuous}. 

A map is \newlink{def:scottContinuity}{(Scott-)continuous} if, for every \newlink{def:directedSet}{directed set} $D$ in $(\proj{2}{\mathbb{V}}, \sqsubseteq)$, the subset $\deflop(D)$ of $(\proj{2}{\mathbb{V}}, \sqsubseteq)$ is directed, and $\deflop(\sup D)=\sup \deflop(D)$. Let $\state \in \states$, if $\state$, under $\sup D$, does not belong to any \BEC, then $\deflop(\sup D)(\state) = \sup D (\state)$. Thus, we proceed with the assumption that under the valuation $\sup D$, $\state \in \ecStates \subseteq \states$ such that $\ecStates$ is a \BEC. 

Let $\badActs_{\sup D}(\ecStates, \state)$ be the set of hazardous strategies and let $\counter_{\sup D}(\ecStates, \state)$ be the set of suitable counter strategies for player $\safe$ (see \Cref{defBec}). Since $\ecStates$ is a \BEC both sets are non-empty, i.e.  $\badActs_{\sup D}(\ecStates, \state) \neq \emptyset$ and $\counter_{\sup D}(\ecStates, \state) \neq \emptyset$, at all $\state \in \ecStates$. Thus, $\deflop(\sup D)(\state) = \min(\sup D(\state),\bestExitVal[\sup D](\ecStates))$. For the sake of readability let $\strategies<\prime>(\state) \coloneqq \distributions(\actionAssignment[\reach](\state) \setminus \bigcup_{\strategy<\prime\prime> \in \badActs_{\sup D}(\ecStates, \state)} \support(\strategy<\prime>))$ and $\strategies[\safe]<\prime>(\state) \coloneqq \counter_{\sup D}(\ecStates, \state)$ for some $\state \in \ecStates$. Then, the following chain of equations holds for $\state \in \ecStates$.
\begingroup
\allowdisplaybreaks
\begin{align*}
&\deflop(\sup D)(\state) = \min(\sup D(\state),\bestExitVal[\sup D](\ecStates)) \\
&= \min(\sup D(\state),\max_{\state<\prime> \in \ecStates} \exitVal[\sup D]( \ecStates, \state<\prime>)) \tag*{(By \Cref{def:bestExitVal} (\newlink{defBestExit}{best exit value}))}\\ 
&= \min(\sup D(\state),\max_{\state<\prime> \in \ecStates} \adjustlimits\sup_{\strategy \in \strategies<\prime>(\state<\prime>)} \inf_{\strategy* \in \strategies[\safe]<\prime>(\state<\prime>)} \preop(\sup_{d \in D} d)(\state<\prime>, \strategy, \strategy*))\tag*{(By \Cref{def:exitVal} (\newlink{def:exitVal}{exit value}))}\\
&= \min(\sup D(\state),\max_{\state<\prime> \in \ecStates} \sup_{d \in D} \adjustlimits\sup_{\strategy \in \strategies<\prime>(\state<\prime>)} \inf_{\strategy* \in \strategies[\safe]<\prime>(\state<\prime>)}  \preop(d)(\state<\prime>, \strategy, \strategy*))  \tag*{(Bellman operator is Scott-continuous (see proof Theorem \ref{theo:BVInoEC}) )}\\
&= \min(\sup D(\state),\sup_{d \in D} \max_{\state<\prime> \in \ecStates} \adjustlimits\sup_{\strategy \in \strategies<\prime>(\state<\prime>)} \inf_{\strategy* \in \strategies[\safe]<\prime>(\state<\prime>)} \preop(d)(\state<\prime>, \strategy, \strategy*)) \tag*{(By \Cref{arg1:lem53})}\\
&= \min(\sup D(\state),\sup_{d \in D} \max_{\state<\prime> \in \ecStates} \exitVal[d](\ecStates, \state<\prime>)) \tag{By \Cref{def:exitVal} (\newlink{def:exitVal}{exit value}))}\\
&= \min(\sup D(\state),\sup_{d \in D} \bestExitVal[d](\ecStates)) \tag{By \Cref{def:bestExitVal} (\newlink{defBestExit}{best exit value})}\\
&= \sup_{d \in D} \min(d(\state),\bestExitVal[d](\ecStates)) \tag{$\min$ is Scott-continuous}\\
&= \sup_{d \in D} \mathfrak{D}(d)(\state). \tag{Since $\deflop$ is the deflate operator}
\end{align*}
\endgroup

\begin{argument}\label{arg1:lem53}
	 Since $\max_{\state<\prime> \in \ecStates} \exitVal[\sup D]( \ecStates, \state<\prime>)$ exists the following chain of equations holds:
	 \begingroup
	 \allowdisplaybreaks
	\begin{align*}
		&\adjustlimits\sup_{d \in D} \max_{\state<\prime> \in \ecStates} \adjustlimits\sup_{\strategy \in \strategies<\prime>(\state<\prime>)} \inf_{\strategy* \in \strategies[\safe]<\prime>(\state<\prime>)} \preop(d)(\state<\prime>, \strategy, \strategy*)\\
		&=\adjustlimits\sup_{d \in D} \sup_{\state<\prime> \in \ecStates} \adjustlimits\sup_{\strategy \in \strategies<\prime>(\state<\prime>)} \inf_{\strategy* \in \strategies[\safe]<\prime>(\state<\prime>)} \preop(d)(\state<\prime>, \strategy, \strategy*) \tag{As the maximum exists}\\
		&= \sup_{(d, \state<\prime>) \in (D \times \ecStates)} \adjustlimits\sup_{\strategy \in \strategies<\prime>(\state<\prime>)} \inf_{\strategy* \in \strategies[\safe]<\prime>(\state<\prime>)} \preop(d)(\state<\prime>, \strategy, \strategy*)\\
		&=\sup_{\state<\prime> \in \ecStates} \sup_{d \in D}  \adjustlimits\sup_{\strategy \in \strategies<\prime>(\state<\prime>)} \inf_{\strategy* \in \strategies[\safe]<\prime>(\state<\prime>)} \preop(d)(\state<\prime>, \strategy, \strategy*)\\
		&=\max_{\state<\prime> \in \ecStates} \sup_{d \in D}  \adjustlimits\sup_{\strategy \in \strategies<\prime>(\state<\prime>)} \inf_{\strategy* \in \strategies[\safe]<\prime>(\state<\prime>)} \preop(d)(\state<\prime>, \strategy, \strategy*).
	\end{align*}
	\endgroup
\end{argument}
We have shown that $\preop$ and $\deflop$ are both Scott-continuous, thus the composition $\deflBell$, i.e. the sequential application of the operators $\preop$ and $\deflop$, is also continuous \cite{scottOutlineMathematicalTheory1977}. Further, $\mathbb{V}$ is a complete partial order. Therefore, Kleen's Fixpoint Theorem \cite[Theorem 8.15]{daveyIntroductionLatticesOrder2002} is applicable.

Then, we know that 
\begin{align*}
\lim_{k \rightarrow \infty} \deflBell^k(\bot) = \sup_{k \geq 0}\deflBell^k(\bot) \tag*{($\star$)}
\end{align*}

holds and by the theorem we know that the fixpoint exists and is given by $\sup_{k \geq 0}\deflBell^k(\bot)$. Now, we can finally conclude:
\begin{align*}
\deflBell(\lim_{k \rightarrow \infty} \deflBell^k(\bot)) &\stackrel{(\star)}{=}  \deflBell(\sup_{k \geq 0} \deflBell^k(\bot))\\
&= \sup_{k \geq 0} \deflBell(\deflBell^k(\bot)) \tag*{(since $\deflBell$ is continuous)}\\
&= \sup_{k \geq 1} \deflBell^k(\bot)\\
&= \sup_{k \geq 0} \deflBell^k(\bot) \tag*{(since $\bot \sqsubseteq \deflBell^k(\bot)$ for all $k$)}\\
&= \lim_{k \rightarrow \infty} \deflBell^k(\bot). \tag*{by ($\star$)}
\end{align*}
\end{proof}

Now we can prove the following theorem.
\soundnessCompleteness*

\begin{proof}
	We denote by $\lowerBound<k>$ and $\upperBound<k>$ the lower and upper bound function after the $k$-th call of \procname{DEFLATE}. $\lowerBound<k>$ and $\upperBound<k>$ are monotonic under-, respectively, over-approximation of $\valR$ because they are updated via Bellman updates respectively $\deflBell$-updates, which are order-preserving as shown in \cref{DBorderPreserving} and soundness as shown in \cref{lemValidDB} .
	
	Since \procname{DEFLATE} iterates over finite sets, the computations take a finite time. Thus, it remains to prove that the main loop of \cref{alg:bvi} terminates, i.e., for all $\varepsilon>0$, there exists an $n \in \NN$ such that for all $\state \in \states$ $\upperBound<n>(\state) - \lowerBound<n>(\state) \leq \varepsilon$. It suffices to show that $\lim_{k\rightarrow \infty} \upperBound<k> - \valR = 0$, because $\lim_{k\rightarrow \infty} \lowerBound<k> = \valR$ (from e.g. \cite{raghavanAlgorithmsStochasticGames1991}).
	
	In the following let $\upperBound<\star>: = \lim_{k \rightarrow \infty} \upperBound<k>$, that exists by Lemma \ref{lem:Ustar-is-fixpoint}, and $\Delta(\state) := \mathsf{U}^\star(\state) - \valR(\state)$. Assume towards a contradiction that the algorithm does not converge, i.e., there exists a state $\state \in \states$ with $\Delta(\state) >0$.
	
	The proof is structured as follows.
	\begin{itemize}
		\item From $\Delta > 0$ we derive that there has to exist a \BEC.
		\item The states with the maximal $\Delta$ contain \BECs.
		\item \cref{alg:bvi} will find  a \BEC contained in $\ecStates \subseteq \states$ and deflate it.
		\item Deflating will decrease the upper bound of that states contained in the \BEC, which is a contradiction because by Lemma \ref{lemma:convergence_to_fixpoint}, $\upperBound<\star>$ is a fixpoint.
	\end{itemize}
	
	Let $\Delta^{\max} := \max_{\state \in \states} \Delta(\state)$ and let $\endComponent:= \{ \state \in \states \;|\; \Delta(\state) = \Delta^{\max}\}$. If $\endComponent$ does not contain any \BECs, then the contraposition of \cref{theo:BVInoEC} proves our goal. Thus, we continue with the assumption that $\endComponent$ contains \BECs. 
	
	Let $\ecStates \subseteq \endComponent$ be a \BEC contained in $\endComponent$ that \cref{alg:deflate_mecs} will eventually find and deflate. We now consider \emph{bottom} \BECs. A \BEC $\ecStates<\prime>$ is called bottom in $\ecStates$ if none of the successors of a strategy that leaves the \BEC $\ecStates<\prime>$, is part of another \BEC in $\ecStates$. A bottom \BEC can be computed by first finding the bottom \MEC within $\ecStates$, then identifying all \BECs, ordering them topologically and finally picking the one at the end of a chain.
	
	Let $\ecStates<\prime> \subseteq \ecStates$ be a bottom \BEC that Algorithm \ref{alg:deflate_mecs} eventually finds. In order to deflate the \BEC, we need to estimate the exit value for each $\state \in \ecStates<\prime>$, as defined in Definition \ref{def:exitVal}. Further, $\badActs_{\upperBound<\star>}(\ecStates<\prime>, \state)$ is the set of hazardous strategies and since $\ecStates<\prime>$ is a \BEC, at all states $\state \in \ecStates<\prime>$, we have $\badActs_{\upperBound<\star>}(\ecStates<\prime>, \state) \neq \emptyset$. 
	
	We distinguish two cases: (i) $\counter_{\upperBound<\star>}(\ecStates<\prime>, \state) \neq \emptyset$; and (ii) $\counter_{\upperBound<\star>}(\ecStates<\prime>, \state) = \emptyset$.
	
	\begin{description}
		\item[Case (i)] In this case we have that $\counter_{\upperBound<\star>}(\ecStates<\prime>, \state) \neq \emptyset$. Then $\deflStrats_{\upperBound<\star>}(\ecStates<\prime>, \state) \neq \emptyset$ (and because $\ecStates \cap \winning = \emptyset$), thus the following chain of equations holds.
		\begin{align*}
			\upperBound<\star>(\state) &= \adjustlimits\sup_{\strategy \in \strategies(\state)}\inf_{\strategy* \in \strategies*(\state)}\preop(\upperBound<\star>)(\state, \strategy, \strategy*) \tag{$\upperBound<\star>$ is a fixpoint}\\
			&=\adjustlimits\sup_{\strategy\in\badActs_{\upperBound<\star>}(\ecStates<\prime>, \state)}\inf_{\strategy* \in \counter_{\upperBound<\star>}(\ecStates<\prime>, \state)}\preop(\upperBound<\star>)(\state, \strategy, \strategy*) \tag{$\badActs_{\upperBound<\star>}(\ecStates<\prime>, \state) \neq \emptyset$}\\
			&>\adjustlimits\sup_{\strategy \in \deflStrats_{\upperBound<\star>}(\ecStates<\prime>, \state)}\inf_{\sigma \in \counter_{\upperBound<\star>}(\ecStates<\prime>, \state)} \preop(\upperBound<\star>)(\state, \strategy, \strategy*) \tag{$\deflStrats_{\upperBound<\star>}(\ecStates<\prime>, \state)$ are sub-optimal with respect to $\upperBound<\star>$}\\
			&=	\exitVal[\upperBound<\star>](\ecStates<\prime>,\state)\tag{By Def. \ref{def:exitVal}}
		\end{align*}
		Since $\bestExitVal[\upperBound<\star>](\ecStates<\prime>) = \max_{\state<\prime> \in \ecStates<\prime>}\exitVal[\upperBound<\star>](\ecStates<\prime>,\state<\prime>)$, and above we have shown that for all $\state \in \ecStates<\prime>$ it holds $\upperBound<\star>(\state) > \exitVal[\upperBound<\star>](\ecStates<\prime>,\state)$, we obtain a contradiction to the assumption that $\upperBound<\star>$ is a fixpoint.
		
		\item[Case (ii)] In this case it holds that $\counter_{\upperBound<\star>}(\ecStates<\prime>, \state)= \emptyset$. Thus, Player $\safe$ prefers strategies that are leaving with respect to $\badActs_{\upperBound<\star>}(\ecStates, \state)$. Since $\ecStates<\prime>$ is a bottom \BEC at least one successor state does not belong to $\ecStates$. Therefore, the following chain of equations holds.
		\begin{align*}
			\Delta^{\max} - \valR(\state) &= \upperBound<\star>(\state) \tag{By definition of $\Delta^{\max}$}\\
			&= \adjustlimits\sup_{\strategy \in \strategies(\state)}\inf_{\strategy* \in \strategies*(\state)}\preop(\upperBound<\star>)(\state, \strategy, \strategy*) \tag{$\upperBound<\star>$ is a fixpoint}\\
			&= \adjustlimits\sup_{\strategy \in \strategies(\state)}\inf_{\strategy* \in \strategies*(\state)}\sum_{(\actionR,\actionS)\in\actions}\sum_{\state<\prime> \in \states} \upperBound<\star>(\state<\prime>) \cdot  \transitions(\state, \actionR, \actionS)(\state<\prime>) \cdot \strategy(\actionR) \cdot \strategy*(\actionS) \tag{Unfolding definition of $\preop$}\\
			&= \adjustlimits\sup_{\strategy \in \strategies(\state)}\inf_{\strategy* \in \strategies*(\state)}\sum_{(\actionR,\actionS)\in\actions}\sum_{\state<\prime> \in \ecStates} \upperBound<\star>(\state<\prime>) \cdot  \transitions(\state, \actionR, \actionS)(\state<\prime>) \cdot \strategy(\actionR) \cdot \strategy*(\actionS) \\
			&+ \sum_{\state<\prime\prime> \in \states \setminus \ecStates} \upperBound<\star>(\state<\prime\prime>) \cdot  \transitions(\state, \actionR, \actionS)(\state<\prime\prime>) \cdot \strategy(\actionR) \cdot \strategy*(\actionS) \tag{Player $\safe$ prefers leaving $\ecStates<\prime>$ and thus $\ecStates$}\\
			&= \adjustlimits\sup_{\strategy \in \strategies(\state)}\inf_{\strategy* \in \strategies*(\state)}\sum_{(\actionR,\actionS)\in\actions}\sum_{\state<\prime> \in \ecStates} \upperBound<\star>(\state<\prime>) \cdot  \transitions(\state, \actionR, \actionS)(\state<\prime>) \cdot \strategy(\actionR) \cdot \strategy*(\actionS) \\
			&+ \sum_{\state<\prime\prime> \in \states \setminus \ecStates} \left(\underbrace{\Delta(\state<\prime\prime>)}_{< \Delta^{\max}}+\valR(\state<\prime\prime>)\right) \cdot  \transitions(\state, \actionR, \actionS)(\state<\prime\prime>) \cdot \strategy(\actionR) \cdot \strategy*(\actionS)\tag{By def. of $\Delta$}\\
			&< \adjustlimits\sup_{\strategy \in \strategies(\state)}\inf_{\strategy* \in \strategies*(\state)}\sum_{(\actionR,\actionS)\in\actions}\sum_{\state<\prime> \in \ecStates} \upperBound<\star>(\state<\prime>) \cdot  \transitions(\state, \actionR, \actionS)(\state<\prime>) \cdot \strategy(\actionR) \cdot \strategy*(\actionS) \\
			&+ \sum_{\state<\prime\prime> \in \states \setminus \ecStates} \left( \Delta^{\max}+\valR(\state<\prime\prime>)\right) \cdot  \transitions(\state, \actionR, \actionS)(\state<\prime\prime>) \cdot \strategy(\actionR) \cdot \strategy*(\actionS) \tag{$\Delta$ for states outside $\ecStates$ is strictly smaller than $\Delta^{\max}$}\\
			&= \upperBound<\star>(\state) \tag{Everything sums-up to 1}
		\end{align*}
		Thus, we get a contradiction to the assumption that $\upperBound<\star>$ is a fixpoint.
	\end{description}

\end{proof}

\section{Counter-examples for previous works}\newtarget{sec:mistakes}
\subsection{Mistake in \cite{Chatterjee2012}}\label{counter-example:chaterjee}

\begin{figure}[tb]
	\centering
	\resizebox{0.7\textwidth}{!}{\begin{tikzpicture}[
	state/.style={draw,circle, minimum size = 0.5cm, align=center},
	stochSt/.style={draw,circle, fill=black, minimum size = 0.1cm, inner sep=0pt},
	transition/.style={->, black, -{Stealth[length=2mm]}, align=center},
	exit/.style={->, dashed, black, -{Stealth[length=2mm]}, align=center},
	invalid/.style={draw,rectangle,fill=black!30,minimum width=2cm,minimum height=1cm},
	every label/.append style = {font=\small}
	]
	
	\node (start) at (5,-1) {};
	\node[state] (S4) at (0,0) {\scalebox{1.3}{$s_4$}};
	\node[state,right of=S4, xshift=2cm] (S3)  {\scalebox{1.3}{$s_3$}};
	\node[state,left of=S4, xshift=-2cm] (S5)  {\scalebox{1.3}{$s_5$}};
	\node[state,right of=S3, xshift=2cm] (S0)  {\scalebox{1.3}{$s_0$}};
	\node[state, fill=lipicsYellow, right of=S0, xshift=2cm,yshift=-1.5cm] (S1)  {\scalebox{1.3}{$s_1$}};
	\node[state,right of=S0, xshift=2cm,yshift=1.5cm] (S2)  {\scalebox{1.3}{$s_2$}};
	\node[stochSt, right of=S0, xshift=0.5cm, yshift=0cm] (stoch) {};

	\draw[transition] (start) to (S0);
	\draw[transition] (S4) edge[bend left=20] node[above,align=center, yshift=0.1cm]{\scalebox{1.3}{$(\square, \mathsf{c})$}} (S3);
	\draw[transition] (S3) edge[bend left=20]  node[below,yshift=-0.1cm] {\scalebox{1.3}{$(\mathsf{b},\square)$}} (S4);
	\draw[transition] (S4) edge[bend left=0] node[above,align=center, yshift=0.1cm]{\scalebox{1.3}{$(\square, \mathsf{d})$}} (S5);
	\draw[transition] (S3) edge[bend left=0] node[above,align=center, yshift=0.1cm]{\scalebox{1.3}{$(\mathsf{a}, \square)$}} (S0);
	\draw[-] (S0) edge[bend left=0] node[above,align=center, yshift=0.1cm]{\scalebox{1.3}{$(\mathsf{a}, \mathsf{d})$}} (stoch);
	\draw[transition] (stoch) edge[bend left=50] (S0);
	\draw[transition] (stoch) edge[bend left=20] (S1);
	\draw[transition] (S0) edge[bend right=40] node[left,align=center, xshift=-0.3cm]{\scalebox{1.3}{$(\mathsf{a}, \mathsf{c})$, $(\mathsf{b},\mathsf{d})$}} (S1);
	\draw[transition] (S0) edge[bend left=40] node[left,align=center, xshift=-0.3cm]{\scalebox{1.3}{$(\mathsf{b}, \mathsf{c})$}} (S2);

	\draw[transition] (S2) edge[in=60,out=10,looseness=8] node[above]{\scalebox{1.3}{$(\square, \square)$}} (S2);
	\draw[transition] (S1) edge[in=60,out=10,looseness=8] node[above]{\scalebox{1.3}{$(\square, \square)$}} (S1);
	\draw[transition] (S5) edge[in=160,out=110,looseness=8] node[above]{\scalebox{1.3}{$(\square, \square)$}} (S5);
\end{tikzpicture}}
	\captionof{figure}{Counter example for the \BVI provided in \cite{Chatterjee2009}.}\label{apxfig:counter-example3}
\end{figure}
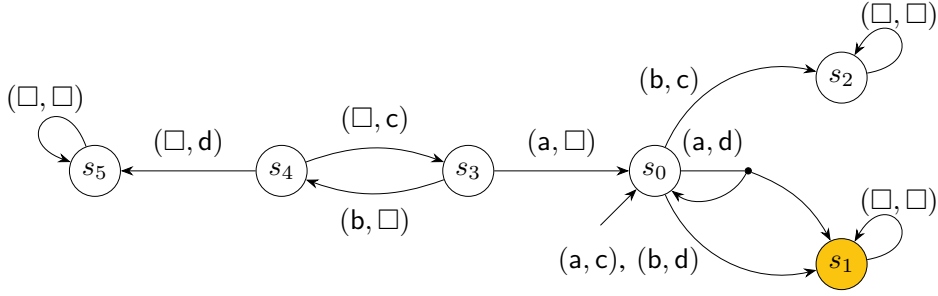

In \cite{Chatterjee2012}, the authors provide an exemplary \CSG that illustrates the incorrectness of the \BVI algorithm presented in \cite{Chatterjee2009}. \Cref{apxfig:counter-example3} shows the \CSG. The value attainable at $\state[5]$ is 0.6 while the value attainable at $\state[0]$ is $2-\sqrt{2}$. In our algorithm (\cref{alg:deflate_mecs}) the set $\{\state[4], \state[3]\}$ does not constitute a \BEC for the valuation $\upperBound(\state[4])=\upperBound(\state[3]) = \upperBound(\state[1])=1$ and $\upperBound(\state[2])=0$. Consequently, the while-loop of \cref{alg:deflate_mecs}, which is responsible for deflating, is not executed and thus the values of all states are updated only using the Bellman operator. 
This yields the correct values by \cref{theo:BVInoEC}. Therefore, \cref{alg:deflate_mecs} correctly sets the valuation of states $\state[4]$ and $\state[3]$ to the value 0.6. In contrast, the algorithm presented in \cite{Chatterjee2009} correctly sets the value of state $\state[4]$ to 0.6 but reduces the value of state $\state[3]$ to $2-\sqrt{2} < 0.6$.

\subsection{Mistake in \cite{eisentrautStoppingCriteriaValue2019}}\label{counter-example:eisentraut}

There was an attempt to fix the above problem in \cite{eisentrautStoppingCriteriaValue2019}, however, also this solution is incorrect. An exemplary \CSG that illustrates the incorrectness of the approach presented in \cite{eisentrautStoppingCriteriaValue2019} is the \CSG illustrated in \Cref{fig:non_trivial_ec}. While our approach correctly deflates the value of state $\state[0]$, the best exit as defined in \cite{eisentrautStoppingCriteriaValue2019} is falsely 1, i.e. the value of $\state[0]$ is never reduced.

\printglossary
\end{document}